\documentclass[11pt, dvipsnames]{article}
\usepackage[margin = 1in]{geometry}
\usepackage{ag}
\usepackage[authoryear]{natbib}

\usepackage{thmtools}
\usepackage[toc, page]{appendix}
\usepackage[capitalise,nameinlink]{cleveref}
\newtheorem{theorem}{Theorem}
\newtheorem{definition}{Definition}
\newtheorem{lemma}[theorem]{Lemma}
\newtheorem{assumption}{Assumption}
\crefname{assumption}{Assumption}{Assumptions}
\newtheorem{proposition}[theorem]{Proposition}

\newtheorem{proc}{Procedure}
\newtheorem{remark}{Remark}
\newtheorem*{remark*}{Remark}

\newtheorem{corollary}[theorem]{Corollary}
\usepackage{float}
\usepackage{placeins}

\newcommand{\shortmonthyear}{%
  \ifcase\month\or
  Jan\or Feb\or Mar\or Apr\or May\or Jun\or
  Jul\or Aug\or Sep\or Oct\or Nov\or Dec\fi
  \space \number\year
}

\title{Non-parametric Causal Inference in\\ Dynamic Thresholding Designs}
\author{Aditya Ghosh\\
\texttt{ghoshadi@stanford.edu} \and Stefan Wager\\
\texttt{swager@stanford.edu}}
\date{Stanford University}
\begin{document}
\allowdisplaybreaks

\maketitle
\begin{abstract}

We consider causal inference in dynamic settings where treatment is assigned by thresholding a state variable that can change over time. There is a large literature on regression-discontinuity methods building on the fact that, in the static setting, treatment assignment via threshold crossing induces a quasi-experimental design that enables pragmatic causal inference. But dynamic settings involve challenges not present in the static setting, e.g., past treatments may affect current state and thus future treatments, and so existing regression-discontinuity methods do not apply. Here, we show that dynamic thresholding designs identify a marginal policy effect that nests the classical regression-discontinuity parameter in the static setting; and propose a tailored local linear regression estimator that is consistent for this marginal policy effect. We demonstrate our approach using an experiment that emulates real-world optimization of thresholds for continuous glucose monitoring using data generated from an FDA-approved simulator.
\end{abstract}

\section{Introduction}

Dynamic threshold-based rules\blfootnote{\hspace{-6mm}Draft version \shortmonthyear.
We are grateful for helpful feedback and suggestions from seminar participants at Boston College, Boston University, Brown, Columbia, Duke, Michigan, MIT, Stanford, Uber, University of Chicago, University of Toronto, the Atlantic Causal Inference Conference, the Joint Statistical Meetings, and the Simons symposium on Bridging Prediction and Intervention Problems in Social Systems. A version of this paper has been accepted for presentation at the 2026 ACM Conference on Economics and Computation, and will appear as an abstract in the proceedings of the conference. This research was supported by the Office of Naval Research under grant number N00014-24-1-2091.}
govern many consequential decisions in healthcare, education, credit markets,
and public policy---and present numerous opportunities for policy evaluation. 
For example,~\citet{IIZUKA2021} study benefits of health signals using data from the Japanese healthcare system, where health signals are driven by threshold rules. Patients receive yearly health checkups at which their fasting blood sugar is measured, and are diagnosed as pre-diabetic and recommended lifestyle interventions along with follow-up care if their fasting blood sugar level crosses 110 mg/dL. \citet{IIZUKA2021} then ask whether such pre-diabetic health signals are efficient from a public health perspective relative to the cost of the induced follow-up care.

The goal of this paper is to develop methods for non-parametric policy evaluation in such dynamic thresholding designs.
The fact that thresholding designs open the door to non-parametric causal inference has been recognized
for a long time \citep{thistlethwaite1960regression}, and recent decades have seen a flurry of work on
regression discontinuity designs following this insight
\citep{Hahn-et-al-2001,IL2008review,CCT2014robustCI,armstrong2018optimal,noack2024bias}.
Existing work on regression discontinuity designs, however, are focused on cross-sectional settings
where each unit only receives treatment once and then experiences an outcome, and so are amenable
to causal analysis using the basic potential outcomes model \citep{imbens2015causal}.
In contrast, we are interested in settings where each unit is eligible for treatment multiple times
(in the example above, each patient gets a new pre-diabetes diagnosis each year)
thus resulting in complex treatment dynamics that need to be modeled in order to achieve correct inferences about the overall effect of a policy \citep{robins1986new}. 
For example, if prescribing lifestyle interventions is effective in lowering fasting blood sugar, then having a patient be diagnosed as prediabetic one year may make them less likely to receive the same diagnosis in subsequent years.

Here, we propose a framework for causal inference in dynamic thresholding designs that
non-parametrically accounts for dynamics as they arise in the model of \citet{robins1986new}.
Our first step is to re-interpret the classical causal target in regression-discontinuity
analyses as the marginal policy effect of modifying the treatment threshold \citep{carneiro2010evaluating}.
We find that such marginal policy effects remain identified in dynamic specifications; furthermore, we show that they can be consistently estimated via carefully tailored local linear regressions.
As such, we argue that marginal policy effects provide an informative yet statistically
tractable target for inference in dynamic thresholding designs.
Our formal approach draws from the literature on reinforcement learning and Markov decision processes
\citep{SuttonBarto2018}. Our analysis is in particular motivated by the policy gradient theorem
and the work of \citet{sutton1999policy} on first-order optimization of reinforcement learning models.

\subsection{Regression Discontinuities as Marginal Policy Effects}\label{sec:classical-RD-as-policy-gradient}

As background for our results on causal inference in dynamic threshold designs, we first
briefly review standard regression discontinuity {(RD)} designs (or, what could be called
cross-sectional thresholding designs) and how the resulting estimand can be interpreted
as a marginal policy effect. Following \citet{IL2008review}, assume that we have data on
IID-sampled pairs $(Z_i,\, Y_i)$ for units $i = 1, \, \ldots, \, n$,
where $Z_i \in \R$ is the running variable and
$Y_i \in \R$ is the outcome of interest.

The sharp {RD} design assumes that
there exists a cutoff $c \in \R$ such that treatment is assigned as $A_i = \ind{Z_i\ge c}$.
We posit potential outcomes $\{Y_i(0), \, Y_i(1)\}$ such that $Y_i = Y_i(A_i)$, and
define the conditional average treatment effect (CATE) as $\tau(z) = \E[Y_i(1) - Y_i(0) \, |\, Z_i = z]$.
Then, the sharp RD design enables us to estimate $\tau(c)$, i.e., the CATE at the cutoff,
as a discontinuity in the conditional response surface at $Z_i = c$ \citep{Hahn-et-al-2001},
\begin{equation}\label{eqn:tau-definition}
\tauRD := \tau(c) = \lim_{z\,\downarrow\, c} \, \mathbb{E}[Y_i \,|\, Z_i = z] - \lim_{z\,\uparrow\, c} \, \mathbb{E}[Y_i \,|\, Z_i = z],
\end{equation}
provided that the $\mu_{a}(z) = \E[Y_i(a) \, |\, Z_i = z]$ are continuous in $z$ and that the
running variable has continuous support around $c$.

The classical interpretation of the RD estimand as the CATE for units whose
running variable $Z_i$ straddles the cutoff relies crucially on $Z_i$ being causally
prior to any actions induced by our thresholding policy. In dynamic thresholding
designs, however, past thresholding actions can affect future values of the running
variable; and a direct approach that defines different causal estimands for every
treatment history suffers from a curse of dimension as time horizon gets large.
To avoid this issue, we find it helpful to re-interpret the classical RD estimand
as a marginal policy effect in the sense of \citet{carneiro2010evaluating}, i.e.,
as essentially the answer to a cost-benefit analysis. Given a threshold $c$, the sharp RD design with threshold introduces a treatment policy $\pi_c(Z_i)=\ind{Z_i\ge c}$ with associated policy value 
\begin{equation*}
%\label{eq:value_static}
V(\pi_c):=\E[Y_i(\pi_c(Z_i))]=\E_{\pi_c}[Y_i].
\end{equation*}
Then, provided the running variable is exogenous to the policy cutoff, $\tauRD$ can be
interpreted as the policy gradient of lowering the cutoff (i.e., of treating more units),
divided by the corresponding increase in the number of units treated as we lower the cutoff.
Further results on policy counterfactuals in thresholding designs are given in \citet{dong2015identifying}.

\begin{lemma}\label{lemma:classical-RD-as-policy-gradient}
Suppose that the distribution of $\{Y_i(0), \, Y_i(1), \, Z_i\}$ is exogenous to the chosen cutoff $c$. 
Suppose furthermore that the running variable $Z_i$ has density $f(\cdot)$ which is continuous and positive at $c$, and that the conditional response functions $\mu_{a}(z):=\E[Y(a)\mid Z = z]$ $(a=0,1)$ are continuous at $c$. Then,
    \begin{equation}
    \label{eq:mpe_static}
       \tauRD = \gradc V(\pi_c) \, \bigg/\, \gradc \E_{\pi_c}[A_i], \ \ \ \
       \gradc \E_{\pi_c}[A_i] = -f(c).
    \end{equation}
\end{lemma}

In other words, provided the running variables $Z_i$ are exogenous to the thresholding
policy, and if the cost of providing treatment to a unit is $\lambda$, then a social
planner could achieve cost-adjusted welfare benefits by reducing $c$ (and thus marginally
increasing the treatment rate) if and only if $\tauRD > \lambda$. As we move to a
multi-period setting, we will find this alternative characterization of the RD estimand
as the solution to a cost-benefit analysis to be remarkably resilient to challenges
induced by treatment dynamics.

\subsection{Modeling Dynamics}
\label{sec:model}

Now consider a setting where units are observed at times $t = 0, \, 1, \, 2, \, \ldots, \, T$,
where the horizon $T$ may be either finite or infinite\footnote{In the infinite-horizon setting, we use $T=\infty$ to define the population quantity of interest. However, the methods we propose for estimation and inference on this estimand are designed to operate in the realistic setting where the observed trajectories are finite (but sufficiently long for our asymptotic results to provide useful approximations).}.
At each time period $t$ we observe a running variable $Z_{i,t} \in \R \, \cup \, \{-\infty\}$,
take a thresholding action
$A_{i,t} = \ind{Z_{i,t}\ge c}$ for some threshold $c \in \R$, and observe an outcome $Y_{i,t} \in \R$. We allow for the case $Z_{i,t} = -\infty$ (so that $A_{i,t} = 0$ regardless of $c$) to account for the possibility
that there may be some time periods where action can never be taken; for example, in the case of yearly
health checkups, it's possible a patient misses their health checkup one year \citep{hsu2024dynamic}.

The causal structure of dynamic problems is considerably richer than in cross-sectional ones:
In addition to affecting outcomes $Y_{i,t}$, actions $A_{i,t}$ taken at time $t$ can affect state---and
thus also actions---at all times $t' > t$.
The induced potential outcomes then acquire a tree-like branching structure indexing over all possible
past treatment sequences \citep{robins1986new}. This branching makes a direct reduced-form approach to dynamic
thresholding designs intractable---or, at the very least, subject to an exponential blow-up in dimensionality
as the time horizon (and thus action space) grows. Instead, it is usually more fruitful to proceed
via what Robins refers to as the $g$-formula which provides a useful factorization for the observed-data
distribution under natural temporal consistency assumptions. Throughout, we will assume
that conditions required for the $g$-formula to hold are satisfied.

\begin{assumption}\label{assump:data-collected-under-thresholding-policy}
We observe data collected under a dynamic thresholding policy $\pi_c$ for some $c \in \R$, i.e.,
actions are taken according to $A_{i,t} = \ind{Z_{i,t}\ge c}$.
\end{assumption}

\begin{assumption}\label{assump:g-formula}
Under policy $\pi_c$, observation sequences for each unit $i = 1, \, \ldots, \, n$ are sampled IID
from a distribution $\mathbb{P}_{\pi_c}$ which factors according to the $g$-formula,
\begin{equation}
\label{eq:g-formula}
\mathbb{P}_{\pi_c}\left[Z_{i,0}, \, Y_{i,0}, \, \ldots, \, Z_{i,T}, \, Y_{i,T}\right]
= \prod_{t = 0}^{T} \mathbb{P}\left[Z_{i,t} \,\big|\, S_{i,t}\right] \mathbb{P}\left[Y_{i,t} \,\big|\, S_{i,t}, \,  Z_{i,t}, \, A_{i,t} = \ind{Z_{i,t}\ge c}\right],
\end{equation}
where $S_{i,t} = \{Z_{i,0}, \, A_{i,0}, \, Y_{i,0}, \, \ldots, \, Z_{i,t-1}, \, A_{i,t-1}, \, Y_{i,t-1}\}$ denotes observation history up
to time $t$ and $S_{i,0} = \emptyset$,
and we emphasize that all conditional probabilities on the right-hand side of the $g$-formula are
policy independent.
\end{assumption}

In the context of dynamic thresholding designs, the $g$-formula emerges
naturally from flexible structural models such as Markov decision processes \citep{SuttonBarto2018}.
For example, suppose there exists a (potentially high-dimensional and unobserved) state variable
$U_{i,t}$ such that it renders the system into a (partially observed) Markov decision process, i.e.,
\begin{equation}
\label{eq:pomdp}
\begin{split}
&\mathbb{P}_{\pi_c}\left[U_{i,0}, \, Z_{i,0}, \, Y_{i,0}, \, \ldots, \, U_{i,T}, \, Z_{i,T}, \, Y_{i,T}, \, U_{i,T+1}\right] \\
&\quad\quad\quad\quad=  \mathbb{P}\left[U_{i,0}\right] \prod_{t = 0}^{T} \mathbb{P}\left[Z_{i,t} \,\big|\, U_{i,t}\right] \mathbb{P}\left[Y_{i,t}, \, U_{i,t+1} \,\big|\, S_{i,t}, \,  Z_{i,t}, \, A_{i,t} = \ind{Z_{i,t}\ge c}\right],
\end{split}
\end{equation}
for any threshold $c \in \mathbb{R}$. Then, the $g$-formula \eqref{eq:g-formula} holds (\cref{prop:g-formula}).
The key insight in verifying this result is that, even though the state $U_{i,t}$ may
be unobserved, the actions themselves depend deterministically on the running variable $Z_{i,t}$
which is observed. Recent textbook discussions on related models are given in
\citet{HernanRobins2020} and \citet{wager2024causal}.

\begin{proposition}
\label{prop:g-formula}
Suppose there exists a variable $U_{i,t}$ such that \eqref{eq:pomdp} holds;
the variable $U_{i,t}$ may be observed or unobserved. Then, under
\cref{assump:data-collected-under-thresholding-policy}, the joint
distribution of the $Z_{i,t}$ and $Y_{i,t}$ across time satisfies the $g$-formula,
i.e., \cref{assump:g-formula} holds.
\end{proposition}

Our main question of interest is how the treatment---as determined by dynamic thresholding as in
\cref{assump:data-collected-under-thresholding-policy}---affects net-present expected welfare and treatment frequency,
\begin{equation}
\label{eq:value_dynamic}
V^Y(\pi_c):=\E_{\pi_c}\left[\sum_{t = 0}^T \gamma^t\, Y_{i,t} \right], \ \ \ \
V^A(\pi_c):=\E_{\pi_c}\left[\sum_{t = 0}^T \gamma^t\, A_{i,t} \right],
\end{equation}
where $0 < \gamma \leq 1$ is a discount rate (if $T = \infty$ then we must have $\gamma < 1$). The discount factor $\gamma$ determines the effective time horizon of the policy value by controlling how quickly the weights on future outcomes and treatments decay.
Because of the branching structure of potential outcomes a direct analogue to \eqref{eqn:tau-definition}
does not immediately enable meaningful program evaluation. However, perhaps surprisingly,
we will find that marginal policy effect characterizations of the form \eqref{eq:mpe_static} remain
useful: An RD estimand defined as
\begin{equation}
\label{eq:mpe_dynamic}
\tauRD := \frac{\partial}{\partial c} V^Y(\pi_c) \,\bigg/\, \frac{\partial}{\partial c} V^A(\pi_c)
\end{equation}
can still be effectively estimated in a sharp RD design---and can still be used to resolve
policy-relevant cost-benefit tradeoffs.\footnote{Concretely, one could use such estimates
to better allocate resources across different interventions (under global budget constraints)
by tuning the eligibility thresholds for each intervention. See \citet{sun2021treatment} for
further discussion of how such ratio quantities can be used to guide budget-constrained
treatment allocation.}

\subsection{Related Work}\label{sec:dynRDrelatedwork}

The modern literature on regression discontinuity designs goes back to \citet{Hahn-et-al-2001};
 influential contributions to this literature include \citet{IL2008review},
\citet{imbens2012optimal}, \citet{CCT2014robustCI} and \citet{armstrong2018optimal}.
Most of the existing methodological literature on regression discontinuity designs, however, is focused
on the cross-sectional setting where treatment is only assigned once. And, when faced with
the longitudinal setting, empirical researchers often reduce the problem to a cross-sectional
setting by simply considering various reduced-form regression discontinuities, e.g., by running
a standard RDD of $Y_t$ on $Z_t$ or of $Y_{t+1}$ on $Z_t$; this is, for example, the strategy
in the study by \citet{IIZUKA2021}. Such reduced form analyses
can be of considerable substantive interest in applications, but they do not capture
full treatment dynamics (e.g., how actions taken in one period may change the running
variable---and thus actions---taken in subsequent ones) and are thus not directly interpretable
as policy-relevant treatment effects \citep{heckman2016dynamic}.

One notable exception is \citet{cellini2010value}, who use regression discontinuities to
identify a type of treatment on the treated (ATT) effect in dynamic designs. They then use
their estimator to identify the effect of local school spending via public bonds on house
prices in California by comparing outcomes in school districts where bond measures are just
barely accepted vs.~barely rejected by voters; and their estimator allows them to formally consider
the fact that approving bonds in the past makes it less likely that additional bonds will
be approved in the future. The approach of \citet{cellini2010value}, however, makes crucial
use of a linear parametric model whereby
\begin{equation}
\label{eq:cellini}
    Y_{i,t} = \sum_{t' \leq t} \theta_{t - t'} A_{i,t'} + \varepsilon_{i,t},
\end{equation}
i.e., treatment effects are homogeneous and decay uniformly over time\footnote{When the parametric model of \citet{cellini2010value} is satisfied and the horizon is infinite $(T=\infty)$, our population target $\tauRD$ reduces to the discounted sum $\sum_{t=0}^\infty \gamma^t\,\theta_t$ where $\theta_t$ is defined in \eqref{eq:cellini}; see \cref{lemma:Celini} for details.}. And, as shown by
\citet{hsu2024dynamic}, their approach no longer recovers an ATT if we allow for treatment
heterogeneity. \citet{hsu2024dynamic} propose an alternative analysis that avoids \eqref{eq:cellini}.
But they in turn require a strong conditional mean independence assumption (CIA) which, e.g., in
the 2-period case requires that the time-2 control potential outcomes be independent of the time-2 running
variable for all units whose time-1 running variable
is near the cutoff.\footnote{See Assumption 3.1.2 of \citet[p.~1049]{hsu2024dynamic} for a precise
statement. This assumption is substantive, and would not hold in generic dynamic thresholding
designs. In particular, the CIA assumption will generally not hold if control potential outcomes
and the running variable both vary smoothly with some time-varying latent confounder; e.g., in our
motivating example, it would generally not hold if FBS
and health outcomes both vary smoothly with unobserved and time-varying health-seeking behaviors.}
To the best of our knowledge our paper is the first to provide results on non-parametric
causal inference for dynamic thresholding designs with generality that's comparable to standard
results in the cross-sectional setting following \citet{Hahn-et-al-2001}.

The dynamic causal model we use for reasoning about counterfactuals goes back to \citet{robins1986new}.
This model is widely used in biostatistics in the context of, e.g., marginal structural models \citep{robins2000marginal} and
optimal treatment regimes \citep{robins2004optimal}.
To the best of our knowledge, this model has not been
previously used in the context of dynamic thresholding designs---the one exception being \citet{hsu2024dynamic},
who pair the model of \citet{robins1986new} with their potentially restrictive CIA assumption to make progress.
Our work is also adjacent to the growing literature on micro-randomized trials (MRTs) \citep{klasnja2015microrandomized,liao2016sample,nahum2016just,boruvka2018assessing,dempsey2020stratified,qian2021estimating}. Like us, the MRT framework considers settings where treatment is offered at intermittent decision points that may depend on the individual's history, and the resulting causal estimands aggregate the downstream consequences of time-varying treatment decisions across future periods. The key distinction from our work is one of identification: MRTs exploit (state-conditional) randomization, whereas we utilize quasi-random variation in treatment assignment whenever the running variable is near the treatment threshold.

Our formal approach is motivated by results from the reinforcement learning (RL) literature \citep{SuttonBarto2018}, and
especially the policy-gradient theorem \citep{marbach2001simulation,sutton1999policy}. The policy-gradient theorem is widely used for optimizing RL systems via first-order algorithms \citep{kakade2001natural,sutton1999policy,williams1992simple}, and has recently been deployed for estimating global treatment effects in nonstationary Markovian A/B tests with temporal interference \citep{johari2025}.
General reinforcement learning concepts such as Q-functions and Bellman recursion are widely used in statistics, including for the study of dynamic treatment regimes \citep{ertefaie2018constructing,laber2014dynamic,luckett2020estimating,murphy2003optimal,schulte2015q} and for off-policy evaluation in Markov decision processes 
\citep{kallus2020double,liao2021off,liao2022batch,shi2022statistical}.
However, we are not aware of previous uses of policy-gradient
theorems for observational study causal inference in settings of the type we consider here.

\section{Characterizing the Marginal Policy Effect}

Our target estimand $\tauRD$ as defined in \eqref{eq:mpe_dynamic} involves counterfactual reasoning about how various moments would change as we change the treatment threshold $c$; however, for the purpose of estimation and inference, we only have access to data collected at a single status-quo threshold $c$. As such, our first task will be to provide an identification result $\tauRD$ in terms of moments of the observed data.

To this end, we start by providing an identification result for the policy gradient $\partial V^Y(\pi_c)/\partial c$ under the general dynamic model introduced in \cref{sec:model}. In the cross-sectional case, \cref{lemma:classical-RD-as-policy-gradient} shows that this gradient can be expressed in terms of conditional response functions and the density of the running variable at the cutoff. The dynamic setting introduces an additional complexity:~The relevant conditional response functions must now account for all future treatment dynamics induced by changing the treatment assignment for the current period.
To formalize this, we introduce the $Q$-function (or action-value function), which plays a central role in reinforcement learning and dynamic treatment regime analysis:
\begin{equation*}
%\label{eq:Qfn}
Q_{c,\,t}^R\left(s_t,\,z_t,\,a_t\right):=\E_{\pi_c}\left[\sum_{j=0}^{T-t} \gamma^j\,R_{i,t+j}\,\bigg|\, S_{i,t}=s_t, \, Z_{i,t}=z_t,\,A_{i,t}=a_t\right] \quad (R\in\{Y,A\}).
\end{equation*}
In words, $Q_{c,\,t}^Y(s_t,\,z_t,\,a_t)$ is the expected discounted sum of future outcomes and $Q_{c,\,t}^A(s_t,\,z_t,\,a_t)$ is the expected number of  future periods in which the unit is exposed to the treatment, starting with history $s_t$, running variable $z_t$ and action $a_t$ at time $t$. 

To characterize the policy gradient, all we need in addition to the basic model from \cref{sec:model}
is that the running variable have a density around the cutoff $c$ conditionally on past state, and that relevant
conditional-response functions vary smoothly with the running variable. We note that both assumption
will hold whenever there is non-trivial (continuously distributed and exogenous) noise in the running variable
\citep{lee2008randomized,Eckles2025}.

\begin{assumption}\label{assump:condtional-density}
    Conditional on the history $S_{i,t}$ and the event that $Z_{i,t}>-\infty$, $Z_{i,t}$ has a density $f_t^*(\,\cdot\mid S_{i,t})$ that is continuous at $c$ (almost surely). Define  $f_t(\,\cdot\mid S_{i,t}):=f_t^*(\,\cdot\mid S_{i,t})\,\P(Z_{i,t}>-\infty\mid S_{i,t})$, and assume that $$\E_{\pi_c}\left[\sum_{t=0}^T \gamma^t f_t(c\mid S_{i,t})\right]>0.$$
\end{assumption}

\begin{assumption}\label{assump:continuous-Q}
For each $t\ge 0$, the $Q$-functions $Q_{c,\,t}^R(s_t,\, z_t\,,\, a_t)$ $(a_t=0, 1,\ R=Y, A)$ are continuous at $z_t = c$ for almost every $s_t$ and for $a_t = 0, \, 1$.  Furthermore, assume that $$\E_{\pi_c}\left[\sum_{t=0}^T\gamma^t\, (Q_{c,\,t}^A(S_{i,t},\,c,\,1) - Q_{c,\,t}^A(S_{i,t},\,c,\,0)) f_t(c\mid S_{i,t})\right]>0.$$
\end{assumption}

The following result explicitly characterizes the policy gradients used in \eqref{eq:mpe_dynamic} to define the marginal policy effect for the dynamic thresholding design. Although it is conceptually similar to the standard policy-gradient theorem \citep{sutton1999policy}, we note that it is not a direct corollary:~The standard policy gradient quantifies the effect of changing action probabilities under overlap conditions (i.e., where treatment and control actions can both occur with positive probability in all states), whereas here we consider the effect of changing the treatment cutoff in a setting without overlap.

\begin{theorem}\label{thm:expression-for-policy-gradient}
   Suppose that \cref{assump:data-collected-under-thresholding-policy,assump:g-formula,assump:condtional-density,assump:continuous-Q} hold, and assume furthermore that the
   following integrability conditions hold: For some $\eta>0$,
   \begin{equation}
       \label{eq:integrability}
    \begin{split}
        &\sup_{t\ge 0}\,\sup_{s_t}\,\sup_{|c'-c|\le\eta}\,\E_{\pi_{c'}}\left[\sum_{j=0}^{T-t}\gamma^j\,|Y_{i,t+j}|\,\bigg|\, S_{i,t}=s_t\right]<\infty, \\
        &\sup_{t\ge 0}\,\sup_{s_t}\,\sup_{|z-c|\le\eta} \max\left\{\left|Q_{c,\,t}^Y(s_t,\,z,\,1)-Q_{c,\,t}^Y(s_t,\,z,\,0)\right|, 1\right\}f_t(z\mid s_t)<\infty.
    \end{split}
   \end{equation}
   Then, the gradient of the total discounted reward under
   the threshold-based policy  $\pi_c$ with respect to the threshold parameter $c$ is given by
\begin{equation*}
%\label{eq:pgt-expr}
-\gradc V^Y(\pi_c)=\sum_{t=0}^T \gamma^t\,  \E_{\pi_c}\left[\left(Q_{c,\,t}^Y(S_{i,t},\,c,\,1)-Q_{c,\,t}^Y(S_{i,t},\,c,\,0)\right)f_t(c\mid S_{i,t})\right],
\end{equation*}
provided that $\gamma < 1$ if $T =\infty$. Moreover, an analogous expression holds for $-\partial V^A(\pi_c)/\partial c$, and thus the dynamic marginal policy effect parameter $\tauRD$ introduced in \eqref{eq:mpe_dynamic} can be equivalently expressed as 
\begin{equation}\label{eq:tauRD}
\tauRD=\frac{\sum_{t=0}^T \gamma^t\, \E_{\pi_c}\left[(Q_{c,\,t}^Y(S_{i,t},\, c,\, 1)-Q_{c,\,t}^Y(S_{i,t},\, c,\, 0))f_t(c\mid S_{i,t})\right]}{\sum_{t=0}^T \gamma^t\, \E_{\pi_c}\left[(Q^A_{c,\,t}(S_{i,t},\, c,\, 1)-Q^A_{c,\,t}(S_{i,t},\, c,\, 0))f_t(c\mid S_{i,t})\right]}.
\end{equation}
\end{theorem}

The above result establishes that each of the policy gradients $\partial V^Y(\pi_c)/\partial c$ and $\partial V^A(\pi_c)/\partial c$ equals a discounted sum of the $Q$-function differences at the threshold across all future time periods, weighted by the conditional density at the threshold in each period, leading to the characterization \eqref{eq:tauRD} for the dynamic marginal policy effect parameter $\tauRD$. Although this characterization may not immediately point to an intuitive estimation result, we emphasize that all right-hand-side quantities in \eqref{eq:tauRD}, i.e., $Q_{c,\,t}^Y(s,\, c,\, w)$, $Q_{c,\,t}^A(s,\, c,\, w)$, and $f_t(c\mid s)$, are conditional moments of the data-collection distribution---and thus identified under continuity assumptions.

\section{Estimation and Inference via Local Linear Regression}

The expression for the causal parameter $\tauRD$ given in \eqref{eq:tauRD}
is explicit---but at first glance may appear unwieldy to operationalize because of
its dependence on the growing-dimensional state $S_{i,t}$. Perhaps
surprisingly, however, it turns out that the quantity $\tauRD$ can be
estimated using a carefully designed local linear regression procedure;
this section details how.

Consider the dynamic thresholding design introduced in \cref{sec:model}.
While \eqref{eq:tauRD} provides a unified characterization of the dynamic marginal policy effect across both finite and infinite horizons, the infinite-horizon case demands separate attention to the fact that the observed trajectories in any real-world data are necessarily finite. To accommodate this, we denote by $T_n$ the horizon for the observed data\footnote{The methods we develop can also accommodate cases where the units have trajectories of varying lengths, provided all trajectories are sufficiently long for asymptotic approximations to hold. However, to keep the exposition simple, we use a common trajectory length $T_n$ to state our results.}  (indexed by the sample size $n$), where $T_n=T$ in the finite-horizon case, and $T_n\to\infty$ as $n\to\infty$ in the infinite-horizon case.

Define the discounted sum of future outcomes and treatment assignments from time $t$ onward: 
\begin{equation}\label{eqn:def-Gamma}
\Gamma_{i,t}^Y :=\sum_{s=0}^{T_n-t} \gamma^s \,Y_{i,\,t+s}, \ \ \ \
\Gamma_{i,t}^A :=\sum_{s=0}^{T_n-t} \gamma^s \,A_{i,\,t+s}
\end{equation}
where $t=0,1,\dots,T_n$, and $i=1,2,\dots,n$.
Motivated by the twice-discounted structure of the expression in \eqref{eq:tauRD}, where discounting appears both within the $Q$-functions and in the outer summation, we propose below a twice-discounted local linear regression procedure for estimating the causal parameter $\tauRD$.
We choose a small bandwidth $h=h_n$ (where $h_n\to 0$ as $n\to\infty$), a weighting
function $K:\R \to [0,\infty)$ and run a weighted linear regression on each side
of the threshold, as follows.
\begin{equation}
    \label{eq:def-LLR}
    \begin{split}
&\wh{\tau}^R(h):=e_1^\top\argmin_{(\tau,\,\alpha,\,\beta_0,\,\beta_1)}\frac{1}{n}\sum_{i=1}^n\sum_{t=0}^{T_n}w_{i,t}(h)\left(\Gamma_{i,t}^R -\tau A_{i,t}-\alpha_t-\beta_0(Z_{i,t}-c)-\beta_1 A_{i,t}(Z_{i,t}-c)\right)^2, \\[2mm]
&\text{where } R\in\{Y,A\},\ \ \text{and}\ \ w_{i,t}(h) := \gamma^t \,K \left(|Z_{i,t}-c|/h\right). \  \text{Define}\ \
    \htauRD:=\wh\tau^Y(h_n) \,\big/\,  \wh{\tau}^A(h_n).
    \end{split}
\end{equation}
Here $e_1=(1,0,\dots,0)^\top$ selects $\tau$ from the entire vector of parameters.
 Popular choices for the weighting function $K(\cdot)$ include the window function $K(z)=\ind{|z|\le 1}$ or the triangular kernel $K(z)=(1-|z|)_+$.

It is important to note that the estimator $\htauRD$ proposed above has a ratio
form that's typically associated with fuzzy RD designs \citep{IL2008review}, even though this paper is focused on `sharp' thresholding designs, i.e., where treatment
assignment is a deterministic function of threshold crossing. This ratio form is explained by the fact that, in the dynamic design, there is
some uncertainty on how moving $c$ will affect the net-present treatment frequency%(since we may not get to observe realizations of $Z_{i,t}$ for $t \geq 1$ for counterfactual thresholding policies)
; and, as argued in \citet{sun2021treatment}, cost-benefit
analyses under cost uncertainty induce statistical structure resembling that encountered
in instrumental-variable analyses. When $K(\cdot)$ is the window function, the above can also be numerically implemented as a two-stage least squares type estimator with ``instrument'' $A_{i,t}$ and ``treatment'' $\Gamma_{i,t}^A$ \citep{IL2008review}. 

It is worth clarifying how the local linear regressions in \eqref{eq:def-LLR} modify the standard local linear regression used in static / cross-sectional RD designs. First, we use the discounted sum of future outcomes $\Gamma_{i,t}^Y$ (resp.~$\Gamma_{i,t}^A$) instead of the immediate outcome $Y_{i,t}$ (resp.~$A_{i,t}$) to approximate the $Q$-functions, thus automatically incorporating long-term downstream effects of treatment decisions. Second, in addition to the local kernel weights, we use the temporal weighting  $\gamma^t$ that mirrors our dynamic policy gradient result (cf.~\cref{thm:expression-for-policy-gradient}). 
Third, the regression \eqref{eq:def-LLR} includes time fixed-effects. This is an algorithmic choice to reduce asymptotic variance, rather than a modeling assumption. The time fixed effects absorb time-specific shifts common across units, removing variation orthogonal to the local comparison at the cutoff.
Our next result illustrates that with these simple tweaks to the standard local linear regression, we can consistently estimate the causal parameter $\tauRD$. Before stating this result, we list some regularity assumptions.

\begin{assumption}\label{assump5:regularity-for-consistency}
     \hspace{1mm} 
\begin{enumerate}[label=(\roman*)]
    \item The kernel $K:\R \to [0,\infty)$ is continuous, supported on $[-1,1]$, and not identically zero.

        \item The conditional second moments $z\mapsto m_{2,\,t}^R(s,\,z,\,a):=\E_{\pi_c}[(\Gamma_{i,t}^R)^2\mid S_{i,t}=s,\, Z_{i,t}=z,\, A_{i,t}=a]$ $(R\in\{Y,A\},\ a=0,1)$ are continuous at $c$ for a.e.~$s$, for every $t\ge 0$. 
        \item The densities $f_t$ and conditional second-moments $m_{2,\,t}^Y$ are locally bounded: For some $\eta>0$, there exist measurable envelopes $B_{f,\,t}$ and $B_{m_2,\,t}$ such that for a.e.~$s$, and for $a=0,1$, 
        $$\sup_{|z-c|\le \eta}f_t(z\mid s)\vee 1\le B_{f,\,t}(s),\quad \sup_{|z-c|\le \eta} m_{2,\,t}^Y(s,\, z,\, a)\vee 1\le B_{m_2,\,t}(s).$$
Moreover, assume that $\E_{\pi_c}\left[\sum_{t= 0}^T\gamma^t\, B_{m_2,\,t}(S_{i,t})\, B_{f,\,t}(S_{i,t})\right]<\infty$.
 
    \item  The unconditional second-moments are uniformly bounded: $\sup_{t\ge 0}\E_{\pi_c}[Y_t^2]<\infty$.
\end{enumerate}

    \end{assumption}
\begin{theorem}\label{consistency}
    Suppose that 
    \cref{assump:data-collected-under-thresholding-policy,assump:g-formula,assump:condtional-density,assump:continuous-Q,assump5:regularity-for-consistency} hold true, and that we run the local linear regression \eqref{eq:def-LLR} with bandwidth $h=h_n$ that satisfies $h_n\to 0$ and $n h_n\to \infty$. If $T=\infty$, assume further that $\gamma<1$ and that the observed horizon is long enough in the sense that $\gamma^{T_n}\le h_n^{1+\eps}$ %$T_n\ge (1+\eps)\log h_n^{-1}/\log \gamma^{-1}$ 
    for all large $n$, for some $\eps>0$. Then  $\htauRD$ converges in probability to the causal parameter  $\tauRD$ defined in \eqref{eq:mpe_dynamic}.
\end{theorem}

Our next goal is to show that our estimator achieves the standard nonparametric rate of $n^{-2/5}$
for local linear regression based RD estimators under appropriate smoothness conditions.
Continuing the parallel between dynamic thresholding designs and classical (cross-sectional) RD designs, we impose the following regularity conditions that are natural extensions of second-order smoothness assumptions standard in classical RD literature (see, e.g., \citet{Hahn-et-al-2001}).

\begin{assumption}\label{assump6:smoothness}
The functions   $f_t(\,\cdot\mid s)$,  $Q_{c,\,t}^Y(s,\, \cdot\,,\, a)$ and $Q_{c,\,t}^A(s,\, \cdot\,,\, a)$ $(a=0,1)$ have second derivatives in $[c-\eta,c+\eta]$, and there exist measurable envelopes $B_{f'',\,t}(s)$ and $B_{Q'',\,t}(s)$ such that
\begin{equation*}
    \begin{split}
        &\sup_{|z-c|\le \eta}\max\left\{\left|\frac{\partial^2}{\partial z^2} \,f_t(z\mid s)\right|, 1\right\}\le B_{f'',\,t}(s),\\ 
        &\sup_{|z-c|\le \eta}\max\left\{ \left|\frac{\partial^2}{\partial z^2} \,Q_{c,\,t}^Y(s,\, z,\, a)\right|, \left|\frac{\partial^2}{\partial z^2} \,Q_{c,\,t}^A(s,\, z,\, a)\right|, 1\right\}\le B_{Q'',\,t}(s),\\
       & \sum_{t=0}^T\gamma^t\, \E_{\pi_c}\left[B_{Q'',\,t}(S_t)\,B_{f'',\,t}(S_t)\right]<\infty,\\
        &\sum_{t=0}^T\sum_{t'=t+1}^{T}\gamma^{t+t'}\, \E_{\pi_c}\left[B_{m_2,\,t}^{1/2}(S_t)\,B_{m_2,\,t'}^{1/2}(S_{t'})\,B_{f,\,t}(S_t)\,B_{f,\,t'}(S_{t'})\right]<\infty,
    \end{split}
\end{equation*}
where the envelopes $B_{m_2,\,t}$ and $B_{f,\,t}$ are as defined in \cref{assump5:regularity-for-consistency}. Also assume that the conditional cross-moments $\E_{\pi_c}[\Gamma_{i,t}^Y\, \Gamma_{i,t}^A\mid S_{i,t}=s,\, Z_{i,t}=z,\, A_{i,t}=a]$ $(a=0,1)$ are continuous in $z$ at $c$ for a.e.~$s$, for every $t\ge 0$.
\end{assumption}

The following result establishes the limiting distribution of the local linear regression estimator proposed in \eqref{eq:def-LLR} with an explicit characterization of the asymptotic variance.

\begin{theorem}\label{clt-for-twice-discounted-llr}
     Suppose that \cref{assump:data-collected-under-thresholding-policy,assump:g-formula,assump:condtional-density,assump:continuous-Q,assump5:regularity-for-consistency,assump6:smoothness} hold true, and that the local linear regression \eqref{eq:def-LLR} is run with bandwidth $h_n=O(n^{-1/5})$. If $T=\infty$, assume further that $\gamma<1$ and that the observed horizon satisfies %$\gamma^{T_n}\le n^{-(\frac{3}{5}+\eps)}$
     $T_n\ge (3/5+\eps)(\log n)/(\log \gamma^{-1})$ 
     for all large $n$, for some $\eps>0$. Then the asymptotic distribution of the local linear regression estimator $\htauRD$ is given by 
\begin{equation*}%\label{eq:asymp-dist-of-ratio-finite-horizon}
    \sqrt{nh_n}\left(\htauRD -\tauRD-\frac{1}{2}h_n^2\xi_1\frac{\Delta\mu''_{Y}(c)-\tauRD \Delta\mu_A''(c)}{\Delta\mu_A(c)}\right)\dto \normal\left(0, V_{\,\mathrm{RD}} \right),
\end{equation*}
where $\xi_1:=(\kappa_2^2-\kappa_1\kappa_3)/(\kappa_0\kappa_2-\kappa_1^2)$,  $\kappa_j:=\int_0^1 u^j K(u)du$, and for $R\in\{Y,A\}$,  $\Delta\mu_{R}:=\mu_{R,1}-\mu_{R,0}$, 
\begin{equation}
    \label{def-muRz}
    \mu_{R,a}(z):=\frac{\sum_{t=0}^T \gamma^t\,\E_{\pi_c}\left[Q_{c,t}^R(S_t,\,z,\,a)\,f_t(z\mid S_t)\right]}{\sum_{t=0}^T \gamma^t\,\E_{\pi_c}\left[f_t(z\mid S_t)\right]}, \quad a=0,1.
\end{equation}
The asymptotic variance is given by \begin{equation}\label{eq:asymp-var}
    V_{\,\mathrm{RD}}:=\frac{V_{YY} + \tauRD^2 V_{AA} - 2\tauRD V_{YA}}{(\Delta\mu_A(c))^2},
\end{equation} where, for any $G,H\in\{Y,A\}$, 
\begin{multline*}
    V_{GH}:=\frac{\xi_2}{F(c)^2}\sum_{t=0}^T\sum_{a=0}^1  \gamma^{2 t}\, \E_{\pi_c}\bigg[\E_{\pi_c}\bigg[\left(\Gamma^G_t - \alpha_{t}^G-a\tau^G\right)\\
    \times\left(\Gamma^H_t - \alpha_{t}^H-a\tau^H\right)\,\Big|\, S_t,\,Z_t=c,\,A_t=a\bigg] f_t(c \mid S_t)\bigg],
\end{multline*}
with  $F(c):=\sum_{t=0}^T\gamma^t\, \E_{\pi_c}\left[f_t(c\mid S_t)\right]$, $\xi_2:=(\kappa_2^2\rho_0 - 2\kappa_1\kappa_2\rho_1+\kappa_1^2\rho_2)/(\kappa_0\kappa_2-\kappa_1^2)^2$, $ \rho_j:=\int_0^1 u^j\,K^2(u)\, du$, 
$\tau^R := \Delta\mu_R(c)$ 
($R=Y,A$),
and $\alpha_t^R$ are the population-level fixed-effects, defined as:
\begin{equation}
    \label{def-alpha-t}
    \alpha_t^R := \frac{1}{2}(m_{t,1}^R(c) + m_{t,0}^R(c)-\Delta\mu_R(c)),\quad  m_{t,a}^R(c) :=\frac{\E_{\pi_c}[Q_{c,t}^R(S_t,\,c,\,a)\,f_t(c\mid S_t)]}{\E_{\pi_c}[f_t(c\mid S_t)]},
\end{equation}
with the convention that $m_{t,a}^R(c)=0$ when $\E_{\pi_c}[f_t(c\mid S_t)]=0$.
\end{theorem}

\cref{clt-for-twice-discounted-llr} shows that when the bandwidth scales as $h_n=O(n^{-1/5})$, our estimator achieves the same $n^{-2/5}$ rate of convergence and exhibits the same bias structure as classical local linear regression estimators in static RD designs. The leading bias term has the familiar $h_n^2$ form, driven by the curvature of the weighted conditional response functions $\mu_{G,\,a}(z)$ at the threshold. The asymptotic variance in \cref{clt-for-twice-discounted-llr} reflects the additional complexity of the dynamic setting: It aggregates uncertainty across all time periods, with contributions weighted by $\gamma^{2t}$ to account for temporal discounting. 
It is also interesting to note that when $T=0$, \cref{clt-for-twice-discounted-llr} reduces to the analogous asymptotic result for the classical RD setting (see, e.g., \citet[Theorem 4]{Hahn-et-al-2001}).

\begin{remark}[Dependence of the asymptotic variance on the discount factor]
Under suitable regularity conditions, the asymptotic variance $V_{\mathrm{RD}}$ in \cref{clt-for-twice-discounted-llr} satisfies $V_{\mathrm{RD}} = O(\cgg{T}),$
where \begin{equation*}
        \cgg{T}:=\frac{\sum_{t=0}^T\gamma^{2t}\cg{T-t}^2}{\cg{T}^2}, \quad \cg{T}:=\sum_{t=0}^T\gamma^t,
    \end{equation*}
see \cref{lemma:VRD-on-gamma} for a precise statement. The behavior of the constant $\cgg{T}$ is governed by the product $T(1-\gamma)$: It scales as $T/3$ when $T(1-\gamma)\to 0$ (short effective horizon) and as 
$(1-\gamma^2)^{-1}$ when $T(1-\gamma)\to\infty$ (long effective horizon); see \cref{lemma:C2rates} for a complete characterization.
\end{remark}

Next we propose a consistent estimator of the asymptotic variance.
The following result shows that the unit-clustered sandwich variance estimator for the coefficient of \(A_{i,t}\) in \eqref{eq:def-LLR} consistently estimates the asymptotic variance of $\htauRD$.

\begin{proposition}[Consistency of unit-clustered  standard errors]
\label{propo:sandwich-estimator}
For $G, H\in\{Y,A\}$, define  
\begin{align*}
    \wh V_{GH,n} &:=
\frac1n\sum_{i=1}^n\wh\psi_{n,i}^G\wh\psi_{n,i}^H,\\
\wh\psi_{n,i}^R&:=\sqrt{h_n}\,e_1^\top\left(\frac1n\sum_{i=1}^n\sum_{t=0}^{T_n}
\gamma^t \,K\bigg(\frac{|Z_{i,t}-c|}{h_n}\bigg)\,
r(Z_{i,t})^{\otimes2}\right)^{-1}
\sum_{t=0}^{T_n}
\gamma^t \, K\bigg(\frac{|Z_{i,t}-c|}{h_n}\bigg)\,
r(Z_{i,t})\,\wh \eps_{i,t}^R,
\end{align*}
where $\wh\eps_{i,t}^R$ are the residuals from the local linear regression \eqref{eq:def-LLR}, and
\[
r(Z_{i,t}):=\left(
A_{i,t},\,
Z_{i,t}-c,\,
A_{i,t}(Z_{i,t}-c)
\right)^\top - 
\frac{\sum_{i=1}^n K(|Z_{i,t}-c|/h_n)\,\left(
A_{i,t},\,
Z_{i,t}-c,\,
A_{i,t}(Z_{i,t}-c)
\right)^\top}
     {\sum_{i=1}^n K(|Z_{i,t}-c|/h_n)},
\]
where the ratio is set to zero when the denominator is zero. 
Then, under the conditions of \cref{clt-for-twice-discounted-llr},
\[
\wh V_{{\rm RD},n}:=
\frac{
\wh V_{YY,n}
+\htauRD^2\wh V_{AA,n}
-2\,\htauRD\wh V_{YA,n}
}
{(\wh\tau^A(h_n))^2}
\,\,\Pto\,\, V_{\rm RD}.
\]
\end{proposition}

\paragraph{Bandwidth choice.}
As in standard (cross-sectional) RD designs, the bandwidth $h$ used in the local linear regression \eqref{eq:def-LLR} determines the bias--variance tradeoff of the proposed estimator $\htauRD$. The asymptotic expansion in \cref{clt-for-twice-discounted-llr} implies that minimizing the asymptotic mean squared error (AMSE) with respect to the bandwidth $h$ results in the following (infeasible) optimal bandwidth:
\begin{equation}\label{eq:dynamic-amse}
\begin{split}
   h_{\,\mathrm{IK},\mathrm{dyn}}^*&=\argmin_{h} \left\{\frac{V_{YY}+\tauRD^2V_{AA}-2\tauRD V_{YA}}{n h}
    +\frac{h^4}{4}\xi_1^2\left(\Delta\mu_Y''(c)-\tauRD\Delta\mu_A''(c)\right)^2 \right\}\\
    &= \left(\frac{V_{YY}+\tauRD^2V_{AA}-2\tauRD V_{YA}}{\xi_1^2\left(\Delta\mu_Y''(c)-\tauRD\Delta\mu_A''(c)\right)^2}\right)n^{-1/5}.
\end{split}
\end{equation} Following \citet{imbens2012optimal}, we also derive a data-driven bandwidth $\wh{h}_{\,\mathrm{IK},\mathrm{dyn}}$ which is consistent for the AMSE-optimal bandwidth as above, in the sense that $\wh{h}_{\,\mathrm{IK},\mathrm{dyn}}/h_{\,\mathrm{IK},\mathrm{dyn}}^*\,\Pto\, 1$.
 Since the proposed estimator $\htauRD$ is a ratio of local linear estimators, we follow the procedure for the fuzzy-RD case in \citet{imbens2012optimal} and choose a single bandwidth (instead of two separate bandwidths for the numerator and the denominator in \eqref{eq:def-LLR}) by applying the sharp-RD calculation to the residualized outcomes $\Gamma^Y-\tauRD\Gamma^A$.
The main strategy is to pick a pilot bandwidth $h_{\mathrm{pilot}}$, calculate the residualized outcomes as $\Gamma^Y-\wh\tau_{\,\mathrm{pilot}}\Gamma^A$ where $\wh\tau_{\,\mathrm{pilot}}=\htauRD(h_{\mathrm{pilot}})$, and estimate the quantities in \eqref{eq:dynamic-amse} exactly as in \citet{imbens2012optimal}. 

The dynamic setting affects the estimation of the bias and variance quantities in \eqref{eq:dynamic-amse} in the following ways. First, the variance must allow for arbitrary dependence across time within a trajectory; we achieve this using the cluster-by-unit procedure as in \cref{propo:sandwich-estimator}. Second, the bias is governed by the curvature of the weighted dynamic response functions $\mu_{R,a}(\cdot)$ as defined in \eqref{def-muRz}, estimated using a local quadratic regression with the same temporal weights and time fixed effects as in \eqref{eq:def-LLR}. The time fixed effects do not change the leading bias, but they change the influence function and hence the variance entering the bandwidth rule. We provide implementation details and show consistency of this data-driven bandwidth in \ref{app:bandwidth-selection}, and demonstrate its finite-sample performance in our numerical experiments in \cref{sec:numerical-experiments}.

\paragraph{Inference.} A well-known challenge in classical RD inference is that the local linear regression run with a bandwidth scaling as $h_n\sim n^{-1/5}$ (the same scale as the AMSE-optimal bandwidth) yields an estimator with the bias and standard error of the same order. As a result, conventional confidence intervals generally do not achieve nominal coverage. One way to resolve this is to `undersmooth':~Use a bandwidth $h_n \ll n^{-1/5}$ so that the variance dominates the bias; we formally state this for the dynamic setting in \cref{coro:undersmooth} below. 
\begin{corollary}\label{coro:undersmooth}
  Suppose that \cref{assump:data-collected-under-thresholding-policy,assump:g-formula,assump:condtional-density,assump:continuous-Q,assump5:regularity-for-consistency,assump6:smoothness} hold, and that the local linear regression \eqref{eq:def-LLR} is run with bandwidth $h_n=o(n^{-1/5})$. If $T=\infty$, assume further that the trajectory length $T_n$ satisfies $T_n\ge (3+\eps)\,(\log h_n^{-1})/(\log \gamma^{-1})$ for all large $n$, for some $\eps>0$. Then, with the variance estimator $\wh{V}_{\,\mathrm{RD},\,n}$ as defined in \cref{propo:sandwich-estimator}, it holds for any $\alpha\in (0,1)$ that
  $$\lim_{n\to\infty} \P_{\pi_c}\left(\tauRD\in\left[\htauRD\pm z_{1-\alpha/2}\,\wh{V}_{\,\mathrm{RD},\,n}^{1/2} (nh_n)^{-1/2}\right]\right)= 1-\alpha.$$
\end{corollary}

The above strategy, however, may not be preferable, as undersmoothing results in larger-than-optimal estimation error, and there is no clear guidance on how much to undersmooth for reliable finite-sample performance. 
More principled alternatives are bias-correction procedures that leverage higher-order smoothness to estimate and eliminate the bias \citep{CCT2014robustCI} and bias-aware procedures that widen intervals to accommodate worst-case bias \citep{armstrong2018optimal}. These methods are all tailored to the static RD setting ($T=0$), and we leave the task of adapting them to dynamic thresholding designs for future work.
We do, however, examine one bias-aware construction in our numerical experiments:~\citet{ArmstrongKolesar2020} show that the standard local linear regression with an MSE-optimal bandwidth can be paired with the critical value $2.18$, rather than the standard normal quantile $1.96$, to yield `honest' $95\%$ confidence intervals that are uniformly valid over a smoothness class of regression functions. Given that our estimator already rests on local linear regression with a data-driven bandwidth, implementing this amounts to a single numerical substitution; we include it as an ad hoc extension of the conventional intervals, while emphasizing that its validity has been established only for the static case.

\section{Numerical Experiments}\label{sec:numerical-experiments}

In this section, we examine the empirical performance of the proposed method. In \cref{sec:autoreg}, we use a stylized autoregressive data-generating process that makes the treatment-dependent state evolution explicit. In \cref{sec:gluc-sim}, we turn to a more calibrated simulation based on the UVa/Padova Type-1 Diabetes simulator \citep{UVA-PADOVA}, where the thresholding rule interacts with richer, physiologically motivated dynamics.

\subsection{An Autoregressive Simulator}\label{sec:autoreg}

Consider a simulation setting where we generate data from the following autoregressive process: 
\begin{equation}
    \begin{split}
        Z_{t+1} &= \delta + Z_t- (1-\rho)(Z_t -\mu_0) - \theta\, A_t\, (Z_t-\mu_0)_+ +4\eps_t,\\[1mm] A_t &= \ind{Z_t\ge c},\\[1mm] Y_t &=-Z_{t+1}, \quad t=0,1,\dots,T-1.
    \end{split}
\end{equation}
     We consider a finite horizon $T=12$, and generate the noise $\eps_t$ as i.i.d.~from $\normal(0,1)$. We use the treatment threshold $c=110$, the baseline mean $\mu_0=100$, the autocorrelation coefficient $\rho=0.9$, the treatment intensity parameter $\theta=0.1$, and consider two choices for the drift parameter $\delta$, namely $\delta=0$ (Setting 1) and $\delta=1$ (Setting 2). We initialize the autoregressive process with $Z_0\sim \normal(0, 4/(1-\rho^2)^{1/2})$. Note that this initialization does not make $(Z_t)_{t\ge 0}$ stationary even when $\delta=0$, because of the treatment-dependent drift $\theta\neq 0$. 
     
We contrast our proposed procedure \eqref{eq:def-LLR} with the following baseline methods, which can be viewed as natural strategies but are not tailored for inference on $\tau_{\mathrm{RD}}$.

\begin{enumerate}
    \item Baseline 1: Run the standard LLR of $Y_{i,t}$ on $Z_{i,t}$, i.e.,
    \begin{multline}\label{eq:static}
       \wh{\tau}_{\mathrm{LLR}}(h):=e_1^\top\argmin_{(\tau,\,\alpha,\,\beta_0,\,\beta_1)}\frac{1}{n}\sum_{i=1}^n\sum_{t=0}^{T_n}K\bigg(\frac{|Z_{i,t}-c|}{h}\bigg)\Big(Y_{i,t} -\tau A_{i,t}-\alpha_t\\-\beta_0(Z_{i,t}-c)-\beta_1 A_{i,t}(Z_{i,t}-c)\Big)^2.
    \end{multline}

Here $e_1=(1,0,\dots,0)^\top$ selects $\tau$ from the entire vector of parameters.

    \item Baseline 2: Consider a naive long-run LLR approach where
    we collapse each trajectory to the first-period net-present outcome $\Gamma_{i,0}^Y$ and treatment exposure $\Gamma_{i,0}^A$ (as defined in~\eqref{eqn:def-Gamma}), and run standard LLR of $\Gamma_{i,0}^Y$ and $\Gamma_{i,0}^A$ on the first-period  running variable $Z_{i,0}$. Precisely,
        \begin{multline}\label{eq:def-naive}
       \wh{\tau}^R_{\,\text{naive}}(h):=e_1^\top\argmin_{(\tau,\,\alpha,\,\beta_0,\,\beta_1)}\frac{1}{n}\sum_{i=1}^n K\bigg(\frac{|Z_{i,0}-c|}{h}\bigg)\Big(\Gamma_{i,0}^R -\tau A_{i,0}-\alpha\\-\beta_0(Z_{i,0}-c)-\beta_1 A_{i,0}(Z_{i,0}-c)\Big)^2,
    \end{multline}
    where $R\in\{Y,A\}$, and define $\wh\tau_{\,\mathrm{RD},\,\text{naive}}:=\wh\tau^Y_{\,\text{naive}}(h_n) \,\big/\,  \wh{\tau}^A_{\,\text{naive}}(h_n)$.
\end{enumerate}

 The first baseline procedure is expected to recover the short-run, partial equilibrium effect (i.e., the one-period RD jump at the threshold), which in this case is given by $\tau_\text{partial  eq.}=\theta (c-\mu_0)=1$. This parameter can be very different from the marginal policy effect parameter $\tauRD$.  For example, in Setting 1 with discount factor $\gamma=0.8$, $\tauRD\approx 3$ (three times the partial equilibrium effect). 
 The second baseline procedure tries to naively incorporate long-run dynamics by first aggregating all future outcomes and treatment decisions into the net-present quantities $\Gamma_{i,0}^Y$ and $\Gamma_{i,0}^A$, and then running a single cross-sectional RD at time $t=0$. In the limit, $\wh{\tau}_{\,\mathrm{RD},\,\text{naive}}(h)$ targets 
 the marginal effect of infinitesimally lowering the threshold at the initial period on discounted outcomes per additional discounted treatment generated by that initial perturbation. However, this estimand only uses information from the first time the running variable is near the cutoff and treats the future evolution of $Z_t$ as fixed. 
 As a result, this naive long-run LLR generally does not target the dynamic marginal policy effect and can be biased if the treatment policy has strong feedback effects on the future state trajectory.

\begin{table}[p]
\centering
\renewcommand{\arraystretch}{0.95}
\setlength{\tabcolsep}{3pt}

\begin{subtable}{\linewidth}
\caption{$\gamma=0.5$ ($\tauRD\approx 1.745$, $\tau_{\,\text{partial eq.}}=1$)}
\centering
\begin{tabular}{@{} c cc cc cc cc c @{}}
\toprule
& \multicolumn{2}{c}{\textbf{Proposed LLR}}
& \multicolumn{2}{c}{\textbf{AK variant}}
& \multicolumn{2}{c}{\textbf{Naive long-run}}
& \multicolumn{3}{c}{\textbf{Standard LLR}}\\
& \multicolumn{2}{c}{\small (targeting $\tau_{\mathrm{RD}}$)}
& \multicolumn{2}{c}{\small (targeting $\tau_{\mathrm{RD}}$)}
& \multicolumn{2}{c}{\small (targeting $\tau_{\mathrm{RD}}$)}
& \multicolumn{2}{c}{\small (targeting $\tau_{\mathrm{RD}}$)}
& \multicolumn{1}{c}{\small (for $\tau_{\text{partial eq.}}$)}\\
\cmidrule(lr){2-3}\cmidrule(lr){4-5}\cmidrule(lr){6-7}\cmidrule(lr){8-9}\cmidrule(l){10-10}
\textbf{$n$}
& \textbf{coverage} & \textbf{width}
& \textbf{coverage} & \textbf{width}
& \textbf{coverage} & \textbf{width}
& \textbf{coverage} & \textbf{width}
& \textbf{coverage}\\
\midrule
500    & 93.1\% & 6.804 & 95.5\% & 7.572 & 92.9\% & 11.645 & 68.7\% & 2.073 & 95.3\% \\
1000   & 95.1\% & 5.054 & 97.5\% & 5.624 & 95.0\% & 8.782  & 52.8\% & 1.545 & 94.5\% \\
2000   & 94.1\% & 3.729 & 97.4\% & 4.149 & 93.7\% & 6.544  & 27.5\% & 1.142 & 95.5\% \\
4000   & 95.6\% & 2.765 & 97.1\% & 3.076 & 95.0\% & 4.899  & 7.7\%  & 0.848 & 94.7\% \\
8000   & 94.9\% & 2.042 & 96.6\% & 2.273 & 95.0\% & 3.625  & 0.5\%  & 0.628 & 94.3\% \\
16000  & 93.8\% & 1.512 & 96.7\% & 1.682 & 94.5\% & 2.689  & 0.0\%  & 0.465 & 95.2\% \\
32000  & 94.3\% & 1.121 & 96.8\% & 1.248 & 95.1\% & 1.995  & 0.0\%  & 0.346 & 95.0\% \\
64000  & 94.4\% & 0.832 & 97.0\% & 0.926 & 94.3\% & 1.486  & 0.0\%  & 0.257 & 95.6\% \\
\bottomrule
\end{tabular}\\\vspace{1mm}
\end{subtable}

\begin{subtable}{\linewidth}
\caption{$\gamma=0.8$ ($\tauRD\approx 3.014$, $\tau_{\,\text{partial eq.}}=1$)}
\centering
\begin{tabular}{@{} c cc cc cc cc c @{}}
\toprule
& \multicolumn{2}{c}{\textbf{Proposed LLR}}
& \multicolumn{2}{c}{\textbf{AK variant}}
& \multicolumn{2}{c}{\textbf{Naive long-run}}
& \multicolumn{3}{c}{\textbf{Standard LLR}}\\
& \multicolumn{2}{c}{\small (targeting $\tau_{\mathrm{RD}}$)}
& \multicolumn{2}{c}{\small (targeting $\tau_{\mathrm{RD}}$)}
& \multicolumn{2}{c}{\small (targeting $\tau_{\mathrm{RD}}$)}
& \multicolumn{2}{c}{\small (targeting $\tau_{\mathrm{RD}}$)}
& \multicolumn{1}{c}{\small (for $\tau_{\text{partial eq.}}$)}\\
\cmidrule(lr){2-3}\cmidrule(lr){4-5}\cmidrule(lr){6-7}\cmidrule(lr){8-9}\cmidrule(l){10-10}
\textbf{$n$}
& \textbf{coverage} & \textbf{width}
& \textbf{coverage} & \textbf{width}
& \textbf{coverage} & \textbf{width}
& \textbf{coverage} & \textbf{width}
& \textbf{coverage}\\
\midrule
500    & 92.7\% & 11.134 & 94.9\% & 12.389 & 91.0\% & 29.268 & 3.0\% & 2.073 & 95.3\% \\
1000   & 95.9\% & 8.026  & 97.4\% & 8.931  & 94.8\% & 22.343 & 0.4\% & 1.545 & 94.5\% \\
2000   & 92.9\% & 5.714  & 96.0\% & 6.359  & 95.0\% & 16.609 & 0.0\% & 1.142 & 95.5\% \\
4000   & 94.2\% & 4.169  & 96.5\% & 4.639  & 95.3\% & 12.260 & 0.0\% & 0.848 & 94.7\% \\
8000   & 95.3\% & 3.004  & 97.0\% & 3.342  & 95.7\% & 9.083  & 0.0\% & 0.628 & 94.3\% \\
16000  & 94.7\% & 2.199  & 96.4\% & 2.447  & 94.9\% & 6.694  & 0.0\% & 0.465 & 95.2\% \\
32000  & 94.6\% & 1.618  & 96.9\% & 1.801  & 94.5\% & 5.011  & 0.0\% & 0.346 & 95.0\% \\
64000  & 94.2\% & 1.194  & 96.4\% & 1.329  & 93.8\% & 3.745  & 0.0\% & 0.257 & 95.6\% \\
\bottomrule
\end{tabular}\\\vspace{1mm}
\end{subtable}

\begin{subtable}{\linewidth}
\caption{$\gamma=1.0$ ($\tauRD\approx 4.325$, $\tau_{\,\text{partial eq.}}=1$)}
\centering
\begin{tabular}{@{} c cc cc cc cc c @{}}
\toprule
& \multicolumn{2}{c}{\textbf{Proposed LLR}}
& \multicolumn{2}{c}{\textbf{AK variant}}
& \multicolumn{2}{c}{\textbf{Naive long-run}}
& \multicolumn{3}{c}{\textbf{Standard LLR}}\\
& \multicolumn{2}{c}{\small (targeting $\tau_{\mathrm{RD}}$)}
& \multicolumn{2}{c}{\small (targeting $\tau_{\mathrm{RD}}$)}
& \multicolumn{2}{c}{\small (targeting $\tau_{\mathrm{RD}}$)}
& \multicolumn{2}{c}{\small (targeting $\tau_{\mathrm{RD}}$)}
& \multicolumn{1}{c}{\small (for $\tau_{\text{partial eq.}}$)}\\
\cmidrule(lr){2-3}\cmidrule(lr){4-5}\cmidrule(lr){6-7}\cmidrule(lr){8-9}\cmidrule(l){10-10}
\textbf{$n$}
& \textbf{coverage} & \textbf{width}
& \textbf{coverage} & \textbf{width}
& \textbf{coverage} & \textbf{width}
& \textbf{coverage} & \textbf{width}
& \textbf{coverage}\\
\midrule
500    & 92.6\% & 19.140 & 94.6\% & 21.298 & 88.2\% & 91.731 & 0.0\% & 2.073 & 95.3\% \\
1000   & 96.6\% & 14.139 & 98.1\% & 15.733 & 92.9\% & 70.204 & 0.0\% & 1.545 & 94.5\% \\
2000   & 94.4\% & 9.982  & 96.3\% & 11.108 & 92.8\% & 51.788 & 0.0\% & 1.142 & 95.5\% \\
4000   & 94.4\% & 7.220  & 96.4\% & 8.035  & 95.0\% & 40.223 & 0.0\% & 0.848 & 94.7\% \\
8000   & 94.7\% & 5.143  & 96.8\% & 5.723  & 97.0\% & 27.793 & 0.0\% & 0.628 & 94.3\% \\
16000  & 95.0\% & 3.723  & 97.2\% & 4.143  & 98.0\% & 20.628 & 0.0\% & 0.465 & 95.2\% \\
32000  & 94.0\% & 2.703  & 96.9\% & 3.008  & 97.4\% & 15.502 & 0.0\% & 0.346 & 95.0\% \\
64000  & 92.7\% & 1.983  & 96.7\% & 2.207  & 90.7\% & 11.681 & 0.0\% & 0.257 & 95.6\% \\
\bottomrule
\end{tabular}
\end{subtable}

\caption{Empirical coverage and median width across $1{,}000$ replications of 95\% CIs for Setting 1 (with drift parameter $\delta=0$). We report results for the proposed LLR (using the standard normal quantile $1.96$ and the \citet{ArmstrongKolesar2020} critical value $2.18$), the naive long-run LLR, and the standard LLR. We also report coverage for the standard LLR estimator targeting the partial equilibrium effect $\tau_{\text{partial eq.}}=1$. We vary the discount factor $\gamma\in\{0.5,0.8,1.0\}$ and use the uniform kernel with the \citet{imbens2012optimal} bandwidth. The oracle estimates of $\tauRD$ are obtained using $5{,}000{,}000$ replications.}
\label{tab:coverage-llr}
\end{table}

\begin{table}[p]
\centering
\renewcommand{\arraystretch}{0.95}
\setlength{\tabcolsep}{3pt}

\begin{subtable}{\linewidth}
\caption{$\gamma=0.5$ ($\tauRD\approx 1.706$, $\tau_{\,\text{partial eq.}}=1$)}
\centering
\begin{tabular}{@{} c cc cc cc cc c @{}}
\toprule
& \multicolumn{2}{c}{\textbf{Proposed LLR}}
& \multicolumn{2}{c}{\textbf{AK variant}}
& \multicolumn{2}{c}{\textbf{Naive long-run}}
& \multicolumn{3}{c}{\textbf{Standard LLR}}\\
& \multicolumn{2}{c}{\small (targeting $\tau_{\mathrm{RD}}$)}
& \multicolumn{2}{c}{\small (targeting $\tau_{\mathrm{RD}}$)}
& \multicolumn{2}{c}{\small (targeting $\tau_{\mathrm{RD}}$)}
& \multicolumn{2}{c}{\small (targeting $\tau_{\mathrm{RD}}$)}
& \multicolumn{1}{c}{\small (for $\tau_{\text{partial eq.}}$)}\\
\cmidrule(lr){2-3}\cmidrule(lr){4-5}\cmidrule(lr){6-7}\cmidrule(lr){8-9}\cmidrule(l){10-10}
\textbf{$n$}
& \textbf{coverage} & \textbf{width}
& \textbf{coverage} & \textbf{width}
& \textbf{coverage} & \textbf{width}
& \textbf{coverage} & \textbf{width}
& \textbf{coverage}\\
\midrule
500    & 93.1\% & 6.739 & 95.6\% & 7.499 & 93.1\% & 11.564 & 71.2\% & 2.006 & 94.9\% \\
1000   & 95.8\% & 5.031 & 97.7\% & 5.599 & 95.5\% & 8.727  & 54.9\% & 1.499 & 96.1\% \\
2000   & 94.7\% & 3.778 & 96.6\% & 4.204 & 93.7\% & 6.506  & 28.6\% & 1.117 & 94.1\% \\
4000   & 95.1\% & 2.816 & 97.4\% & 3.134 & 95.1\% & 4.854  & 9.1\%  & 0.829 & 94.8\% \\
8000   & 94.8\% & 2.088 & 96.6\% & 2.323 & 95.3\% & 3.585  & 1.6\%  & 0.616 & 94.7\% \\
16000  & 94.4\% & 1.551 & 96.8\% & 1.725 & 94.5\% & 2.660  & 0.0\%  & 0.457 & 95.4\% \\
32000  & 94.9\% & 1.149 & 97.2\% & 1.278 & 95.2\% & 1.980  & 0.0\%  & 0.340 & 93.9\% \\
64000  & 92.9\% & 0.854 & 95.0\% & 0.951 & 94.2\% & 1.472  & 0.0\%  & 0.252 & 94.0\% \\
\bottomrule
\end{tabular}\\\vspace{1mm}
\end{subtable}

\begin{subtable}{\linewidth}
\caption{$\gamma=0.8$ ($\tauRD\approx 2.819$, $\tau_{\,\text{partial eq.}}=1$)}
\centering
\begin{tabular}{@{} c cc cc cc cc c @{}}
\toprule
& \multicolumn{2}{c}{\textbf{Proposed LLR}}
& \multicolumn{2}{c}{\textbf{AK variant}}
& \multicolumn{2}{c}{\textbf{Naive long-run}}
& \multicolumn{3}{c}{\textbf{Standard LLR}}\\
& \multicolumn{2}{c}{\small (targeting $\tau_{\mathrm{RD}}$)}
& \multicolumn{2}{c}{\small (targeting $\tau_{\mathrm{RD}}$)}
& \multicolumn{2}{c}{\small (targeting $\tau_{\mathrm{RD}}$)}
& \multicolumn{2}{c}{\small (targeting $\tau_{\mathrm{RD}}$)}
& \multicolumn{1}{c}{\small (for $\tau_{\text{partial eq.}}$)}\\
\cmidrule(lr){2-3}\cmidrule(lr){4-5}\cmidrule(lr){6-7}\cmidrule(lr){8-9}\cmidrule(l){10-10}
\textbf{$n$}
& \textbf{coverage} & \textbf{width}
& \textbf{coverage} & \textbf{width}
& \textbf{coverage} & \textbf{width}
& \textbf{coverage} & \textbf{width}
& \textbf{coverage}\\
\midrule
500    & 94.3\% & 8.146 & 96.3\% & 9.064 & 91.3\% & 26.702 & 6.3\% & 2.006 & 94.9\% \\
1000   & 95.5\% & 6.065 & 97.2\% & 6.749 & 94.3\% & 20.831 & 0.4\% & 1.499 & 96.1\% \\
2000   & 95.4\% & 4.532 & 97.5\% & 5.043 & 95.3\% & 15.376 & 0.0\% & 1.117 & 94.1\% \\
4000   & 95.4\% & 3.377 & 97.4\% & 3.758 & 95.5\% & 11.403 & 0.0\% & 0.829 & 94.8\% \\
8000   & 95.4\% & 2.508 & 96.9\% & 2.790 & 95.6\% & 8.438  & 0.0\% & 0.616 & 94.7\% \\
16000  & 93.9\% & 1.860 & 95.8\% & 2.069 & 95.4\% & 6.263  & 0.0\% & 0.457 & 95.4\% \\
32000  & 94.0\% & 1.382 & 96.2\% & 1.538 & 94.8\% & 4.662  & 0.0\% & 0.340 & 93.9\% \\
64000  & 93.5\% & 1.027 & 96.0\% & 1.142 & 93.4\% & 3.479  & 0.0\% & 0.252 & 94.0\% \\
\bottomrule
\end{tabular}\\\vspace{1mm}
\end{subtable}

\begin{subtable}{\linewidth}
\caption{$\gamma=1.0$ ($\tauRD\approx 4.308$, $\tau_{\,\text{partial eq.}}=1$)}
\centering
\begin{tabular}{@{} c cc cc cc cc c @{}}
\toprule
& \multicolumn{2}{c}{\textbf{Proposed LLR}}
& \multicolumn{2}{c}{\textbf{AK variant}}
& \multicolumn{2}{c}{\textbf{Naive long-run}}
& \multicolumn{3}{c}{\textbf{Standard LLR}}\\
& \multicolumn{2}{c}{\small (targeting $\tau_{\mathrm{RD}}$)}
& \multicolumn{2}{c}{\small (targeting $\tau_{\mathrm{RD}}$)}
& \multicolumn{2}{c}{\small (targeting $\tau_{\mathrm{RD}}$)}
& \multicolumn{2}{c}{\small (targeting $\tau_{\mathrm{RD}}$)}
& \multicolumn{1}{c}{\small (for $\tau_{\text{partial eq.}}$)}\\
\cmidrule(lr){2-3}\cmidrule(lr){4-5}\cmidrule(lr){6-7}\cmidrule(lr){8-9}\cmidrule(l){10-10}
\textbf{$n$}
& \textbf{coverage} & \textbf{width}
& \textbf{coverage} & \textbf{width}
& \textbf{coverage} & \textbf{width}
& \textbf{coverage} & \textbf{width}
& \textbf{coverage}\\
\midrule
500    & 94.6\% & 13.578 & 97.8\% & 15.109 & 93.4\% & 78.775 & 0.0\% & 2.006 & 94.9\% \\
1000   & 94.6\% & 9.903  & 96.8\% & 11.019 & 94.6\% & 59.428 & 0.0\% & 1.499 & 96.1\% \\
2000   & 95.4\% & 7.363  & 97.5\% & 8.193  & 94.2\% & 43.770 & 0.0\% & 1.117 & 94.1\% \\
4000   & 93.6\% & 5.450  & 96.7\% & 6.064  & 95.5\% & 32.096 & 0.0\% & 0.829 & 94.8\% \\
8000   & 94.3\% & 4.010  & 96.3\% & 4.462  & 97.0\% & 23.007 & 0.0\% & 0.616 & 94.7\% \\
16000  & 94.1\% & 2.984  & 96.5\% & 3.320  & 96.4\% & 16.948 & 0.0\% & 0.457 & 95.4\% \\
32000  & 93.7\% & 2.213  & 96.3\% & 2.463  & 95.3\% & 12.729 & 0.0\% & 0.340 & 93.9\% \\
64000  & 94.0\% & 1.639  & 96.1\% & 1.824  & 92.2\% & 9.429  & 0.0\% & 0.252 & 94.0\% \\
\bottomrule
\end{tabular}
\end{subtable}

\caption{Empirical coverage and median width across $1{,}000$ replications of 95\% CIs for Setting 2 (with drift parameter $\delta=1$). We report results for the proposed LLR (using the standard normal quantile $1.96$ and the \citet{ArmstrongKolesar2020} critical value $2.18$), the naive long-run LLR, and the standard LLR.  We also report coverage for the standard LLR estimator targeting the partial equilibrium effect $\tau_{\text{partial eq.}}=1$. We vary the discount factor $\gamma\in\{0.5,0.8,1.0\}$ and use the uniform kernel with the \citet{imbens2012optimal} bandwidth. The oracle estimates of $\tauRD$ are obtained using $5{,}000{,}000$ replications.}
\label{tab:coverage-llr-nonstationary}
\end{table}

The empirical performance of our proposed method, as well as the baseline procedures described above, depends crucially on the choice of the bandwidth $h$ used in the kernel function. Here we use the data-driven IK bandwidth \citep{imbens2012optimal} for both baseline approaches, where for the second baseline we use a single bandwidth for the numerator and the denominator using the fuzzy-RD approach of \citep{imbens2012optimal}. The proposed local linear regression \eqref{eq:def-LLR} is run with the data-driven IK-style bandwidth we derive in \ref{app:bandwidth-selection}.  We also include the \citet{ArmstrongKolesar2020} variant which replaces the standard normal quantile $1.96$ with the critical value $2.18$ to account for worst-case bias over a smoothness class of regression functions. We emphasize that this is an ad hoc extension to the conventional confidence intervals, as the theoretical analysis has been established only for the static ($T=0$) case.  %While this bandwidth choice may not be optimal for our method, it provides a transparent comparison across all candidate approaches. 

We vary the discount factor $\gamma$ in $\{0.5, 0.8, 1\}$, and use an exponentially increasing sequence of sample sizes, namely $\{500, 1000, 2000, 4000,\dots, 64000\}$. We report in \cref{tab:coverage-llr,tab:coverage-llr-nonstationary} the empirical coverage and median width of $95\%$ confidence intervals constructed in Settings 1 and 2, respectively. The results are aggregated across $1000$ replications.
\cref{tab:coverage-llr,tab:coverage-llr-nonstationary} demonstrate that the conventional confidence intervals constructed from the proposed LLR provide near-nominal coverage (as expected with the bandwidth scaling at the MSE-optimal rate), and its \citet{ArmstrongKolesar2020} variant provides nominal coverage (as expected---heuristically---from the static RD case). In contrast, standard LLR confidence intervals severely undercover when used to make inference on $\tauRD$, with coverage dropping to zero. This is also not surprising, since the standard LLR targets  $\tau_{\,\text{partial eq.}}$, not the dynamic marginal policy effect $\tauRD$. Indeed, \cref{tab:coverage-llr,tab:coverage-llr-nonstationary} illustrate that the standard LLR confidence intervals provide nominal coverage for the parameter $\tau_{\,\text{partial eq.}}$.

The naive long-run LLR (Baseline 2) also provides near-nominal coverage in these simulations. However, the width of the naive confidence intervals are substantially wider than our proposed ones. In some cases, its coverage deteriorates, dropping to around~90\%. This breakdown is expected since this naive method only exploits the first-period threshold proximity and ignores how changing the threshold affects the frequency with which units cross the threshold in future periods. In contrast, our method maintains nominal coverage across all settings by correctly pooling information from all time periods using the twice-discounted weighting scheme.

The proposed confidence intervals are substantially wider than those from standard LLR, with widths increasing as $\gamma\uparrow 1$. This is due to the fact that the asymptotic variance of our estimator aggregates uncertainty across all future time periods, and the effective number of periods contributing to this variance grows as the discount factor increases. More fundamentally, the problem of estimating the long-term marginal policy effect $\tauRD$ is intrinsically more difficult than estimating the short-run partial equilibrium effect $\tau_{\,\text{partial eq.}}$, because $\tauRD$  incorporates spillovers and dynamic treatment switching across all future periods, each contributing additional uncertainty. 

\subsection{UVa/Padova Type-1 Diabetes Simulator}\label{sec:gluc-sim}

Here we consider a simulation exercise based on the FDA-approved UVa/Padova Type-1 Diabetes simulator \citep{UVA-PADOVA}, implemented via the \texttt{simglucose} Python package \citep{simglucose}. Motivated by clinical decision rules in diabetes management, this simulator incorporates realistic physiological dynamics: Insulin administration at time $t$ affects not only the immediate next glucose reading, but also influences future glucose trajectories and thus future treatment decisions.

\paragraph{Simulation design.} 
Consider a diabetes management protocol where Type-1 diabetes patients are monitored at $30$-minute intervals. At each time index $t$, we observe the patients' continuous glucose monitoring (CGM) readings denoted by $\mathrm{CGM}_{i,t}$ (in mg/dL). The patients are given an insulin bolus \emph{before taking a meal} if their CGM reading exceeds a pre-specified threshold $c$. Concretely, the treatment indicator is defined as $$A_{i,t} =  \mathrm{Meal}_{i,t}\cdot\mathbf{1}\{\mathrm{CGM}_{i,t} \geq c\},$$ where $\mathrm{Meal}_{i,t}$ is the indicator for whether a meal is consumed between time periods $t$ and $t+1$, chosen randomly (and endogenously) in the simulator. The reward at time $t$ is the negative Kovatchev risk index  \citep{ClarkeKovatchev}, defined as:
$$Y_{i,t}= -10 \left(1.509\left( \left(\log \mathrm{CGM}_{i,t+1}\right)^{1.084}-5.381\right)\right)^2.$$
Since the threshold-based treatment rule is only imposed before a meal and no action is taken otherwise, we define the running variable as 
$$Z_{i,t}=\begin{cases}
    \mathrm{CGM}_{i,t}&\text{if}\quad\mathrm{Meal}_{i,t}=1,\\ -\infty &\text{if}\quad\mathrm{Meal}_{i,t}=0.
\end{cases}.$$
With this convention, we indeed have $A_{i,t}=\pi_c(Z_{i,t})=\ind{Z_{i,t}\ge c}$ at all time periods.

\paragraph{Oracle estimation.}
To obtain ground-truth values of $\tauRD$ for different discount factors $\gamma$, we generate $100{,}000$ trajectories at each of the thresholds $c \in \{120, 125, \dots, 180\}$ and estimate the policy gradients via finite differences. \cref{fig:diabetes_oracle_days} reports oracle estimates across observation horizons of $1$, $2$, $4$, and $8$ days. We note that both the policy gradients $\gradc V^Y(\pi_c)$ and $\gradc V^A(\pi_c)$ vary systematically with the horizon, yet their ratio $\tauRD$ remains stable across all horizons. This suggests that, extending the horizon changes the outcome and treatment gradients in tandem, so the marginal change in welfare per unit change in treatment intensity is largely insensitive to the length of the observation window. On the other hand, \cref{fig:diabetes_oracle_gammas} reports the oracle estimates across discount factors $\gamma$, from $\gamma = 0$ (immediate reward) to $\gamma = 1$ (infinite horizon). Here we vary the discount factor $\gamma$ as $\{1-\frac{1}{24}, 1-\frac{1}{48},1-\frac{1}{96},1-\frac{1}{192},1\}$, corresponding to an effective horizon of $0.5, 1, 2, 4, \infty$ days, respectively. The marginal policy effect $\tauRD$ exhibits a clean, monotone behavior:~It is approximately null at $\gamma=0$, and increases steadily to a positive value as $\gamma \to 1$. In other words, a marginal increase in the threshold has negligible effect on the glucose reading immediately after the insulin dose, but yields a positive, cumulative effect on the long-run outcomes.

\begin{figure}[tp]
    \centering
    \includegraphics[width=\linewidth]{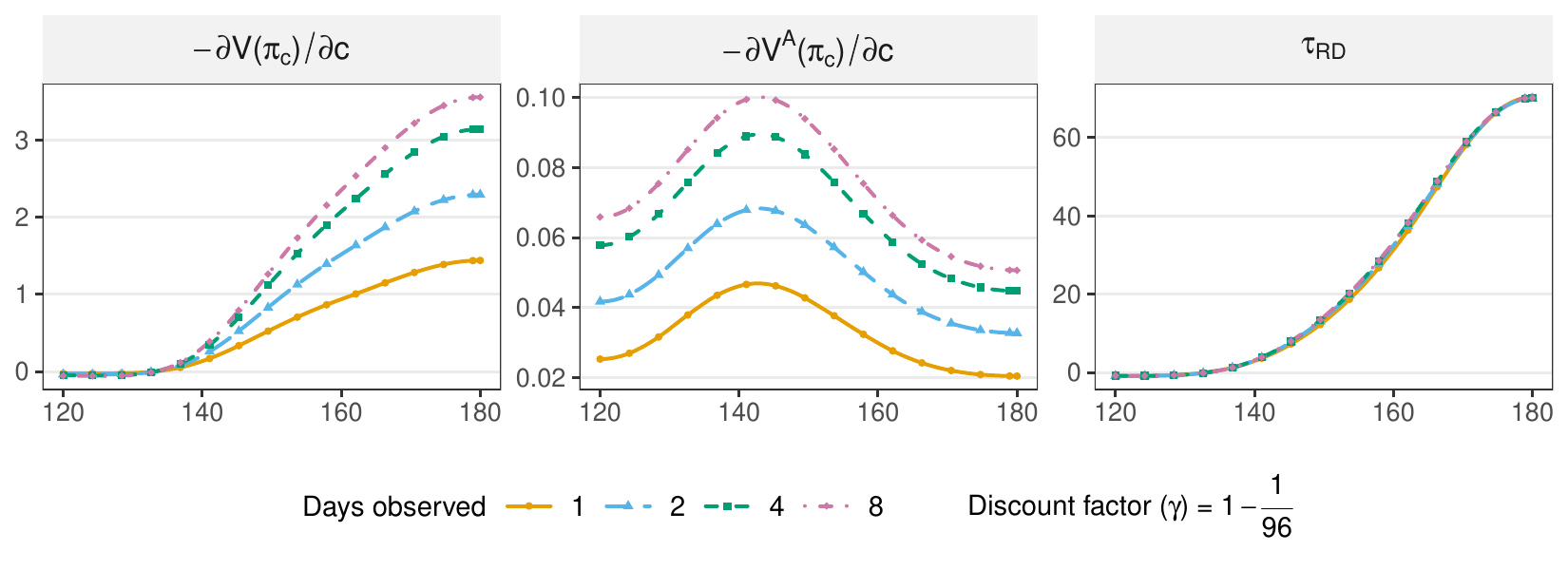}
    \caption{Oracle estimates across observation horizons of $1$, $2$, $4$, and $8$ days. The first two panels plot the derivatives of the outcome and treatment value functions against the threshold; both vary systematically with the horizon. The third panel shows that the marginal policy effect $\tauRD$, defined as their ratio,  remains stable across all four horizons.}
    \label{fig:diabetes_oracle_days}
\end{figure}

\begin{figure}[tp]
    \centering
    \includegraphics[width=\linewidth]{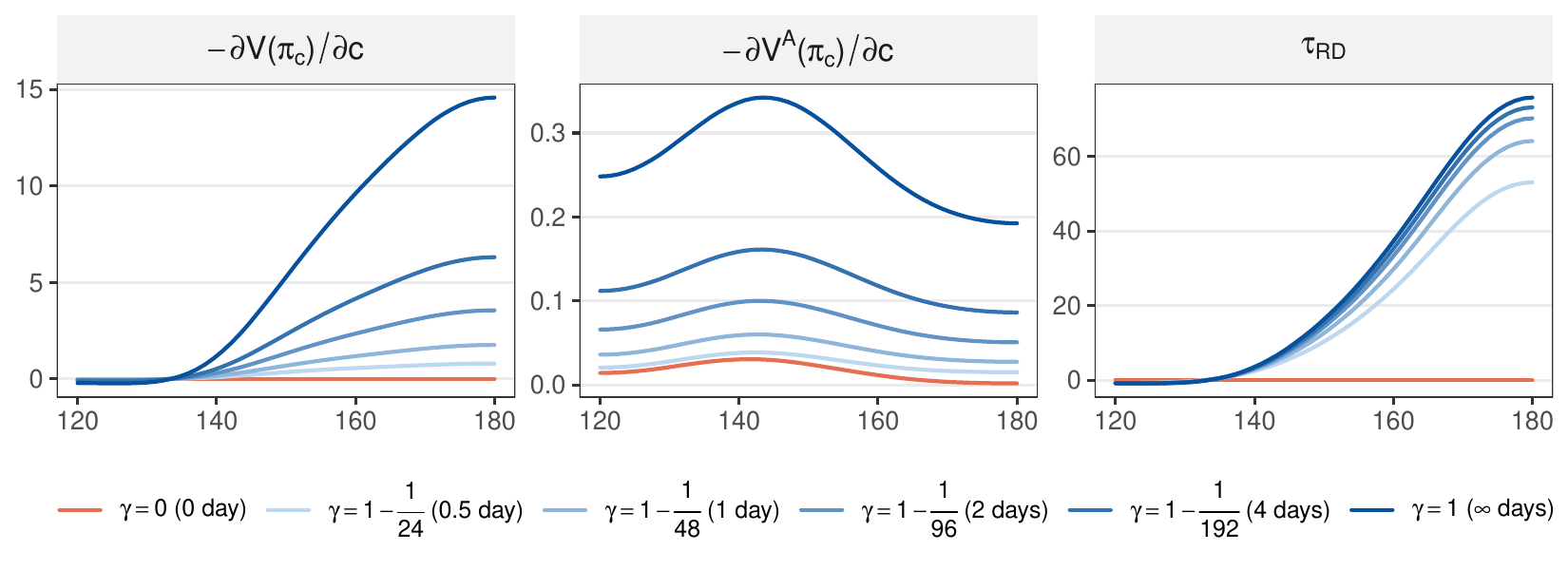}
    \caption{Oracle estimates across different discounting factors (with the effective horizons (in days) shown in parentheses). The first two panels plot the derivatives of the outcome and treatment value functions against the threshold. The third panel shows that the marginal policy effect $\tauRD$ increases gradually from $\gamma=0$ (short-term reward) to $\gamma=1$ (infinite horizon, long-term reward).}
    \label{fig:diabetes_oracle_gammas}
\end{figure}

\begin{table}[tp]
\centering
\renewcommand{\arraystretch}{1.2}
\setlength{\tabcolsep}{4pt}
\begin{tabular}{@{}cc cc cc cc cc@{}}
\toprule
& & \multicolumn{2}{c}{\textbf{Proposed LLR}}
  & \multicolumn{2}{c}{\textbf{AK variant}}
  & \multicolumn{2}{c}{\textbf{Naive long-run}}
  & \multicolumn{2}{c}{\textbf{Standard LLR}} \\
\cmidrule(lr){3-4}\cmidrule(lr){5-6}\cmidrule(lr){7-8}\cmidrule(l){9-10}
$\gamma$ & Eff.~horizon
& coverage & width
& coverage & width
& coverage & width
& coverage & width \\
\midrule
$1-\frac{1}{24}$  & $0.5$ day & $94.8\%$ & $20.12$ & $96.8\%$ & $22.39$ & $95.3\%$ & $22.69$  & $0.0\%$ & $0.71$ \\
$1-\frac{1}{48}$  & $1$ day   & $94.3\%$ & $23.80$ & $97.2\%$ & $26.49$ & $94.2\%$ & $35.57$  & $0.0\%$ & $0.71$ \\
$1-\frac{1}{96}$  & $2$ days  & $95.6\%$ & $28.00$ & $97.2\%$ & $31.16$ & $94.4\%$ & $53.78$  & $0.0\%$ & $0.71$ \\
$1-\frac{1}{192}$ & $4$ days  & $95.4\%$ & $34.09$ & $97.1\%$ & $37.94$ & $95.3\%$ & $74.74$  & $0.0\%$ & $0.71$ \\
$1$               & $\infty$  & $95.4\%$ & $45.62$ & $97.4\%$ & $50.77$ & $97.6\%$ & $119.21$ & $0.0\%$ & $0.71$ \\
\bottomrule
\end{tabular}
\caption{Empirical coverage and median width of 95\% confidence intervals for the dynamic marginal policy effect $\tauRD$ for the threshold $c=150$ using the Type-1 diabetes simulator \citep{UVA-PADOVA}. The confidence intervals are constructed from the proposed LLR (using the standard normal quantile $1.96$ as well as the \citet{ArmstrongKolesar2020} critical value $2.18$), the naive long-run LLR, and the standard LLR.  Results are aggregated across $1{,}000$ replications, each with $n=1{,}000$ independent trajectories.}
\label{tab:diabetes_results}
\end{table}

\paragraph{Confidence interval evaluation.} 
Next, we generate data under the thresholding policy with threshold $c={150}$ mg/dL and construct 95\% confidence intervals using the same methods as in \cref{sec:autoreg}, namely: The standard LLR \eqref{eq:static}, the naive long-run approach as in \eqref{eq:def-naive}, the conventional confidence intervals constructed from our proposed twice-discounted LLR \eqref{eq:def-LLR}, and its \citet{ArmstrongKolesar2020} variant that uses the critical value $2.18$ instead of the standard normal quantile of $1.96$. As in \cref{sec:autoreg}, we use time fixed-effects for variance reduction and follow \citet{imbens2012optimal} to select the bandwidth for each the methods. 

We simulate $n = 1000$ adult patients\footnote{The \texttt{simglucose} simulator \citep{simglucose} only generates the trajectories for $10$ adult patients, but it allows using `random scenarios', which we use with $n/10$ unique seeds per replication to generate $n$ trajectories.} in each replication, and report empirical coverage and interval width across $1000$ replications in \cref{tab:diabetes_results}. Here, the conventional confidence intervals constructed from our proposed LLR maintain nominal coverage across all discount factors (and so does the \citet{ArmstrongKolesar2020} variant). The standard LLR confidence intervals provide zero coverage, which is expected since the dynamic marginal policy effect is much larger than the short-term / partial equilibrium effect $\tau_{\text{partial eq.}}\approx 0$ (\cref{fig:diabetes_oracle_gammas})---which is what this method is designed to target. The naive long-run method, on the other hand, provides nominal coverage, but at the expense of extremely wide confidence intervals---which is uninformative in most of the replications (contains both $\tauRD>0$ as well as zero).

The above results demonstrate that our proposed method provides reliable performance for the UVa/Padova Type-1 Diabetes simulator with complex treatment dynamics and physiological feedback mechanisms. The breakdown of competing methods at higher discount factors highlights the importance of properly accounting for the temporal structure of thresholding designs is essential for valid inference about long-run policy effects.

\section{Discussion}

Treatment decisions based on clinical thresholds are ubiquitous in modern healthcare. Patients are often evaluated periodically against the same threshold, creating a dynamic setting where treatment decision at one visit affects clinical measurements, and thus treatment assignments, at subsequent visits. This feedback loop between treatment and state evolution poses a fundamental challenge for policy evaluation: A social planner considering whether to adjust the treatment threshold needs to determine the long-term welfare consequences, accounting for how today's treatment decisions alter tomorrow's patient states and treatment patterns.

The standard approaches for analyzing threshold-based policies in regression discontinuity designs compare the outcomes of units just above and below the treatment threshold and have been widely used in cross-sectional settings. However, the same approach, when naively applied in the dynamic setting, can only recover the immediate, partial equilibrium effect of crossing the threshold rather than the cumulative impact on health over time. And existing attempts to accommodate dynamics rely on restrictive parametric assumptions \citep{cellini2010value} or strong conditional independence conditions \citep{hsu2024dynamic} that are unlikely to hold in generic dynamic thresholding designs.

In this paper, we pursued a new approach to causal inference in dynamic thresholding designs that focuses on the marginal policy effect of increasing/decreasing the treatment intensity via moving the treatment-eligibility threshold. We showed that our marginal policy effect estimand, $\tauRD$, is equivalent to the standard regression-discontinuity parameter in the static setting, and remains identified in the dynamic setting using only data collected at the status-quo threshold. Finally, we proposed a tractable local linear regression algorithm that consistently estimates this quantity.

In the introduction, we discussed a paper by \citet{IIZUKA2021} that sought to assess cost-effectiveness of preventative care assigned to patients whose fasting blood sugar (FBS) was higher than a pre-specified cutoff, $c = 110$ mg/dL, at which they were declared prediabetic. In their original study, the authors only considered partial-equilibrium effects of treatment (i.e., of preventative care given one year on next-period outcomes), thus avoiding the need to model dynamics.

When applied to the setting of \citet{IIZUKA2021}, our results imply that it is possible to estimate the marginal policy effect of moving the cutoff $c$ under a fully dynamics-aware model. Furthermore, in this setting, our estimand $\tauRD$ can be interpreted as the net-present health benefit of infinitesimally lowering the FBS cutoff for pre-diabetes, scaled by the net-present cost of treatment (assuming each treatment triggered by a pre-diabetes diagnosis has the same cost). It is well known that such ratio-form estimands are exactly of the form that's needed to formally optimize cost-benefit tradeoffs in resource allocation \citep{dantzig1957discrete}: We should generally lower $c$ (and give preventative care to more people) if $\tauRD$ exceeds the corresponding cost-benefit ratio for an outside option, and raise $c$ (i.e., give less preventative care and re-allocate budget to the outside option) otherwise. As such, at least in the setting of \citet{IIZUKA2021}, our approach doesn't just give a crude proxy for a causal effect of interest, but in fact directly enables a formal, dynamics-aware cost-benefit analysis.

There are many questions left open in this paper. Here, we showed that one specific local linear regression estimator is consistent for $\tauRD$, but it would be interesting to see if other algorithms may be more efficient. Similarly, it would be valuable to investigate how best to make use of time-varying contextual information, perhaps extending results in \citet{calonico2019regression} or \citet{noack2021flexible}. Finally, in the cross-sectional setting, there has been a large literature on methods for inference based on either eliminating \citep{CCT2014robustCI} or accommodating \citep{armstrong2018optimal,noack2024bias} bias when constructing confidence intervals, and in using convex optimization to design point estimators with improved worst-case error properties relative to local linear regression \citep{imbens2019optimized,ghosh2025plrd}. Extending such contributions to the dynamic setting would be of considerable interest.

\bibliographystyle{chicago}
\bibliography{dynamic_RD_refs}

\newpage

% ---------- Appendix A (before references) ----------
\appendix
\numberwithin{equation}{section}
\numberwithin{lemma}{section}
\numberwithin{theorem}{section}
\numberwithin{proposition}{section}

\renewcommand{\thesection}{Appendix~\Alph{section}}

\makeatletter
\renewcommand{\theequation}{\@Alph\c@section.\arabic{equation}}
\renewcommand{\thetheorem}{\@Alph\c@section.\arabic{theorem}}
\renewcommand{\thelemma}{\@Alph\c@section.\arabic{lemma}}
\renewcommand{\theproposition}{\@Alph\c@section.\arabic{proposition}}
\renewcommand{\thesubsection}{\@Alph\c@section.\arabic{subsection}}
\makeatother

%\onehalfspacing
\linespread{1.1}

\section{Bandwidth Selection}\label{app:bandwidth-selection}

This appendix describes the bandwidth choices used in the simulations in \cref{sec:numerical-experiments}. For the static and naive long-run baselines, we use the corresponding \citet{imbens2012optimal} bandwidth rules. For the proposed estimator \eqref{eq:def-LLR}, we use the same residualization idea as in the fuzzy-RD bandwidth rule of \citet{imbens2012optimal}, with the variance and curvature quantities replaced by those appearing in the dynamic expansion in \cref{clt-for-twice-discounted-llr}.

\subsection{Bandwidths for the Static LLR and Naive Long-run LLR}

For the static / cross-sectional LLR baseline \eqref{eq:static}, we use the usual \citet{imbens2012optimal} bandwidth applied to the stacked one-period observations $(Y_{i,t},Z_{i,t})$. 
The naive long-run estimator \eqref{eq:def-naive} divides the local linear jump in the first-period long-run outcome $\Gamma_{i,0}^Y$ by the corresponding jump in the first-period long-run exposure $\Gamma_{i,0}^A$. We therefore use the fuzzy-RD extension of \citet{imbens2012optimal} to select a single bandwidth for this ratio estimator.
\begin{enumerate}[label=(\roman*)]
    \item Pick a pilot bandwidth $h_{\mathrm{pilot}}$, e.g., by applying the sharp IK bandwidth selector to the local linear regression of $\Gamma_{i,0}^Y$ on $Z_{i,0}$.

    \item Using $h_{\mathrm{pilot}}$, apply the naive long-run LLR \eqref{eq:def-naive} and obtain the pilot ratio estimator:
    \[
        \wh\tau_{\,\mathrm{pilot}}
        =
        \wh\tau^Y_{\mathrm{naive}}(h_{\mathrm{pilot}})
        \big/
        \wh\tau^A_{\mathrm{naive}}(h_{\mathrm{pilot}}).
    \]

    \item Run the sharp IK bandwidth selector for the LLR of $\wt\Gamma_i$ on $Z_{i,0}$, where
\[
    \wt\Gamma_i
    =
    \Gamma_{i,0}^Y-\wh\tau_{\,\mathrm{pilot}}\Gamma_{i,0}^A.
\]
\end{enumerate}
This is the usual residualized IK rule for a ratio of two RD jumps.

\subsection{Bandwidth for our Proposed Method}

The asymptotic expansion in \cref{clt-for-twice-discounted-llr} implies that minimizing the asymptotic mean squared error (AMSE) with respect to the bandwidth $h$ results in the following (infeasible) optimal bandwidth:
\begin{equation}\label{eq:dynamic-infeasible-bandwidth}
\begin{split}
   h_{\,\mathrm{IK},\mathrm{dyn}}^*&=\argmin_{h} \left\{\frac{V_{YY}+\tauRD^2V_{AA}-2\tauRD V_{YA}}{n h}
    +\frac{h^4}{4}\xi_1^2\left(\Delta\mu_Y''(c)-\tauRD\Delta\mu_A''(c)\right)^2 \right\}\\
    &= \left(\frac{V_{YY}+\tauRD^2V_{AA}-2\tauRD V_{YA}}{\xi_1^2\left(\Delta\mu_Y''(c)-\tauRD\Delta\mu_A''(c)\right)^2}\right)n^{-1/5}.
\end{split}
\end{equation} 
 Next, we derive a feasible, data-driven bandwidth $\wh{h}_{\,\mathrm{IK},\mathrm{dyn}}$ which is consistent for the AMSE-optimal bandwidth as above, in the sense that $\wh{h}_{\,\mathrm{IK},\mathrm{dyn}}/h_{\,\mathrm{IK},\mathrm{dyn}}^*\,\Pto\, 1$.
 Since the proposed estimator $\htauRD$ is a ratio of local linear estimators, we follow the procedure for the fuzzy-RD case in \citet{imbens2012optimal} and choose a single bandwidth (instead of two separate bandwidths for the numerator and the denominator in \eqref{eq:def-LLR}) by applying the sharp-RD calculation to the residualized outcomes $\Gamma^Y-\tauRD\Gamma^A$. We note, however, that for the dynamic setting, the residualized IK calculation must be combined with the `new' bias and variance expressions as given in \cref{clt-for-twice-discounted-llr}.

For ease of exposition, we first present the `unregularized' version of the proposed bandwidth; the finite-sample regularization used in the simulations is introduced later.

\begin{proc}[Dynamic IK bandwidth]\label{proc:dynamic-ik-bandwidth}
{\hspace{1mm}}
%Given data collected under threshold $c$,
\begin{enumerate}[label=(\roman*)]
        \item Compute the discounted sums  $\Gamma^Y_{i,t}$ and $\Gamma^A_{i,t}$ as in \eqref{eqn:def-Gamma}. Obtain $h_{\mathrm{pilot}}$ by applying the sharp IK bandwidth selector to the stacked pairs $(\Gamma^Y_{i,t},Z_{i,t})$. 
        
        \item Run \eqref{eq:def-LLR} with this bandwidth $h_{\mathrm{pilot}}$ to obtain the point estimate $\wh\tau_{\mathrm{pilot}}$ and the unit-clustered covariance estimates $\wh V_{YY}$, $\wh V_{AA}$ and $\wh V_{YA}$ for the two local-linear jumps (\cref{propo:sandwich-estimator}). Set
    \[
        \widehat V_{\mathrm{dyn}}
        =
        nh_{\mathrm{pilot}}\left\{
        \wh V_{YY}
        +\wh\tau_{\mathrm{pilot}}^2\wh V_{AA}
        -2\wh\tau_{\mathrm{pilot}}\wh V_{YA}
        \right\}.
    \]
    Also compute the residualized outcomes $\wt\Gamma_{i,t}:=\Gamma^Y_{i,t}-\wh\tau_{\mathrm{pilot}}\Gamma^A_{i,t}$. 
    
    \item Following \citet{imbens2012optimal}, we find pilot curvature bandwidths as follows. First pick a pilot bandwidth using the Silverman rule
    \(
        h_1=1.84\,\wh{\mathrm{sd}}\{Z_{i,t}\}\, n^{-1/5}
    \). Then compute
$$\wh{f}_\gamma(c)=\frac{\sum_{i=1}^n\sum_{t=0}^{T_n}\gamma^t \ind{|Z_{i,t}-c|\le h_1}}{2h_1\sum_{i=1}^n\sum_{t=0}^{T_n}\gamma^t}.$$
Also run the cubic regression of $\wt\Gamma_{i,t}$ on $1$, $A_{i,t}$, $Z_{i,t}-c$, $\frac{1}{2}(Z_{i,t}-c)^2$ and $\frac{1}{6}(Z_{i,t}-c)^3$, using observations with $Z_{i,t}-c$ between the left- and right-side medians. Denote by $\wh{m}_{3,\gamma}$ the coefficient of the cubic term, and by $\wh{\sigma}^2_\gamma$ the weighted residual variance: $\wh{\sigma}^2_\gamma:=\sum_{i,t}\gamma^t\,\wh{e}_{i,t}^2/\sum_{i,t}\gamma^t$. Finally, set
\[
    \begin{split}
        h_{2,+}
        &:=
        3.56\left\{\frac{\wh\sigma^2_\gamma}{\wh f_\gamma(c)\max(\wh m_{3,\gamma}^2,0.01)}\right\}^{1/7}n_+^{-1/7},\qquad n_+ := n\frac{\sum_{i=1}^n\sum_{t=0}^{T_n}\gamma^t\,\ind{Z_{i,t}\ge c}}{\sum_{i=1}^n\sum_{t=0}^{T_n}\gamma^t},\\
        h_{2,-}
        &:=
        3.56\left\{\frac{\wh\sigma^2_\gamma}{\wh f_\gamma(c)\max(\wh m_{3,\gamma}^2,0.01)}\right\}^{1/7}n_{-}^{-1/7},\qquad n_{-} := n\frac{\sum_{i=1}^n\sum_{t=0}^{T_n}\gamma^t\,\ind{Z_{i,t}< c}}{\sum_{i=1}^n\sum_{t=0}^{T_n}\gamma^t}.
    \end{split}
    \]

    \item 
    With the pilot curvature bandwidths as above, estimate $\widehat B_{\mathrm{dyn}}$ as the coefficient on $\frac12 A_{i,t}(Z_{i,t}-c)^2$ in the local quadratic regression of $\wt\Gamma_{i,t}$ on
    \[
    1,\ A_{i,t},\ Z_{i,t}-c,\ A_{i,t}(Z_{i,t}-c),\ \frac12 (Z_{i,t}-c)^2,\ \frac12 A_{i,t}(Z_{i,t}-c)^2,
    \]
    using weights $\gamma^t\, K(|Z_{i,t}-c|/h_{2,+})$ on the right side and $\gamma^t\, K(|Z_{i,t}-c|/h_{2,-})$ on the left side. Additionally, include time fixed effects exactly as in \eqref{eq:def-LLR}.
    \item Return
    \[
        \wh h_{\,\mathrm{IK},\mathrm{dyn}}
        :=
        \left(
        \frac{\widehat V_{\mathrm{dyn}}}{\xi_1^2\widehat B_{\mathrm{dyn}}^2}
        \right)^{1/5} n^{-1/5},
    \]
    where $\xi_1$ is the same constant as in \cref{clt-for-twice-discounted-llr}.
\end{enumerate}
\end{proc}

The preliminary window $h_1$ chosen by the Silverman rule is used to estimate the local noise level and density at the cutoff; the constant $1.84=(12\sqrt{\pi})^{1/5}$ is the same as in \citet{imbens2012optimal} for this preliminary uniform-window calculation. The bandwidths $h_{2,+}$ and $h_{2,-}$ are the one-sided pilot bandwidths for estimating second derivatives by local quadratic regressions. The constant $3.56$ is also borrowed from \citet{imbens2012optimal}; it comes from the corresponding MSE calculation for this curvature pilot, where the variance is of order $(n_\pm f_\gamma(c)h_{2,\pm}^5)^{-1}$ and the squared bias is of order $m_{3,\gamma}^2h_{2,\pm}^2$. As in \citet{imbens2012optimal}, the term $\max(\wh m_{3,\gamma}^2,0.01)$ stabilizes this pilot bandwidth when the global cubic estimate is close to zero. These pilot choices enter only through $\widehat B_{\mathrm{dyn}}$; the final bandwidth is governed by the dynamic AMSE as in \eqref{eq:dynamic-infeasible-bandwidth}.

The main changes in the dynamic setting enter through the two estimated objects in \eqref{eq:dynamic-infeasible-bandwidth}. The variance pilot is computed from the unit-clustered covariance matrix in \cref{propo:sandwich-estimator}, allowing temporal correlations across periods within the same trajectory. On the other hand, to ensure that the curvature pilot estimates the curvature of the dynamic response functions in \eqref{def-muRz}, we modify the standard local quadratic regression with the same discounting ($\gamma^t$) weights as in \eqref{eq:def-LLR}. We also use time fixed effects which remove period-specific local means from the pilot regression. As in the main paper, these fixed effects leave the leading curvature term unchanged.

\begin{proposition}[Consistency of the dynamic IK bandwidth]\label{prop:dynamic-bandwidth-consistency}
Suppose that the conditions of \cref{clt-for-twice-discounted-llr} hold, and that $\Delta\mu_Y''(c)-\tauRD\Delta\mu_A''(c)\neq0$. Assume further that the residualized dynamic response functions are three times continuously differentiable in a neighborhood of $c$. Then, the feasible bandwidth $\wh h_{\,\mathrm{IK},\mathrm{dyn}}$ computed in \cref{proc:dynamic-ik-bandwidth} satisfies
\[
    \wh h_{\,\mathrm{IK},\mathrm{dyn}}/h_{\,\mathrm{IK},\mathrm{dyn}}^{*}\,\Pto\,1.
\]
\end{proposition}

\paragraph{Regularization.} The plug-in rule divides by an estimated squared curvature, and this denominator is sensitive to sampling noise when the curvature estimate is close to zero. For robust finite-sample performance, \citet{imbens2012optimal} stabilize the denominator  by adding a regularization term to the squared curvature, which is defined as three times the estimated variance of the estimator of the second-derivative jump. %For a one-sided local quadratic regression with a uniform pilot window, the leading variance of the second-derivative estimator is given by\[   \frac{720\,\sigma^2}{N_\pm f(c)h_{2,\pm}^5}.\]
%The constant $720$ is the local-quadratic variance constant for estimating a boundary second derivative with regressors $1$, $Z-c$, and $(Z-c)^2$: the lower-right element of the inverse one-sided moment matrix is $180$, and the second derivative is twice the quadratic coefficient. The factor $3$ is the finite-sample inflation used by IK for this variance term. Replacing $N_\pm f(c)h_{2,\pm}$ by the number of observations in the one-sided pilot window gives the usual IK regularization term $2160\,\wh\sigma^2/(n_{\pm,2}h_{2,\pm}^4)$.
We apply the same regularization to the curvature difference estimated in Step $(iv)$ of \cref{proc:dynamic-ik-bandwidth}. The two one-sided curvatures are estimated in a single weighted local quadratic regression, so the unit-clustered covariance matrix for that regression includes the covariance between the two curvature estimates. Let $\wh R_{\mathrm{dyn}}$ denote the resulting clustered variance estimate for $\widehat B_{\mathrm{dyn}}$. The regularized bandwidth used in the simulations is
\[
    \wh h_{\,\mathrm{IK},\mathrm{dyn}}
    =
    \left(
    \frac{\widehat V_{\mathrm{dyn}}}{\xi_1^2(\widehat B_{\mathrm{dyn}}^2+3\wh R_{\mathrm{dyn}})}
    \right)^{1/5}n^{-1/5}.
\]
This regularization has no first-order effect under the conditions of \cref{prop:dynamic-bandwidth-consistency}. For the second-derivative coefficient in a local quadratic regression, the variance is of order $(nh_2^5)^{-1}$, and therefore
\[
    \wh R_{\mathrm{dyn}}=O_p\left((nh_{2,+}^5)^{-1}+(nh_{2,-}^5)^{-1}\right)=o_p(1).
\]
Thus the regularized bandwidth also satisfies $\wh h_{\,\mathrm{IK},\mathrm{dyn}}/h_{\,\mathrm{IK},\mathrm{dyn}}\,\Pto\,1$. 

\section{Main Proofs}

\subsection{\texorpdfstring{Proof of \cref{lemma:classical-RD-as-policy-gradient}}{classical RD as a policy gradient}}\label{proof-of-lemma:classical-RD-as-policy-gradient}

\begin{proof}%[Proof of \cref{lemma:classical-RD-as-policy-gradient}]
    First observe that the value function under policy $\pi_c$ is given by $$V(\pi_c)=\E[Y_i(\pi_c(Z_i))]=\int_c^\infty \mu_{1}(z)\,f(z)\,dz +\int_{-\infty}^c \mu_{0}(z)\,f(z)\,dz,$$
    where $\mu_{a}(z):=\E[Y(a)\mid Z = z]$. Differentiating with respect to $c$ using the Leibniz rule, we obtain
    $$\gradc V(\pi_c)=-(\mu_{1}(c) -\mu_{0}(c))f(c).$$
    Similarly,
    $$\gradc \E_{\pi_c}[A_i]=\gradc\int_c^\infty f(z)\,dz=-f(c).$$
    Combining the last two displays we arrive at the desired conclusion.
\end{proof}

\subsection{\texorpdfstring{Proof of \cref{prop:g-formula}}{verifying the g-formula}}
\label{proof-of-prop:g-formula}

\begin{proof}
Our proof proceeds by recursion. Our goal is to argue that, if \eqref{eq:pomdp} holds,
then the following conditional distributions are all policy independent for all $t = 0, \, 1, \, 2, \ldots$
\begin{align}
\label{eq:recurse1}
&\mathbb{P}_{\pi_c}\left[U_{i,t} \,\big|\, S_{i,t}\right] = \mathbb{P}\left[U_{i,t} \,\big|\, S_{i,t}\right], \\
\label{eq:recurse2}
&\mathbb{P}_{\pi_c}\left[Z_{i,t} \,\big|\, S_{i,t}\right] = \mathbb{P}\left[Z_{i,t} \,\big|\, S_{i,t}\right], \\
\label{eq:recurse3}
&\mathbb{P}_{\pi_c} \left[Y_{i,t} \,\big|\, S_{i,t}, \,  Z_{i,t}, \, A_{i,t} = \ind{Z_{i,t}\ge c}\right]=\mathbb{P}\left[Y_{i,t} \,\big|\, S_{i,t}, \,  Z_{i,t}, \, A_{i,t} = \ind{Z_{i,t}\ge c}\right].
\end{align}
The second and third equalities then immediately imply that the $g$-formula \eqref{eq:g-formula}
holds; while the first is a useful stepping stone in establishing the recursion.

Our base case is to note that \eqref{eq:recurse1} holds with $t = 0$; this is
immediate from \eqref{eq:pomdp}. Next, we verify
that if \eqref{eq:recurse1} holds for a given $t$, then \eqref{eq:recurse2} must
also hold for the same $t$. Clearly
\begin{align*}
\mathbb{P}_{\pi_c}\left[Z_{i,t} \,\big|\, S_{i,t}\right] = 
\int \mathbb{P}_{\pi_c}\left[Z_{i,t} \,\big|\, S_{i,t}, \, U_{i,t}\right] \mathbb{P}_{\pi_c}\left[U_{i,t} \,\big|\, S_{i,t}\right] d\lambda(U_{i,t}),
\end{align*}
and $\mathbb{P}_{\pi_c}\left[Z_{i,t} \,\big|\, S_{i,t}, \, U_{i,t}\right] = \mathbb{P}\left[Z_{i,t} \,\big|\, S_{i,t}, \, U_{i,t}\right]$ by \eqref{eq:pomdp}, while the second term on the right is
policy independent by our recursion hypothesis on \eqref{eq:recurse1}.

Moving onwards, we verify that if \eqref{eq:recurse1} and \eqref{eq:recurse2} hold
for a given $t$, then \eqref{eq:recurse3} also holds for that $t$. We again expand
\begin{align*}
&\mathbb{P}_{\pi_c} \left[Y_{i,t} \,\big|\, S_{i,t}, \,  Z_{i,t}, \, A_{i,t} = \ind{Z_{i,t}\ge c}\right] \\
&\quad\quad\quad\quad  = \int \mathbb{P}_{\pi_c} \left[Y_{i,t} \,\big|\, S_{i,t}, \,  Z_{i,t}, \, A_{i,t} = \ind{Z_{i,t}\ge c},\, U_{i,t}\right] \\
&\quad\quad\quad\quad\quad\quad\quad\quad  \mathbb{P}_{\pi_c} \left[U_{i,t} \,\big|\, S_{i,t}, \,  Z_{i,t}, \, A_{i,t} = \ind{Z_{i,t}\ge c}\right]
d\lambda(U_{i,t}).
\end{align*}
The first of the right-hand-side terms is again immediately policy independent
by \eqref{eq:pomdp}. Meanwhile, for the second, we can use Bayes' rule to get
\begin{align*}
\mathbb{P}_{\pi_c} \left[U_{i,t} \,\big|\, S_{i,t}, \,  Z_{i,t}, \, A_{i,t} = \ind{Z_{i,t}\ge c}\right]
&= \frac{\mathbb{P}_{\pi_c} \left[U_{i,t} \,\big|\, S_{i,t} \right]\mathbb{P}_{\pi_c} \left[Z_{i,t} \,\big|\, S_{i,t}, \, U_{i,t} \right] \ind{ A_{i,t} = \ind{Z_{i,t}\ge c} }}
{\mathbb{P}_{\pi_c} \left[Z_{i,t} \,\big|\, S_{i,t} \right] \ind{ A_{i,t} = \ind{Z_{i,t}\ge c} }} \\
&= \frac{\mathbb{P}_{\pi_c} \left[U_{i,t} \,\big|\, S_{i,t} \right]\mathbb{P}_{\pi_c} \left[Z_{i,t} \,\big|\, S_{i,t}, \, U_{i,t} \right]}
{\mathbb{P}_{\pi_c} \left[Z_{i,t} \,\big|\, S_{i,t} \right]}.
\end{align*}
The two terms in the numerator are threshold independent by \eqref{eq:recurse1}
and the assumed decomposition \eqref{eq:pomdp} respectively, while the one in
the denominator is threshold independent by \eqref{eq:recurse2}.

It remains to verify that \eqref{eq:recurse1} holds for time $t$ when all of
(\ref{eq:recurse1}--\ref{eq:recurse3}) hold for time $t-1$. We expand
\begin{equation*}
\mathbb{P}_{\pi_c} \left[U_{i,t} \,\big|\, S_{i,t} \right]
= \int \mathbb{P}_{\pi_c} \left[U_{i,t} \,\big|\, S_{i,t}, \, U_{i,t-1} \right]
\mathbb{P}_{\pi_c} \left[U_{i,t-1} \,\big|\, S_{i,t} \right] d\lambda(U_{t-1}).
\end{equation*}
Again $\mathbb{P}_{\pi_c} \left[U_{i,t} \,\big|\, S_{i,t}, \, U_{i,t-1} \right]
= \mathbb{P} \left[U_{i,t} \,\big|\, S_{i,t}, \, U_{i,t-1} \right]$ thanks to
the Markovian structure implied by \eqref{eq:pomdp}. Meanwhile,
\begin{align*}
\mathbb{P}_{\pi_c} \left[U_{i,t-1} \,\big|\, S_{i,t} \right]
= \frac{\mathbb{P}_{\pi_c} \left[U_{i,t-1} \,\big|\, S_{i,t-1} \right]\mathbb{P}_{\pi_c} \left[S_{i,t} \,\big|\, S_{i,t-1},\, U_{i,t-1} \right]}{\mathbb{P}_{\pi_c} \left[S_{i,t} \,\big|\, S_{i,t-1} \right]}
\end{align*}
The first term in the numerator is policy independent by induction on \eqref{eq:recurse1}.
Meanwhile, the ratio $\mathbb{P}_{\pi_c} \left[S_{i,t} \,\big|\, S_{i,t-1},\, U_{i,t-1} \right] \,\big/\, \mathbb{P}_{\pi_c} \left[S_{i,t} \,\big|\, S_{i,t-1} \right]$ can further be simplified as
{\footnotesize
\begin{align*}
&\frac{\mathbb{P}_{\pi_c} \left[Z_{i,t-1} \,\big|\, S_{i,t-1},\, U_{i,t-1} \right] \ind{A_{i,t-1} = \ind{Z_{i,t-1} \geq c}} \mathbb{P}_{\pi_c} \left[Y_{i,t-1} \,\big|\, S_{i,t-1}, \, Z_{i,t-1}, \, A_{i,t-1} = \ind{Z_{i,t-1} \geq c},\, U_{i,t-1} \right]}{\mathbb{P}_{\pi_c} \left[Z_{i,t-1} \,\big|\, S_{i,t-1} \right] \ind{A_{i,t-1} = \ind{Z_{i,t-1} \geq c}} \mathbb{P}_{\pi_c} \left[Y_{i,t-1} \,\big|\, S_{i,t-1}, \, Z_{i,t-1}, \, A_{i,t-1} = \ind{Z_{i,t-1} \geq c} \right]} \\
&\quad= \frac{\mathbb{P}_{\pi_c} \left[Z_{i,t-1} \,\big|\, S_{i,t-1},\, U_{i,t-1} \right] \mathbb{P}_{\pi_c} \left[Y_{i,t-1} \,\big|\, S_{i,t-1}, \, Z_{i,t-1}, \, A_{i,t-1} = \ind{Z_{i,t-1} \geq c},\, U_{i,t-1} \right]}{\mathbb{P}_{\pi_c} \left[Z_{i,t-1} \,\big|\, S_{i,t-1} \right] \mathbb{P}_{\pi_c} \left[Y_{i,t-1} \,\big|\, S_{i,t-1}, \, Z_{i,t-1}, \, A_{i,t-1} = \ind{Z_{i,t-1} \geq c} \right]}
\end{align*}}

\noindent Finally, the two conditional expectations in the numerator are policy independent by  \eqref{eq:pomdp},
while those in the denominator are policy independent by recursing on \eqref{eq:recurse2} and
\eqref{eq:recurse3} respectively.
\end{proof}

\subsection{\texorpdfstring{Proof of \cref{thm:expression-for-policy-gradient}}{Policy gradient formula}}\label{proof-of-policy-gradient}

\begin{proof}%[Proof of \cref{thm:expression-for-policy-gradient}]
     Since units $i=1,\dots,n$ are identically distributed, we drop the unit index throughout this proof. Denote by $P_t^{\pi_c}$ the distribution of $S_t$ under the policy $\pi_c$. Define
     $$V^Y_{c,\,t}(P_t)=\E_{\pi_c}\left[\sum_{j=0}^{T-t} \gamma^j\,Y_{t+j}\,\,\bigg|\,\, S_t\sim P_t\right]=\E_{S_t\sim P_t} \int_\R Q_{c,\,t}^Y(S_t,\,z,\,\pi_c(z)) f_t(z\mid S_t)\,dz,$$
where $$Q_{c,\,t}^Y(s,\,z,\,a)=\E_{\pi_c}\left[\,\sum_{j=0}^{T-t}\gamma^j\,Y_{t+j}\,\,\bigg|\,\, S_t=s,\,Z_t=z,\,A_t=a\right].$$
 Assume without loss of generality that $h>0$ (the case $h<0$ follows similarly) and $h<\eta$ where $\eta$ is as in \eqref{eq:integrability}. Since $S_0=\{\emptyset\}$, we can write
 \begin{equation}\label{eq:trivial}
     \frac{V^Y(\pi_{c+h})-V^Y(\pi_c)}{h}=\frac{V^Y_{c+h,\,0}(P_0^{\pi_{c}})-V^Y_{c,\,0}(P_0^{\pi_{c}})}{h}.
 \end{equation}
 Observe now that for any distribution $P_t$ for $S_t$,
\begin{align*}
    \frac{V^Y_{c+h,\,t}(P_t)-V^Y_{c,\,t}(P_t)}{h}  
    &=\frac{1}{h}\,\E_{S_t\sim P_t}\int_\R \left\{Q_{c+h,\,t}^Y(S_t,\, z,\, \pi_{c+h}(z)) - Q_{c,\,t}^Y(S_t,\, z,\, \pi_{c}(z))\right\}f_t(z\mid S_t)dz\\[2mm]
    &=\frac{1}{h}\,\E_{S_t\sim P_t}\int_\R \left\{Q_{c,\,t}^Y(S_t,\, z,\, \pi_{c+h}(z)) - Q_{c,\,t}^Y(S_t,\, z,\, \pi_{c}(z))\right\}f_t(z\mid S_t)dz\\[2mm]
    &\quad+\frac{1}{h}\,\E_{S_t\sim P_t}\int_\R \left\{Q_{c+h,\,t}^Y(S_t,\, z,\, \pi_{c+h}(z)) - Q_{c,\,t}^Y(S_t,\, z,\, \pi_{c+h}(z))\right\}f_t(z\mid S_t)dz\\[2mm]
    &=-M_t(h; P_t)+R_t(h; P_t),\numberthis\label{policy-grad-finite-differencing}
\end{align*}
     where
$$M_t(h; P_t):=\E_{S_t\sim P_t}\left[\frac{1}{h}\int_c^{c+h}\left(Q_{c,\,t}^Y(S_t,\,z,\,1)-Q_{c,\,t}^Y(S_t,\,z,\,0)\right)f_t(z\mid S_t)\,dz\right],$$
and $$R_t(h; P_t):=\E_{S_t\sim P_t}\left[\frac{1}{h} \int_\R\left(Q_{c+h,\,t}^Y\left(S_t, z, \pi_{c+h}(z)\right)-Q_{c,\,t}^Y\left(S_t, z, \pi_{c+h}(z)\right)\right) f_t\left(z \mid S_t\right) dz\right].$$
      Next, define $\mu_t(s,\, z,\, a):=\E\left[Y_{t}\mid S_{t}=s, \, Z_{t}=z,\, A_t=a\right]$ and recall that the Bellman equation reads
$$Q_{c',\,t}^Y(s,\, z,\, a) = \mu_t(s,\, z,\, a)+\gamma V^Y_{c',\,t+1}(P_{t+1}^{\,s,\, z,\, a}),$$
for any $c'$, where $P_{t+1}^{\,s,\, z,\, a}$ is the law after one step transition from $(S_t=s,\, Z_t=z,\, A_t=a)$.
Consequently, 
\begin{equation}\label{bellman}
    Q_{c+h,\,t}^Y\left(s,\, z,\, a\right)-Q_{c,\,t}^Y\left(s,\, z,\, a\right)=\gamma\left(V^Y_{c+h,\,t+1}(P_{t+1}^{\,s,\, z,\, a})-V^Y_{c,\,t+1}(P_{t+1}^{\,s,\, z,\, a})\right).
\end{equation}
This combined with the definition of $R_0(h)$ yields that
\begin{align*}
  &  R_t(h; P_t^{\pi_c})\\
  &=\E_{S_t\sim P_t^{\pi_c}}\left[\frac{1}{h} \int_\R\left(Q_{c+h,\,t}^Y\left(S_t, z, \pi_{c+h}(z)\right)-Q_{c,\,t}^Y\left(S_t, z, \pi_{c+h}(z)\right)\right) f_t\left(z \mid S_t\right) dz\right]\tag{by definition}\\[2mm]
    &=\E_{S_t\sim P_t^{\pi_c}}\left[\frac{1}{h} \int_\R\left(Q_{c+h,\,t}^Y\left(S_t, z, \pi_{c}(z)\right)-Q_{c,\,t}^Y\left(S_t, z, \pi_{c}(z)\right)\right) f_t\left(z \mid S_t\right) dz\right] \\[2mm]
    &\quad-\E_{S_t\sim P_t^{\pi_c}}\left[\frac{1}{h} \int_c^{c+h}\left((Q_{c+h,\,t}^Y - Q_{c,\,t}^Y)(S_t, z, 1)-(Q_{c+h,\,t}^Y - Q_{c,\,t}^Y)(S_t, z, 0)\right) f_t(z \mid S_t) dz\right]\\[2mm]
    &=\E_{S_t\sim P_t^{\pi_c}}\left[\frac{1}{h} \int \gamma\left(V^Y_{c+h,\,t+1}\left(P_{t+1}^{\,S_t,\, z,\, \pi_{c}(z)}\right)-V^Y_{c,\,t+1}\left(P_{t+1}^{\,S_t,\, z,\, \pi_{c}(z)}\right)\right) f_t\left(z \mid S_t\right) dz\right]\tag{using \eqref{bellman}}\\[2mm]
    &\quad-\E_{S_t\sim P_t^{\pi_c}}\left[\frac{1}{h} \int_c^{c+h}\left((Q_{c+h,\,t}^Y - Q_{c,\,t}^Y)(S_t, z, 1)-(Q_{c+h,\,t}^Y - Q_{c,\,t}^Y)(S_t, z, 0)\right) f_t(z \mid S_t) dz\right]\\[2mm]
    &=\gamma\,\frac{V^Y_{c+h,\,t+1}(P_{t+1}^{\pi_{c}})-V^Y_{c,\,t+1}(P_{t+1}^{\pi_{c}})}{h}-\Delta_t(h),\tag{by definition}
\end{align*}
where 
\begin{align*}
    \Delta_t(h)&:=\E_{S_t\sim P_t^{\pi_c}}\left[\frac{1}{h} \int_c^{c+h}\left((Q_{c+h,\,t}^Y - Q_{c,\,t}^Y)(S_t, z, 1)-(Q_{c+h,\,t}^Y - Q_{c,\,t}^Y)(S_t, z, 0)\right) f_t(z \mid S_t) dz\right].
\end{align*}
Plugging the above expression for $R_t(h)$ into \eqref{policy-grad-finite-differencing}, we arrive at the following recursion.
\begin{equation*}
    \frac{V^Y_{c+h,\,t}(P_t^{\pi_c})-V^Y_{c,\,t}(P_t^{\pi_c})}{h}=-M_t(h)+\gamma\,\frac{V^Y_{c+h,\,t+1}(P_{t+1}^{\pi_c})-V^Y_{c,\,t+1}(P_{t+1}^{\pi_c})}{h}-\Delta_t(h).
\end{equation*}
Using the above recursion, we deduce from \eqref{eq:trivial} that
\begin{equation}\label{2.1.4}
   \frac{V^Y(\pi_{c+h})-V^Y(\pi_{c})}{h}=-\sum_{t=0}^{T'}\gamma^{t} \, (M_t(h)+\Delta_t(h))+\gamma^{T'+1}\frac{V^Y_{c+h, \, T'+1}(P_{T'+1}^{\pi_{c}})-V^Y_{c,\,T'+1}(P_{T'+1}^{\pi_c})}{h}. 
\end{equation}
For the finite-horizon ($T<\infty$) case, we use $T'=T$ and $V^Y_{c',\,T+1}(P_{T+1}^{\pi_c})=0$ to deduce that
\begin{equation}\label{2.1.5}
   \frac{V^Y(\pi_{c+h})-V^Y(\pi_{c})}{h}=-\sum_{t=0}^{T}\gamma^{t} \, (M_t(h)+\Delta_t(h)). 
\end{equation}
For the infinite-horizon ($T=\infty$) case, we claim that the above display also holds, in a limiting sense: We obtain \eqref{2.1.5} by letting $T'\to\infty$ in \eqref{2.1.4}. To prove this, derive a uniform upper bound on $(V^Y_{c+h,\,t}(P_t)-V^Y_{c,\,t}(P_t))/h$ so that the last term in \eqref{2.1.4} goes to $0$ as $T'\to\infty$, for any fixed $h>0$.
The following chain of inequalities fits the bill. For any $t\ge 0$ and any distribution $P_t$ for $S_t$,
\begin{align*}
   \left|\frac{V^Y_{c+h,\,t}(P_t)-V^Y_{c,\,t}(P_t)}{h}\right|&\le \frac{2}{h}\sup_{|c'-c|\le \eta}\left|V^Y_{c',\,t}(P_t)\right|\tag{since $|h|<\eta$}\\
    &\le\frac{2}{h}\sup_{|c'-c|\le \eta} \E_{S_t\sim P_t}\,\E_{\pi_{c'}}\left[\sum_{j=0}^{T-t} \gamma^j\left|Y_{t+j}\right|\,\bigg|\, S_t\right]\\
    &\le\frac{2}{h(1-\gamma)}\sup_{t\ge 0}\, \sup_{s_t}\sup_{|c'-c|\le \eta}\E_{\pi_{c'}}\left[\sum_{j=0}^{T-t}\gamma^j\left|Y_{t+j}\right|\,\bigg|\, S_t=s_t\right]<\infty.\numberthis\label{2.1.2last}
\end{align*}
The above bound combined with $\gamma^{T'+1}\to 0$ as $T'\to\infty$ yields \eqref{2.1.5} from  \eqref{2.1.4}.

\medskip

Having shown \eqref{2.1.5}, we now consider taking limits as $h\to 0$. Since $z\mapsto Q_{c,\,t}^Y(s,\, z,\, a)$ and $z\mapsto f_t(z\mid s)$ are continuous at $c$ for a.e.~$s$ and for every $t\ge 0$, and it follows from the fundamental theorem of Calculus that
$$\lim_{h\to 0}\left[\frac{1}{h}\int_c^{c+h} (Q_{c,\,t}^Y(s_t,\, z,\, 1)-Q_{c,\,t}^Y(s_t,\,z,\,0))f_t\left(z \mid s_t\right) dz\right]=\left(Q_{c,\,t}^Y(s_t,\,c,\,1)-Q_{c,\,t}^Y(s_t,\,c,\,0)\right)f_t(c\mid s_t),$$
for a.e.~$s_t$ and for every $t$. Since $\sup_t\sup_{s_t}\sup_{|z-c|\le \eta}|Q_{c,\,t}^Y(s_t,\,z,\,1)-Q_{c,\,t}^Y(s_t,\,z,\,0)|f_t(z\mid s_t)<\infty$, we can apply the dominated convergence theorem  to conclude that $$\lim_{h\to 0} M_t(h) =\E_{S_t\sim P_t^{\pi_c}} \E_{\pi_c}(Q_{c,\,t}^Y(S_t,\, c,\,1)-Q_{c,\,t}^Y(S_t,\, c,\,0))\,f_t(c\mid S_t),$$
for each $t\ge 0$. The geometric weights allow us to pass the limit under the sum and write the following.
\begin{equation}\label{2.1.5.5}
    \lim_{h\to0}\sum_{t=0}^T \gamma^t\, M_t(h)= \sum_{t=0}^T \gamma^t\,\E_{\pi_c} (Q_{c,\,t}^Y(S_t,\,c,\,1)-Q_{c,\,t}^Y(S_t,\,c,\,0))f_t(c\mid S_t).
\end{equation}
In view of \eqref{2.1.5}, it only remains to show that
\begin{equation}\label{2.1.6}
    \lim_{h\to 0}\,\sum_{t=0}^T \gamma^t\,\Delta_t(h) = 0.
\end{equation}
Toward proving the above display, we first note that 
\begin{align*}
& \left|R_t(h; P_t)\right|\\ & =\left|\E_{S_t\sim P_t^{\pi_c}}\left[\frac{1}{h} \int_\R\left(Q_{c+h,\,t}^Y\left(S_t, z, \pi_{c+h}(z)\right)-Q_{c,\,t}^Y\left(S_t, z, \pi_{c+h}(z)\right)\right) f_t\left(z \mid S_t\right) dz\right]\right|\tag{by definition}\\[2mm]
    &=\left|\E_{S_t\sim P_t^{\pi_c}}\left[\frac{1}{h} \int \gamma\left(V^Y_{c+h,\,t+1}\left(P_{t+1}^{\,S_t,\, z,\, \pi_{c+h}(z)}\right)-V^Y_{c,\,t+1}\left(P_{t+1}^{\,S_t,\, z,\, \pi_{c+h}(z)}\right)\right) f_t\left(z \mid S_t\right) dz\right]\right|\tag{using \eqref{bellman}}\\
    &\le \gamma\,\sup_{P_{t+1}}\left|\frac{V^Y_{c+h,\,t+1}(P_{t+1})-V^Y_{c,\,t+1}(P_{t+1})}{h}\right|\tag{by triangle inequality}\\
    &\le \gamma\,\sup_{t\ge 0}\sup_{P_{t}}\left|\frac{V^Y_{c+h,\,t}(P_{t})-V^Y_{c,\,t}(P_{t})}{h}\right|.
\end{align*}
Note that the above bound is finite, due to \eqref{2.1.2last}. Using the above bound and \eqref{policy-grad-finite-differencing}, we deduce that
$$
  \left| \frac{V^Y_{c+h,\,t}(P_t)-V^Y_{c,\,t}(P_t)}{h}\right|\le \left|M_t(h; P_t)\right|+\gamma\,\sup_{t\ge 0}\sup_{P_{t}}\left|\frac{V^Y_{c+h,\,t}(P_{t})-V^Y_{c,\,t}(P_{t})}{h}\right|
$$
This, in turn, implies that
\begin{align*}
&\sup_{t\ge 0}\,\sup_{P_t} \left|\frac{V^Y_{c+h}(P_t)-V^Y_{c}(P_t)}{h}\right| \\
&\le\frac{1}{1-\gamma}\sup_{t\ge 0}\,\sup_{P_t}|M_t(h\,;P_t)|\\
  &\le  \frac{1}{1-\gamma}\sup_{t\ge 0}\,\sup_{P_t}\E_{S_t\sim P_t}\left[\frac{1}{h}\int_c^{c+h}\left|Q_{c,\,t}^Y(S_t,\,z,\,1)-Q_{c,\,t}^Y(S_t,\,z,\,0)\right|f_t(z\mid S_t)\,dz\right]
  \\
  &\le  \frac{1}{1-\gamma}\sup_{t\ge 0}\,\sup_{s_t}\sup_{|z-c|\le  \eta}\left|Q_{c,\,t}^Y(s_t,\,z,\,1)-Q_{c,\,t}^Y(s_t,\,z,\,0)\right|f_t(z\mid s_t),
\end{align*}
where the last quantity is finite, by assumption. This implies that for all $|h|<\eta$, for any $t\ge 0$ and any distribution $P_t$ for $S_t$,
$$|V^Y_{c+h,\,t}(P_t)-V^Y_{c,\,t}(P_t)|\le Ch,$$ for some universal constant $C\in (0,\infty)$. We can therefore invoke \eqref{bellman} once again to conclude that
\begin{align*}
    &\left|\Delta_t(h)\right|\\
    &=  \left|  \E_{S_t\sim P_t^{\pi_c}}\left[\frac{1}{h} \int_c^{c+h}\left(\left(Q_{c+h,\,t}^Y - Q_{c,\,t}^Y\right)(S_t, z, 1)-\left(Q_{c+h,\,t}^Y - Q_{c,\,t}^Y\right)(S_t, z, 0)\right) f_t\left(z \mid S_t\right) dz\right]\right|\\[1mm]
    &\le  \left|  \E_{S_t\sim P_t^{\pi_c}}\left[\frac{1}{h} \int_c^{c+h}\gamma\left(\left(V^Y_{c+h,\,t+1}-V^Y_{c,\,t+1}\right)(P_{t+1}^{S_t,\, z,\, 1})-\left(V^Y_{c+h,\,t+1}-V^Y_{c,\,t+1}\right)(P_{t+1}^{S_t,\, z,\, 0})\right) f_t\left(z \mid S_t\right) dz\right]\right|\\[1mm]
    &\le \E_{S_t\sim P_t^{\pi_c}}\left[\frac{1}{h} \int_c^{c+h}\gamma\cdot 2Ch \cdot f_t\left(z \mid S_t\right) dz\right]\\[1mm]
    &=2C\gamma\, \P_{\pi_c}\left(Z_t\in [c,c+h]\right).
\end{align*}
Finally, the dominated convergence theorem gives
$$\lim_{h\to 0}\,\sum_{t=0}^T\gamma^t\, \P_{\pi_c}\left(Z_t\in [c,c+h]\right)= 0,$$
which completes the proof of \eqref{2.1.6}. Combining \eqref{2.1.5}, \eqref{2.1.5.5} and \eqref{2.1.6} we conclude that
$$-\frac{\partial}{\partial c}V^Y(\pi_c)=\E_{\pi_c}\left[\sum_{t=0}^T \gamma^t\, (Q_{c,\,t}^Y(S_t,\,c,\,1)-Q_{c,\,t}^Y(S_t,\,c,\,0))f_t\left(c \mid S_t\right)\right],$$
as desired to show. Finally, note that the conditions \eqref{eq:integrability} for $Y$ imply analogous conditions for $A$. Thus,
$$-\frac{\partial}{\partial c}V^A(\pi_c)=\E_{\pi_c}\left[\sum_{t=0}^T \gamma^t\, (Q_{c,\,t}^A(S_t,\,c,\,1)-Q_{c,\,t}^A(S_t,\,c,\,0))f_t\left(c \mid S_t\right)\right].$$
Moreover, \cref{assump:continuous-Q} ensures that the above right-hand side is non-zero. Combining the last two displays yields the desired characterization \eqref{eq:tauRD}.
 \end{proof}

\subsection{\texorpdfstring{Proof of \cref{consistency}}{consistency}}\label{proof:consistency}

\begin{proof}
    Recall that the twice-discounted local linear regression estimator is defined as $$\htauRD(h_n)= \wh\tau^Y(h_n)\,\big/\,\wh\tau^A(h_n),$$ where $\wh\tau^Y(h)$ and $\wh\tau^A(h)$ are defined in \eqref{eq:def-LLR}. Define the corresponding population-level optimization problems as follows.
      \begin{equation}\label{eq:def-LLR-popln}
      \begin{split}
         & \tau^R_n :=e_1^\top \argmin_{(\tau,\alpha_t,\beta_0,\beta_1)}\, \E_{\pi_c}\left[\sum_{t=0}^{T_n} w_{i,t}(h_n)\left(\Gamma_{i,t}^R -\tau A_{i,t}-\alpha_t-\beta_0(Z_{i,t}-c)-\beta_1 A_{i,t}(Z_{i,t}-c)\right)^2\right], \\[2mm]
&\text{where } R\in\{Y,A\},\ \ \text{and}\ \ w_{i,t}(h) := \gamma^t \,K \left(|Z_{i,t}-c|/h\right).
      \end{split}   
  \end{equation}
  We establish consistency through two key lemmas. First, \cref{lemma:localization} demonstrates via standard analytic arguments that $\tau^R_n$ converges to $\tau^R$, where for $R\in\{Y,A\}$,
  \begin{equation}\label{eq:def-taustar}
      \begin{split}
          &\tau^R :=-\gradc V^R(\pi_c)=\frac{\E_{\pi_c}\left[\sum_{t=0}^T \gamma^t \left(Q_{c,\,t}^R(S_{i,t},\,c,\,1)- Q_{c,\,t}^R(S_{i,t},\,c,\,0)\right)f_t(c\mid S_{i,t})\right]}{\E_{\pi_c}\left[\sum_{t=0}^T \gamma^t  f_t(c\mid S_{i,t})\right]}.
      \end{split}
  \end{equation}
  Second, \cref{lemma:concentration} employs a concentration argument to show that $\wh\tau^R(h_n) -\tau^R_n=\o(1)$ ($R=Y,A$). Combining this with $\tau^R_n-\tau^R=o(1)$, we conclude that $$\wh\tau^Y(h_n) -\tau^Y=\o(1),\quad\text{and}\quad \wh\tau^A(h_n) -\tau^A=\o(1).$$ Applying the continuous mapping theorem, we are through.
\end{proof}

\subsection{\texorpdfstring{Proof of \cref{clt-for-twice-discounted-llr}}{consistency}}\label{proof-of-clt-for-llr}

\begin{proof}
    We show in \cref{bivariate-clt-for-twice-discounted-llr} that under the conditions of \cref{clt-for-twice-discounted-llr}, the numerator $\wh{\tau}^Y_n$ and the denominator $\wh{\tau}^A_n$ of the twice-discounted local linear regression \eqref{eq:def-LLR} satisfy the following.
    \begin{equation}\label{eq:multi-clt}
        \sqrt{n h_n}
\left(\begin{matrix}\wh{\tau}^Y_n-\Delta\mu_Y(c)
-\frac{1}{2}h_n^2\,\xi_1\Delta\mu_Y''(c)\\[2mm]
\wh{\tau}^A_n-\Delta\mu_A(c)
-\frac{1}{2}h_n^2\,\xi_1\Delta\mu_A''(c)
\end{matrix}\right) \stackrel{d}{\longrightarrow} \mathcal{N}\left(\begin{pmatrix}
    0 \\ 0
\end{pmatrix}, \begin{pmatrix}
V_{YY}   & V_{YA}\\V_{YA} & V_{AA}
\end{pmatrix}\right),
    \end{equation}
where $\xi_1:=(\kappa_2^2-\kappa_1\kappa_3)/(\kappa_0\kappa_2-\kappa_1^2)$ with $\kappa_j:=\int_0^1 u^j K(u)du$, and the quantities $\Delta\mu_Y$, $\Delta\mu_A$, $V_{YY}$, $V_{AA}$ and $V_{YA}$ are as defined in \cref{clt-for-twice-discounted-llr}. Using this notation we can write from \eqref{eq:def-taustar} that $\tau^Y=\Delta\mu_Y(c)$ and $\tau^A=\Delta\mu_A(c)$, and thus \eqref{eq:tauRD} gives $\tauRD=\Delta\mu_Y(c)/\Delta\mu_A(c)$. Now applying \cref{lemma:multiv-delta} (delta-method) to \eqref{eq:multi-clt}, we conclude that $\htauRD = \wh{\tau}^Y_n/\wh{\tau}^A_n$ satisfies
\begin{equation*}%\label{eq:asymp-dist-of-ratio-finite-horizon}
    \sqrt{nh_n}\left(\htauRD -\tauRD-\frac{1}{2}h_n^2\xi_1\frac{\Delta\mu_Y''(c)-\tauRD\Delta\mu_A''(c)}{\Delta\mu_A(c)}\right)\dto \normal\left(0\,,\frac{V_{YY} + \tauRD^2 V_{AA} - 2\tauRD V_{YA}}{(\Delta\mu_A(c))^2}\right),
\end{equation*}
which completes the proof.
\end{proof}

\subsection{\texorpdfstring{Proof of \cref{propo:sandwich-estimator}}{variance estimator}}\label{proof:propo:sandwich-estimator}

\begin{proof} It follows from \cref{lem:fe-representation} that,
 for any \(G,H\in\{Y,A\}\), we can rewrite $V_{GH,\,n}$ as
\[
  \wh V_{GH,\,n}:=\frac1n\sum_{i=1}^n\wh\psi_{n,i}^G\,\wh\psi_{n,i}^H,
\]
where
\begin{align*}
    \wh\psi_{n,i}^R
&:=\sqrt{h_n}\,e_1^\top\wh M_n^{-1}
\sum_{t=0}^{T_n}\gamma^tK_{h_n}(Z_{i,t})
\{X_{i,t}(h_n)-\overline X_t(h_n)\}
\{\Gamma_{i,t}^R-\wh\alpha_{n,t}^R-(\wh\zeta_n^R)^\top X_{i,t}(h_n)\},\\
    \wh M_n&:=\frac1n\sum_{i=1}^n\sum_{t=0}^{T_n}\gamma^tK_{h_n}(Z_{i,t})
\{X_{i,t}(h_n)-\overline X_t(h_n)\}^{\otimes2}, \qquad \overline X_t(h_n):=
\frac{n^{-1}\sum_{i=1}^nK_{h_n}(Z_{i,t})X_{i,t}(h_n)}
     {n^{-1}\sum_{i=1}^nK_{h_n}(Z_{i,t})},\\
     \wh\zeta_n^R
&:=\wh M_n^{-1}\left\{
\frac1n\sum_{i=1}^n\sum_{t=0}^{T_n}\gamma^tK_{h_n}(Z_{i,t})
\{X_{i,t}(h_n)-\overline X_t(h_n)\}\Gamma_{i,t}^R
\right\},\\
\wh\alpha_{n,t}^R
&:=
\frac{\sum_{i=1}^nK_{h_n}(Z_{i,t})\{\Gamma_{i,t}^R-(\wh\zeta_n^R)^\top X_{i,t}(h_n)\}}
     {\sum_{i=1}^nK_{h_n}(Z_{i,t})},
\end{align*}
Fix \(G,H\in\{Y,A\}\).  We first prove consistency for the oracle scores \(\psi_{n,i}^R\), and then show that replacing them by the feasible regression scores \(\wh\psi_{n,i}^R\) changes the cluster covariance by \(\o(1)\).
Recall from the proof of \cref{bivariate-clt-for-twice-discounted-llr} that
\[
  \E_{\pi_c}\{\psi_{n,i}^G\,\psi_{n,i}^H\}\to V_{GH}.    
\]
The Lindeberg condition shown in the proof of \cref{bivariate-clt-for-twice-discounted-llr} implies uniform integrability of \(\{(\psi_{n,i}^Y)^2:n\geq1\}\) and \(\{(\psi_{n,i}^A)^2:n\geq1\}\).  Since
\(
  \abs{\psi_{n,i}^G\,\psi_{n,i}^H}
  \leq \frac12\{(\psi_{n,i}^G)^2+(\psi_{n,i}^H)^2\}
\),
the products \(\psi_{n,i}^G\,\psi_{n,i}^H\) are also uniformly integrable.  Therefore we can apply a triangular-array law of large numbers (\cref{triangular-array-lln}) to conclude that
\[
  \frac1n\sum_{i=1}^n\psi_{n,i}^G\,\psi_{n,i}^H
  -\E_{\pi_c}\{\psi_{n,i}^G\,\psi_{n,i}^H\}\,\Pto\,0, 
\]
and thus
\begin{equation}
    \label{eq:sw-32}
\frac1n\sum_{i=1}^n\psi_{n,i}^G\,\psi_{n,i}^H\,\Pto\, V_{GH}.
\end{equation}

It only remains to show that
\begin{equation}
    \frac1n\sum_{i=1}^n(\wh\psi_{n,i}^R-\psi_{n,i}^R)^2\,\Pto\,0,
  \qquad R\in\{Y,A\}. \label{eq:sw-33}
\end{equation}
We prove this for a fixed \(R\).  \cref{lemma:concentration} and the proof of \cref{lem:asymp_linearity} give
\[
  \wh M_n-M_n=\o(1),
  \qquad
  \wh M_n^{-1}-M_n^{-1}=\o(1),
  \qquad
  \wh\zeta_n^R-\zeta_n^R=\o(1).                             
\]
It follows from \cref{eq:uncentered-cross-conc,eq:uncentered-design-conc,eq:fixed-mass,eq:fixed-first,eq:fixed-outcome,eq:fixed-x-outcome,eq:fixed-second,eq:mean-quotient,eq:fixed-L-design,eq:fixed-L-cross} in \cref{lemma:concentration}, after summing over \(t\) with the discount weights, that
\begin{equation}    
\sum_{t=0}^{T_n}\gamma^t\norm{\overline X_t(h_n)-\mu_t(h_n)}=\O((nh_n)^{-1/2}), \label{eq:sw-35}
\end{equation}
\begin{equation}
    \sum_{t=0}^{T_n}\gamma^t\abs{\wh\alpha_{n,t}^R-\alpha_{n,t}^R}=\O((nh_n)^{-1/2})+\o(1)\norm{\wh\zeta_n^R-\zeta_n^R}. \label{eq:sw-36}
\end{equation}
%For \eqref{eq:sw-35}, note that for each fixed \(t\), the numerator defining \(\overline X_t(h_n)-\mu_t(h_n)\) is an empirical average of mean-zero terms of the form \(K_{h_n}(Z_{i,t})\{X_{i,t}(h_n)-\mu_t(h_n)\}\), whose second moment is \(O(h_n^{-1})\) by the kernel localization bound.  Hence it is \(\O((nh_n)^{-1/2})\).  The denominator is bounded away from zero when \(\E_{\pi_c}\{f_t(c\mid S_t)\}>0\), by (7); when this limit is zero, the weighted contribution of that time point is negligible by the same localization bound.  The proof for \(\wh\alpha_{n,t}^R-\alpha_{n,t}^R\) is identical after replacing \(X_{i,t}(h_n)\) by the residual \(\Gamma_{i,t}^R-(\zeta_n^R)^\top X_{i,t}(h_n)\), and then adding the term caused by replacing \(\zeta_n^R\) with \(\wh\zeta_n^R\).  The passage from fixed \(t\) to the discounted sum is justified by the summability envelopes in Assumptions 6 and 7, as in \cref{lemma:concentration}.

Define the population residual
\[
  \varepsilon_{i,t,n}^R
  :=\Gamma_{i,t}^R-\alpha_{n,t}^R-(\zeta_n^R)^\top X_{i,t}(h_n).
\]
The covariance calculation in \cref{bivariate-clt-for-twice-discounted-llr}, together with the triangular-array LLN (\cref{triangular-array-lln}), gives
\[
\frac1n\sum_{i=1}^n
h_n\left\|
\sum_{t=0}^{T_n}\gamma^tK_{h_n}(Z_{i,t})
\{X_{i,t}(h_n)-\mu_t(h_n)\}\varepsilon_{i,t,n}^R
\right\|^2
=\O(1).     
\]
Now decompose \(\wh\psi_{n,i}^R-\psi_{n,i}^R\) according to the four replacements: \(M_n^{-1}\) by \(\wh M_n^{-1}\), \(\mu_t(h_n)\) by \(\overline X_t(h_n)\), \(\alpha_{n,t}^R\) by \(\wh\alpha_{n,t}^R\), and \(\zeta_n^R\) by \(\wh\zeta_n^R\).  The first replacement contributes at most
\[
\norm{\wh M_n^{-1}-M_n^{-1}}^2
\frac1n\sum_{i=1}^n
h_n\left\|
\sum_{t=0}^{T_n}\gamma^tK_{h_n}(Z_{i,t})
\{X_{i,t}(h_n)-\mu_t(h_n)\}\varepsilon_{i,t,n}^R
\right\|^2
=\o(1)                                                      
\]
using the last display.  For the replacement of \(\mu_t(h_n)\), Cauchy--Schwarz gives
\begin{align*}
&\frac1n\sum_{i=1}^n h_n
\left\|
\sum_{t=0}^{T_n}\gamma^tK_{h_n}(Z_{i,t})
\{\overline X_t(h_n)-\mu_t(h_n)\}\varepsilon_{i,t,n}^R
\right\|^2                                                        \notag\\
&\quad\leq
C\left(
\sum_{t=0}^{T_n}\gamma^t\norm{\overline X_t(h_n)-\mu_t(h_n)}
\left\{\frac1n\sum_{i=1}^nh_nK_{h_n}^2(Z_{i,t})(\varepsilon_{i,t,n}^R)^2\right\}^{1/2}
\right)^2
=\o(1),     
\end{align*}
where the empirical averages in braces are \(\O(1)\) by the diagonal part of the covariance calculation, and \eqref{eq:sw-35} gives the vanishing factor.  The replacement of \(\alpha_{n,t}^R\) is analogous:
\begin{align*}
&\frac1n\sum_{i=1}^n h_n
\left\|
\sum_{t=0}^{T_n}\gamma^tK_{h_n}(Z_{i,t})
\{X_{i,t}(h_n)-\mu_t(h_n)\}
\{\wh\alpha_{n,t}^R-\alpha_{n,t}^R\}
\right\|^2                                                        \notag\\
&\quad\leq
C\left(
\sum_{t=0}^{T_n}\gamma^t\abs{\wh\alpha_{n,t}^R-\alpha_{n,t}^R}
\left\{\frac1n\sum_{i=1}^nh_nK_{h_n}^2(Z_{i,t})\right\}^{1/2}
\right)^2
=\o(1),     
\end{align*}
using \eqref{eq:sw-36} and the kernel localization bound.  Finally, the replacement of \(\zeta_n^R\) contributes
\begin{align*}
&\frac1n\sum_{i=1}^n h_n
\left\|
\sum_{t=0}^{T_n}\gamma^tK_{h_n}(Z_{i,t})
\{X_{i,t}(h_n)-\mu_t(h_n)\}
\{(\wh\zeta_n^R-\zeta_n^R)^\top X_{i,t}(h_n)\}
\right\|^2                                                        \notag\\
&\quad\leq
\norm{\wh\zeta_n^R-\zeta_n^R}^2
\frac1n\sum_{i=1}^n h_n
\left(
\sum_{t=0}^{T_n}\gamma^tK_{h_n}(Z_{i,t})C
\right)^2
=\o(1),     
\end{align*}
because \(\wh\zeta_n^R-\zeta_n^R=\o(1)\) and the empirical average is \(\O(1)\).  This completes the proof of \eqref{eq:sw-33}.

\medskip

Finally,
\begin{align*}
\left|\frac1n\sum_{i=1}^n\wh\psi_{n,i}^G\,\wh\psi_{n,i}^H
      -\frac1n\sum_{i=1}^n\psi_{n,i}^G\,\psi_{n,i}^H\right| &\leq
\left(\frac1n\sum_{i=1}^n(\wh\psi_{n,i}^G-\psi_{n,i}^G)^2\right)^{1/2}
\left(\frac1n\sum_{i=1}^n(\wh\psi_{n,i}^H)^2\right)^{1/2} \\
&\qquad+
\left(\frac1n\sum_{i=1}^n(\wh\psi_{n,i}^H-\psi_{n,i}^H)^2\right)^{1/2}
\left(\frac1n\sum_{i=1}^n(\psi_{n,i}^G)^2\right)^{1/2}
=\o(1),
\end{align*}
where \eqref{eq:sw-32} and \eqref{eq:sw-33} imply that the second factors are \(\O(1)\).  Combining this display with \eqref{eq:sw-32} gives \(\wh V_{GH,\,n}\,\Pto\, V_{GH}\).
The consistency of \(\wh V_{{\rm RD},n}\) follows from \(\wh\tau_{\rm RD}\,\Pto\,\tau_{\rm RD}\), \(\wh\tau_n^A\,\Pto\,\tau^A=\Delta\mu_A(c)>0\), the three covariance-consistency statements shown above, and the continuous mapping theorem.
\end{proof}

\subsection{Proof of \texorpdfstring{\cref{prop:dynamic-bandwidth-consistency}}{IK-style bandwidth}}
\begin{proof} 
We first justify the conditions the pilot bandwidths need to satisfy; the point that needs care is the passage from the iid static calculation in \citet{imbens2012optimal} to the dynamic stacked sample.  Let $h\to0$ and $nh\to\infty$.  Any localized empirical moment used in the local linear or local quadratic regressions in \cref{proc:dynamic-ik-bandwidth} can be written, after the scaling in \eqref{eq:def-scaling}, as 
\begin{equation}\label{B.20}
    \frac1n\sum_{i=1}^n
    \sum_{t=0}^{T_n}\gamma^t\, K_h(Z_{i,t})\,
    q\{U_{i,t}(h),A_{i,t}\}\,H_{i,t},
    \tag{B.20}
\end{equation}
where $U_{i,t}(h)=(Z_{i,t}-c)/h$, $q$ is one of finitely many fixed polynomials and $H_{i,t}$ is one of $1$, $\Gamma^Y_{i,t}$, $\Gamma^A_{i,t}$, or $\widetilde\Gamma_{i,t}$.  The summands in \eqref{B.20}, viewed as functions of the whole trajectory of unit $i$, are independent across $i$. \cref{lem:kernel-localization} gives the population localization limits.  Moreover, the concentration argument used in \cref{lemma:concentration} gives
\begin{equation}
    \frac1n\sum_{i=1}^n
    \sum_{t=0}^{T_n}\gamma^t\, K_h(Z_{i,t})\,
    q\{U_{i,t}(h),A_{i,t}\}\,H_{i,t}
    -
    \E_{\pi_c}\left[
    \sum_{t=0}^{T_n}\gamma^t\, K_h(Z_t)
    q\{U_t(h),A_t\}\,H_t
    \right]
    =\O\{(nh)^{-1/2}\}.
    \label{B.21}
\end{equation}
%Indeed, the variance of each trajectory-level summand is of order $h^{-1}$: conditioning on $S_t$, changing variables $z=c+hu$, and using the envelopes in Assumptions 5--6 gives the diagonal terms, and the cross-time terms are summable exactly as in (C.16)--(C.17).
This explains why every occurrence of the iid sample size $N$ in the IK calculation will be replaced by the number of independent trajectories $n$; the stacked count $n(T_n+1)$ never enters the stochastic order because the time series within each unit is part of a single cluster-level summand. Using \cref{lem:kernel-localization} with outcome equal to one, and using \eqref{B.21} at $h=h_1$, the density estimate in Step (iii) of \cref{proc:dynamic-ik-bandwidth} satisfies
\begin{equation}
    \widehat f_\gamma(c)
    =
    \frac{
    \sum_{i=1}^n\sum_{t=0}^{T_n}\gamma^t\ind{|Z_{i,t}-c|\le h_1}
    }{
    2h_1\sum_{i=1}^n\sum_{t=0}^{T_n}\gamma^t
    }
   \, \Pto\,
    \frac{
    \sum_{t=0}^{T}\gamma^t\E_{\pi_c}[f_t(c\mid S_t)]
    }{
    \sum_{t=0}^{T}\gamma^t
    }.                                    \label{B.23}
\end{equation}
The limit is finite by the local density envelope and strictly positive by \cref{assump:condtional-density}.  The same localization and concentration argument, applied to the squared residuals from the side-specific means used in the IK variance pilot, gives that the local variance estimate used in Step (iii) of \cref{proc:dynamic-ik-bandwidth} is $\O(1)$ and bounded away from zero in probability. Next, consider the median-window cubic regression. Similar to in \citet{imbens2012optimal}, we do not require the coefficient from this regression to consistently estimate the third derivative at $c$; we only need it to be $\O(1)$. This follows by the ordinary law of large numbers applied to the independent trajectory-level sums
\[
    \sum_{t=0}^{T_n}\gamma^t\,
    \ind{\widehat q_-\le Z_{i,t}\le \widehat q_+}
    r_3(Z_{i,t})r_3(Z_{i,t})^\top,
    \qquad
    \sum_{t=0}^{T_n}\gamma^t\,
    \ind{\widehat q_-\le Z_{i,t}\le \widehat q_+}
    r_3(Z_{i,t})\widetilde\Gamma_{i,t},
\]
where $r_3(z)=(1,\ind{z\ge c},z-c,(z-c)^2/2,(z-c)^3/6)^\top$, and  $\widehat q_-$, $\widehat q_+$ respectively denote the left- and right-side medians used in the global cubic regression. Finally, because \cref{proc:dynamic-ik-bandwidth} uses the stabilizer $\max\{\widehat m_{3,\gamma}^2,0.01\}$, the factor multiplying the side sample sizes in the definitions of $h_{2,+}$ and $h_{2,-}$ is bounded away from zero and infinity in probability.  Moreover, by the trajectory-level law of large numbers,
\[
    \frac{n_+}{n}	\;\Pto\;
    \frac{\sum_{t=0}^{T}\gamma^t\,\P_{\pi_c}(Z_t\ge c)}{\sum_{t=0}^{T}\gamma^t},
    \qquad
    \frac{n_-}{n}	\;\Pto\;
    \frac{\sum_{t=0}^{T}\gamma^t\,\P_{\pi_c}(Z_t<c)}{\sum_{t=0}^{T}\gamma^t},
\]
and both limits are positive by \cref{assump:condtional-density}. Consequently, 
the one-sided curvature-pilot bandwidths in Step (iii) of \cref{proc:dynamic-ik-bandwidth} satisfy
\begin{equation}
    \label{pilotcurvbw}
    h_{2,+}\to0,\qquad h_{2,-}\to0,\qquad
    nh_{2,+}^5\to\infty,\qquad nh_{2,-}^5\to\infty,
\end{equation}
with probability tending to one. Similar arguments combined with the sharp-RD case of \citet{imbens2012optimal} tell us that the pilot bandwidth $h_{\mathrm{pilot}}$ satisfies $h_\mathrm{pilot}\to 0$ and $nh_\mathrm{pilot}\to \infty$.

We now return to the main proof. For the residualized future outcome $\wt\Gamma_{i,t}=\Gamma^Y_{i,t}-\tauRD\Gamma^A_{i,t}$, the density-weighted response functions satisfy
\[
    \wt\mu_a(z)=\mu_{Y,a}(z)-\tauRD\mu_{A,a}(z),\qquad a\in\{0,1\},
\]
because the weights and normalizing denominator in \eqref{def-muRz} are common to $R=Y$ and $R=A$. Hence
\[
    \Delta\wt\mu''(c)
    =
    \Delta\mu_Y''(c)-\tauRD\Delta\mu_A''(c).
\]
The pilot ratio is consistent by \cref{consistency}, because $h_{\mathrm{pilot}}\to0$ and $nh_{\mathrm{pilot}}\to\infty$. It is therefore straightforward to verify that replacing $\tauRD$ by $\wh\tau_{\mathrm{pilot}}$ in the residualized outcome is asymptotically equivalent. The unit-clustered covariance estimates entering $\widehat V_{\mathrm{dyn}}$ are consistent by \cref{propo:sandwich-estimator}, evaluated at the pilot bandwidth. It remains to argue why the coefficient from the local quadratic regression consistently estimates the jump in the second-derivative. The standard boundary local-polynomial results from \citet{fan1996framework} imply that the bias in estimating the second-derivative jump is $O_p(h_{2,+}+h_{2,-})$ and its variance is $O_p((nh_{2,+}^5)^{-1}+(nh_{2,-}^5)^{-1})$. Hence the pilot bandwidth conditions  \eqref{pilotcurvbw} imply that
\[
    \widehat B_{\mathrm{dyn}}\,\Pto\, \Delta\mu_Y''(c)-\tauRD\Delta\mu_A''(c).
\]
The continuous mapping theorem then gives $\wh h_{\,\mathrm{IK},\mathrm{dyn}}/h_{\,\mathrm{IK},\mathrm{dyn}}^*\,\Pto\,1$.
\end{proof}

\section{Supporting Results}

Throughout, sums of the form $\sum_{t=0}^{T}$ are interpreted as finite sums when $T<\infty$ and as infinite sums when $T=\infty$. We denote by $T_n$ the observed horizon: In the finite-horizon case, $T_n=T$.  In the infinite-horizon case, we assume that $T_n\to\infty$ and $\gamma<1$. 

\begin{definition}
    Define scaled versions of the kernel weights, the running variable and the regressors as follows.
\begin{equation}\label{eq:def-scaling}
      K_h(z):=\frac{1}{h}K\left(\frac{|z-c|}{h}\right),\qquad
  U_{i,t}(h):=\frac{Z_{i,t}-c}{h},\qquad
  X_{i,t}(h):=\big(A_{i,t},\,U_{i,t}(h),\,A_{i,t}U_{i,t}(h)\big)^\top.
\end{equation}
\end{definition}

%The coefficient on $A_{i,t}$ is unaffected by replacing $Z_{i,t}-c$ by $U_{i,t}(h)$ and by multiplying all kernel weights by $h^{-1}$. 

\begin{lemma}[Kernel localization with discounting]\label{lem:kernel-localization}
Suppose that the conditions of \cref{consistency} hold true.  For every fixed integer $j\ge0$ and each $R\in\{Y,A\}$,
\begin{align*}
&\E_{\pi_c}\!\left[\sum_{t=0}^{T_n}\gamma^t\,
K_{h_n}(Z_t)\,\ind{Z_t\ge c}\left(U_t(h_n)\right)^j\Gamma^R_t
\right] \to
\kappa_j\sum_{t=0}^{T}\gamma^t\,
\E_{\pi_c}\!\left[Q^R_{c,t}(S_t,c,1)f_t(c\mid S_t)\right],
\end{align*}
while
\begin{align*}
&\E_{\pi_c}\!\left[\sum_{t=0}^{T_n}\gamma^t\,
K_{h_n}(Z_t)\,\ind{Z_t<c}\left(U_t(h_n)\right)^j\Gamma^R_t
\right] \to
(-1)^j\kappa_j\sum_{t=0}^{T}\gamma^t\,
\E_{\pi_c}\!\left[Q^R_{c,t}(S_t,c,0)f_t(c\mid S_t)\right],
\end{align*}
where $\kappa_j:=\int_0^1 u^j \,K(u)\,du$.
The same two conclusions with $\Gamma_t^R$ replaced by $1$ hold after replacing $Q^R_{c,t}(S_t,c,a)$ by $1$; in particular,
\[
  \E_{\pi_c}\!\left[\sum_{t=0}^{T_n}\gamma^t\,
  K_{h_n}(Z_t)\,\ind{Z_t\ge c}\left(U_t(h_n)\right)^j\right]\to
  \kappa_j\sum_{t=0}^{T}\gamma^t\,\E_{\pi_c}\!\left[f_t(c\mid S_t)\right],
\]
and the left-side analogue has limit $(-1)^j\kappa_j\sum_{t=0}^{T}\gamma^t\,\E_{\pi_c}\!\left[f_t(c\mid S_t)\right]$.
\end{lemma}

\begin{proof}
We prove the first display; the second follows by the same argument with $z=c-h_nu$.  
Let $\Gamma^R_{t,\infty}$ denote the corresponding full future discounted sum, i.e., $\Gamma^R_{t,\infty}=\Gamma^R_t$ when $T<\infty$ and $\Gamma^R_{t,\infty}=\sum_{s=0}^{\infty}\gamma^s R_{t+s}$ when $T=\infty$.  Thus $$Q^R_{c,t}(s,z,a)=\E_{\pi_c}[\Gamma^R_{t,\infty}\mid S_t=s,Z_t=z,A_t=a].$$  
We first handle the truncation error---which arises only in the infinite-horizon case---from replacing $\Gamma^R_t$ with $\Gamma^R_{t,\infty}$. The goal is to show that
\begin{align}
&\E_{\pi_c}\!\left[\sum_{t=0}^{T_n}\gamma^t\,
K_{h_n}(Z_t)\,\ind{Z_t\ge c}(U_t(h_n))^j
\left(\Gamma^R_{t,\infty}-\Gamma^R_t\right)
\right]\to0 . \label{eq:truncation-tail}
\end{align}
\begin{proof}[Proof of \eqref{eq:truncation-tail}]
The conclusion is trivial in the finite-horizon case, since $\Gamma^R_{t,\infty}=\Gamma^R_t$. Consider now the infinite-horizon case, and thus $\gamma<1$. Note that when $K_{h_n}(Z_t)\neq 0$, we have $|U_t(h_n)|\le 1$. Thus, $|K_{h_n}(Z_t)\,(U_t(h_n))^j|\le \|K\|_\infty /h_n$. On the other hand, $$\Gamma^R_{t,\infty}-\Gamma^R_t=\sum_{j=0}^\infty \gamma^j\, R_{t+j}-\sum_{j=0}^{T_n-t} \gamma^j\, R_{t+j}=\sum_{j=T_n-t+1}^\infty \gamma^j\, R_{t+j}=\gamma^{T_n-t+1}\,\Gamma^R_{T_n+1,\infty}.$$
Therefore,
\begin{align*}
&\left|\E_{\pi_c}\!\left[\sum_{t=0}^{T_n}\gamma^t\,
K_{h_n}(Z_t)\,\ind{Z_t\ge c}(U_t(h_n))^j
\left(\Gamma^R_t-\Gamma^R_{t,\infty}\right)
\right]\right|\\
&\le \E_{\pi_c}\!\left[\sum_{t=0}^{T_n}\gamma^t\,\frac{\|K\|_\infty}{h_n}\,\gamma^{T_n-t+1}\left|\Gamma^R_{T_n+1,\infty}
\right|\right]=\frac{\|K\|_\infty}{h_n}(T_n+1)\gamma^{T_n+1}\E_{\pi_c}|\Gamma^R_{T_n+1,\infty}|\\
&\le \frac{\|K\|_\infty}{h_n}(T_n+1)\gamma^{T_n+1}\frac{1}{1-\gamma}\sup_{t\ge 0}\E_{\pi_c}|Y_t|.
\end{align*}
The above goes to zero as $n\to\infty$, since $(T_n+1)\gamma^{T_n+1}=o(h_n)$ follows from the assumption that $T_n\ge (1+\eps)(\log h_n^{-1})/(\log \gamma^{-1})$ for all large $n$.
\end{proof}

Returning to the main proof, it remains to prove that the localization result the full future sum $\Gamma^R_t$.  Conditional on $S_t$, using $A_t=1$ on $\{Z_t\ge c\}$ and the change of variables $z=c+h_nu$ gives
\begin{align*}
&\E_{\pi_c}\!\left[
K_{h_n}(Z_t)\,\ind{Z_t\ge c}(U_t(h_n))^j\,\Gamma^R_{t,\infty}
\right] \\
&\quad=
\E_{\pi_c}\!\left[
\int_0^1 u^jK(u)\,Q^R_{c,t}(S_t,c+h_nu,1) \,f_t(c+h_nu\mid S_t)du
\right].
\end{align*}
For each fixed $t$, the integrand converges a.e.~to
$u^jK(u)\,Q^R_{c,t}(S_t,c,1)\,f_t(c\mid S_t)$ by the continuity of $z\mapsto Q^R_{c,t}(S_t,\,z,\,1)$ and $z\mapsto f_t(z\mid S_t)$ at $z=c$.  The envelope conditions $\sup_{|z-c|\le \eta} f_t(z\mid S_t)\vee 1\le B_{f,\,t}(S_t)$ and $\sup_{|z-c|\le\eta} m_{2,\,t}^Y(s,\, z,\,a)\vee 1\le B_{m_2,\,t}(S_t)$ imply that, for all large $n$ (large enough so that $|h_n|\le\eta$),
\[
\sup_{|u|\le 1}\left|Q^R_{c,t}(S_t,c+h_nu,1) f_t(c+h_nu\mid S_t)\right|
\le C B_{m_2,\,t}(S_t)\,B_{f,\,t}(S_t),
\]
with the same interpretation for $R=A$ (since $B_{m_2,\,t}\ge 1$). Therefore we can apply the dominated convergence theorem to conclude that
\[
\E_{\pi_c}\!\left[
K_{h_n}(Z_t)\,\ind{Z_t\ge c}(U_t(h_n))^j\,\Gamma^R_{t,\infty}
\right]\to
\kappa_j\E_{\pi_c}\!\left[Q^R_{c,t}(S_t,c,1)f_t(c\mid S_t)\right],
\]
for each fixed $t$. The right hand side is summable after multiplying by $\gamma^t$.  Hence the dominated convergence for series yields
\begin{align*}
&\sum_{t=0}^{T_n}\gamma^t\,
\E_{\pi_c}\!\left[
K_{h_n}(Z_t)\,\ind{Z_t\ge c}(U_t(h_n))^j\,\Gamma^R_{t,\infty}
\right] \to
\kappa_j\sum_{t=0}^{T}\gamma^t\,
\E_{\pi_c}\!\left[Q^R_{c,t}(S_t,c,1)f_t(c\mid S_t)\right].
\end{align*}
Combining this display with \eqref{eq:truncation-tail} proves the first claim. The same proof works when we replace $\Gamma^R_t$ with $1$ and $Q^R_{c,t}$ with $1$.
\end{proof}

\begin{lemma}[Fixed-effects representation]\label{lem:fe-representation}
With the scaled kernel $K_h(\cdot)$ and regressors $X_{i,t}(h)$ as defined in \eqref{eq:def-scaling},  define the weighted time-specific means
\[
  \mu_t(h)
  :=\frac{\E_{\pi_c} \left[K_h(Z_{t})X_{t}(h)\right]}
          {\E_{\pi_c}\left[ K_h(Z_{t})\right]},\qquad\overline X_t(h)
  :=\frac{n^{-1}\sum_{i=1}^n K_h(Z_{i,t})X_{i,t}(h)}
          {n^{-1}\sum_{i=1}^n K_h(Z_{i,t})},
\]
whenever the denominator is nonzero, with arbitrary value otherwise.  Then, the population-level optimizer in \eqref{eq:def-LLR-popln} satisfies
\begin{equation*}
\resizebox{\textwidth}{!}{$\displaystyle
\begin{aligned}
\tau^R_n
&=
e_1^\top
\left(
\sum_{t=0}^{T_n}\gamma^t\,
\E_{\pi_c}\!\left[
K_{h_n}(Z_t)
\{X_t(h_n)-\mu_t(h_n)\}^{\otimes 2}
\right]
\right)^{-1} 
\sum_{t=0}^{T_n}\gamma^t\,
\E_{\pi_c}\!\left[
K_{h_n}(Z_t)
\{X_t(h_n)-\mu_t(h_n)\}
\Gamma^R_t
\right],
\end{aligned}
$}
\end{equation*}
where $v^{\otimes2}$ denotes $vv^\top$. The inverse in the above display exists for all large $n$.
Similarly, the empirical optimizer in \eqref{eq:def-LLR} satisfies
\begin{equation*}
\resizebox{\textwidth}{!}{$\displaystyle
\begin{aligned}
\wh\tau^R_n
&=
e_1^\top
\left(\frac{1}{n}\sum_{i=1}^n
\sum_{t=0}^{T_n}\gamma^t\,
K_{h_n}(Z_{i,t})
\{X_{i,t}(h_n)-\overline X_t(h_n)\}^{\otimes 2}
\right)^{-1} \frac{1}{n}\sum_{i=1}^n
\sum_{t=0}^{T_n}\gamma^t\,
K_{h_n}(Z_{i,t})
\{X_{i,t}(h_n)-\overline X_t(h_n)\}
\Gamma^R_{i,t},
\end{aligned}
$}
\end{equation*}
where the displayed inverse exists with probability tending to $1$ as $n\to\infty$.
\end{lemma}

\begin{proof}  The proof is a straightforward application of the Frisch--Waugh--Lovell theorem for weighted least squares with time fixed-effects.
Multiplying all weights in \eqref{eq:def-LLR} and \eqref{eq:def-LLR-popln} by $h_n^{-1}$ does not change the optimizer.  Likewise, replacing $Z_{i,t}-c$ by $U_{i,t}(h_n)$ only rescales the two slope coefficients and leaves the coefficient on $A_{i,t}$ unchanged.  Thus, for the population-level problem, it is equivalent to minimize
\begin{equation}
    \label{eq:rescaled-LLR}
    \sum_{t=0}^{T_n}\gamma^t\,
  \E_{\pi_c}\!\left[K_{h_n}(Z_t)
  \{\Gamma^R_t-\alpha_t-\zeta^\top X_t(h_n)\}^2\right],
  \qquad \zeta=(\tau,\beta_0,\beta_1)^\top .
\end{equation}
For a fixed $\zeta$, the time fixed-effects minimizing this objective are given by
\[
  \alpha_t(\zeta)
  =\frac{\E_{\pi_c}\!\left[K_{h_n}(Z_t)\{\Gamma^R_t-\zeta^\top X_t(h_n)\}\right]}
          {\E_{\pi_c}\!\left[K_{h_n}(Z_t)\right]},
\]
with arbitrary value when the denominator is zero. After substituting these profiled fixed-effects in the objective \eqref{eq:rescaled-LLR}, we minimize over $\zeta$. Since
\(
  \E_{\pi_c}\!\left[K_{h_n}(Z_t)\{X_t(h_n)-\mu_t(h_n)\}\right]=0
\),
the profiled normal equations for $\zeta$ are precisely
\begin{equation}\label{eq:normal-eqn}
\left(
\sum_{t=0}^{T_n}\gamma^t\,
\E_{\pi_c}\!\left[K_{h_n}(Z_t)
\{X_t(h_n)-\mu_t(h_n)\}^{\otimes2}\right]
\right)\zeta =
\sum_{t=0}^{T_n}\gamma^t\,
\E_{\pi_c}\!\left[K_{h_n}(Z_t)
\{X_t(h_n)-\mu_t(h_n)\}\Gamma^R_t\right].
\end{equation}
Solving these equations and selecting the first coordinate gives the stated population representation. The non-singularity follows from the first part of the proof of \cref{lemma:localization}, where we show (cf.~\eqref{eq:hessian-limit}) that 
\begin{align}
&\sum_{t=0}^{T_n}\gamma^t\,
\E_{\pi_c}\!\bigl[
K_{h_n}(Z_t)
\{X_t(h_n)-\mu_t(h_n)\}^{\otimes2}
\bigr] \to
\left(\sum_{t=0}^{T}\gamma^t\bar f_t (c)\right)M,
\label{eq:limiting-Hessian}
\end{align}
where
\[
M:=
\begin{pmatrix}
\kappa_0 & \kappa_1 & \kappa_1\\
\kappa_1 & 2\kappa_2 & \kappa_2\\
\kappa_1 & \kappa_2 & \kappa_2
\end{pmatrix}
-
\frac{1}{2\kappa_0}
\begin{pmatrix}
\kappa_0\\[0.15em]
0\\[0.15em]
\kappa_1
\end{pmatrix}
\begin{pmatrix}
\kappa_0 & 0 & \kappa_1
\end{pmatrix}.
\]
A direct calculation gives
\(
\det(M)={(\kappa_0\kappa_2-\kappa_1^2)^2}/{2\kappa_0},
\) which is strictly positive since the kernel is non-degenerate (\cref{assump5:regularity-for-consistency}). Since $\sum_{t=0}^{T}\gamma^t\bar f_t (c)>0$ (\cref{assump:condtional-density}), we conclude that the limit in \eqref{eq:limiting-Hessian} is non-singular, which completes the proof.

The empirical case is identical with expectations replaced by empirical averages. 
\end{proof}

\begin{lemma}[Consistency of population-level optimizers]\label{lemma:localization}
Define the population-level optimizers $\tau^R_n$ ($R\in\{Y,A\}$) as in \eqref{eq:def-LLR-popln}.
    Under the conditions of \cref{consistency}, it holds that $\tau^R_n \,\to\, \tau^R$ as $n\to\infty$,
    where $\tau^R$ is defined in \eqref{eq:def-taustar}.
\end{lemma}

\begin{proof} By \cref{lem:fe-representation}, after the fixed effects are profiled out, the population normal equations \eqref{eq:normal-eqn} for the coefficients $(\tau,\beta_0,\beta_1)^\top$ have Hessian
\[
\sum_{t=0}^{T_n}\gamma^t\,
\E_{\pi_c}\!\bigl[
K_{h_n}(Z_t)
\{X_t(h_n)-\mu_t(h_n)\}^{\otimes 2}
\bigr]
\]
and right-hand side
\[
\sum_{t=0}^{T_n}\gamma^t\,
\E_{\pi_c}\!\bigl[
K_{h_n}(Z_t)
\{X_t(h_n)-\mu_t(h_n)\}\Gamma_t^R
\bigr].
\]
It is therefore enough to identify the limits of these two objects.  Define $\bar f_t (c) := \E_{\pi_c}\big[f_t(c\mid S_t)\big]
$.
It follows from the proof of~\cref{lem:kernel-localization} (with $\Gamma^R_t$ replaced by $1$) that for each fixed $t$ and each $j=0,1,2$,
\begin{equation*}
\E_{\pi_c}\!\bigl[K_{h_n}(Z_t)A_t (U_t(h_n))^j\bigr]
\to \kappa_j\bar f_t (c),\qquad
\E_{\pi_c}\!\bigl[K_{h_n}(Z_t)(1-A_t) (U_t(h_n))^j\bigr]
\to (-1)^j\kappa_j\bar f_t (c) .
\end{equation*}
Since $X_t(h)=(A_t,U_t(h),A_tU_t(h))^\top$, these limits imply
\begin{align}
\E_{\pi_c}\!\bigl[K_{h_n}(Z_t)X_t(h_n)\bigr]
&\to
\bar f_t (c)
\begin{pmatrix}
\kappa_0\\[0.15em]
0\\[0.15em]
\kappa_1
\end{pmatrix},
\label{eq:period-x}\\
\E_{\pi_c}\!\bigl[K_{h_n}(Z_t)\bigr]
&\to 2\kappa_0\bar f_t (c),
\label{eq:period-one}\\
\E_{\pi_c}\!\bigl[K_{h_n}(Z_t)X_t(h_n)^{\otimes2}\bigr]
&\to
\bar f_t (c)
\begin{pmatrix}
\kappa_0 & \kappa_1 & \kappa_1\\
\kappa_1 & 2\kappa_2 & \kappa_2\\
\kappa_1 & \kappa_2 & \kappa_2
\end{pmatrix}.
\label{eq:period-xx}
\end{align}
By the definition of $\mu_t(h_n)$,
\begin{align*}
&\E_{\pi_c}\!\bigl[
K_{h_n}(Z_t)
\{X_t(h_n)-\mu_t(h_n)\}^{\otimes2}
\bigr]=
\E_{\pi_c}\!\bigl[K_{h_n}(Z_t)X_t(h_n)^{\otimes2}\bigr]
-
\frac{
\E_{\pi_c}\!\bigl[K_{h_n}(Z_t)X_t(h_n)\bigr]^{\otimes2}
}{
\E_{\pi_c}\!\bigl[K_{h_n}(Z_t)\bigr]
},
\end{align*}
with the correction interpreted as zero when $\E_{\pi_c}[K_{h_n}(Z_t)]=0$, in which case $K_{h_n}(Z_t)=0$ almost surely. Combining \eqref{eq:period-x}--\eqref{eq:period-xx} gives the period-wise limit
\begin{align*}
&\E_{\pi_c}\!\bigl[
K_{h_n}(Z_t)
\{X_t(h_n)-\mu_t(h_n)\}^{\otimes2}
\bigr]\to
\bar f_t (c)\left\{
\begin{pmatrix}
\kappa_0 & \kappa_1 & \kappa_1\\
\kappa_1 & 2\kappa_2 & \kappa_2\\
\kappa_1 & \kappa_2 & \kappa_2
\end{pmatrix}
-
\frac{1}{2\kappa_0}
\begin{pmatrix}
\kappa_0\\[0.15em]
0\\[0.15em]
\kappa_1
\end{pmatrix}
\begin{pmatrix}
\kappa_0 & 0 & \kappa_1
\end{pmatrix}
\right\}.
\end{align*}
The same convergence may be summed over $t$. Indeed, on the event $K_{h_n}(Z_t)\neq0$, $|U_t(h_n)|\le1$ and $A_t\in\{0,1\}$, so the norm of each matrix term above is bounded by a constant times $\E_{\pi_c}[K_{h_n}(Z_t)]$. For all sufficiently large $n$ (large enough so that $|h_n|\le \eta$ where $\eta$ is as in \cref{assump5:regularity-for-consistency}),
\[
\E_{\pi_c}[K_{h_n}(Z_t)]=\E_{\pi_c}\left[\sum_{a=0}^1\int_0^1 K(u) f_t(S_t,\, c+(-1)^{1-a}h_nu,\, a)\, du\right]
\le 2\|K\|_\infty\E_{\pi_c}[B_{f,\,t}(S_t)],
\]
and $\sum_{t=0}^T\gamma^t\,\E_{\pi_c}[B_{f,\,t}(S_t)]<\infty$ follows from \cref{assump5:regularity-for-consistency}, because $B_{m_2,\,t}\ge1$. Applying the dominated convergence for series therefore yields
\begin{align}
&\sum_{t=0}^{T_n}\gamma^t\,
\E_{\pi_c}\!\bigl[
K_{h_n}(Z_t)
\{X_t(h_n)-\mu_t(h_n)\}^{\otimes2}
\bigr] \to
\left(\sum_{t=0}^{T}\gamma^t\bar f_t (c)\right)M,
\label{eq:hessian-limit}
\end{align}
where
\[
M:=
\begin{pmatrix}
\kappa_0 & \kappa_1 & \kappa_1\\
\kappa_1 & 2\kappa_2 & \kappa_2\\
\kappa_1 & \kappa_2 & \kappa_2
\end{pmatrix}
-
\frac{1}{2\kappa_0}
\begin{pmatrix}
\kappa_0\\[0.15em]
0\\[0.15em]
\kappa_1
\end{pmatrix}
\begin{pmatrix}
\kappa_0 & 0 & \kappa_1
\end{pmatrix}.
\]
A direct calculation gives
\(
\det(M)={(\kappa_0\kappa_2-\kappa_1^2)^2}/{2\kappa_0}>0,
\)
because the kernel is non-degenerate (\cref{assump5:regularity-for-consistency}). Since $\sum_{t=0}^{T}\gamma^t\bar f_t (c)>0$ (\cref{assump:condtional-density}), the limiting Hessian in \eqref{eq:hessian-limit} is non-singular.

\medskip

We now turn to the right-hand side of the normal equations \eqref{eq:normal-eqn}.
It follows from the proof of~\cref{lem:kernel-localization} that
\begin{align}
\E_{\pi_c}\!\bigl[K_{h_n}(Z_t)\,\Gamma_t^R\bigr]
&\to
\kappa_0
\E_{\pi_c}\!\bigl[
\{Q^R_{c,t}(S_t,c,1)+Q^R_{c,t}(S_t,c,0)\}f_t(c\mid S_t)
\bigr],
\label{eq:period-gamma}\\[2mm]
\E_{\pi_c}\!\bigl[K_{h_n}(Z_t)\,X_t(h_n)\,\Gamma_t^R\bigr]
&\to
\begin{pmatrix}
\kappa_0\,\E_{\pi_c}\!\bigl[Q^R_{c,t}(S_t,c,1)f_t(c\mid S_t)\bigr]\\[0.35em]
\kappa_1\E_{\pi_c}\!\bigl[\{Q^R_{c,t}(S_t,c,1)-Q^R_{c,t}(S_t,c,0)\}f_t(c\mid S_t)\bigr]\\[0.35em]
\kappa_1\E_{\pi_c}\!\bigl[Q^R_{c,t}(S_t,c,1)f_t(c\mid S_t)\bigr]
\end{pmatrix}.
\label{eq:period-xgamma}
\end{align}
Also, by \eqref{eq:period-x} and \eqref{eq:period-one},
\begin{equation}
    \label{eq:limit-of-mu-t}
    \mu_t(h_n)
=
\frac{\E_{\pi_c}[K_{h_n}(Z_t)X_t(h_n)]}
{\E_{\pi_c}[K_{h_n}(Z_t)]}
\to
\frac{1}{2\kappa_0}
\begin{pmatrix}
\kappa_0\\[0.15em]
0\\[0.15em]
\kappa_1
\end{pmatrix},
\end{equation}
whenever $\bar f_t (c)>0$. This combined with \eqref{eq:period-gamma} and \eqref{eq:period-xgamma} gives
\begin{align*}
&\E_{\pi_c}\!\bigl[
K_{h_n}(Z_t)
\{X_t(h_n)-\mu_t(h_n)\}\Gamma_t^R
\bigr]\to
\E_{\pi_c}\!\bigl[
\{Q^R_{c,t}(S_t,c,1)-Q^R_{c,t}(S_t,c,0)\}f_t(c\mid S_t)
\bigr]
\begin{pmatrix}
\kappa_0/2\\[0.15em]
\kappa_1\\[0.15em]
\kappa_1/2
\end{pmatrix},
\end{align*}
provided $\bar f_t (c)>0$.  If $\bar f_t (c)=0$, the term
$\mu_t(h_n)\,\E_{\pi_c}[K_{h_n}(Z_t)\,\Gamma_t^R]$ converges to zero because $\mu_t(h_n)$ is bounded on the kernel support and \eqref{eq:period-gamma} has zero limit (since $\bar f_t (c)=0$ implies that $f_t(c\mid S_t)=0$ almost surely). Thus, the above display holds for every fixed $t$, regardless of whether $\bar f_t (c)>0$.

As earlier, this period-wise convergence can be summed over $t$: For $R=Y$, the absolute first moment of the centered term is bounded by a constant times
$\E_{\pi_c}[B_{m_2,\,t}(S_t)\,B_{f,\,t}(S_t)]$ by the same envelope calculation used in \cref{lem:kernel-localization}. For $R=A$, the same bound is immediate because $\Gamma_t^A$ is uniformly bounded by $(1-\gamma)^{-1}$ when $\gamma<1$ and by $T+1$ when $\gamma = 1$ and $T<\infty$. Therefore, using the dominated convergence theorem for series, we conclude that
\begin{align}
&\sum_{t=0}^{T_n}\gamma^t\,
\E_{\pi_c}\!\bigl[
K_{h_n}(Z_t)
\{X_t(h_n)-\mu_t(h_n)\}\Gamma_t^R
\bigr]
\notag\\
&\qquad\to
\left(
\sum_{t=0}^{T}\gamma^t\,
\E_{\pi_c}\!\bigl[
\{Q^R_{c,t}(S_t,c,1)-Q^R_{c,t}(S_t,c,0)\}f_t(c\mid S_t)
\bigr]
\right)
\begin{pmatrix}
\kappa_0/2\\[0.15em]
\kappa_1\\[0.15em]
\kappa_1/2
\end{pmatrix}.
\label{eq:rhs-limit}
\end{align}
Finally, a direct calculation gives
\[
M e_1
=
\begin{pmatrix}
\kappa_0/2\\[0.15em]
\kappa_1\\[0.15em]
\kappa_1/2
\end{pmatrix}, \quad\text{since}\quad M=
\begin{pmatrix}
\kappa_0 & \kappa_1 & \kappa_1\\
\kappa_1 & 2\kappa_2 & \kappa_2\\
\kappa_1 & \kappa_2 & \kappa_2
\end{pmatrix}
-
\frac{1}{2\kappa_0}
\begin{pmatrix}
\kappa_0\\[0.15em]
0\\[0.15em]
\kappa_1
\end{pmatrix}
\begin{pmatrix}
\kappa_0 & 0 & \kappa_1
\end{pmatrix}.
\]
Combining \eqref{eq:hessian-limit} and \eqref{eq:rhs-limit}, and using continuity of the matrix inverse at the non-singular limiting Hessian, we obtain
\begin{align*}
\tau_n^R
&\to
e_1^\top
\left\{
\left(\sum_{t=0}^{T}\gamma^t\bar f_t (c)\right)M
\right\}^{-1}
M e_1
\left(
\sum_{t=0}^{T}\gamma^t\,
\E_{\pi_c}\!\bigl[
\{Q^R_{c,t}(S_t,c,1)-Q^R_{c,t}(S_t,c,0)\}f_t(c\mid S_t)
\bigr]
\right)\\[2mm]
&\quad=
\frac{
\sum_{t=0}^{T}\gamma^t\,
\E_{\pi_c}\!\bigl[
\{Q^R_{c,t}(S_t,c,1)-Q^R_{c,t}(S_t,c,0)\}f_t(c\mid S_t)
\bigr]
}{
\sum_{t=0}^{T}\gamma^t\,\bar f_t (c)
}
=\tau^R,
\end{align*}
as desired to show.
\end{proof}

\begin{lemma}[Concentration of empirical optimizers]\label{lemma:concentration}
 Define the population-level optimizers $\tau^R_n$ ($R\in\{Y,A\}$) as in \eqref{eq:def-LLR-popln}.
    Under the conditions of \cref{consistency}, the empirical optimizers $\wh{\tau}^R(h_n)$ defined in \eqref{eq:def-LLR} satisfy $$\wh{\tau}^R(h_n) -\tau^R_n=\o(1),\quad \text{as }n\to\infty.$$
\end{lemma}

\begin{proof}
In view of~\cref{lem:fe-representation}, it suffices to show the convergence of the residualized design matrix and residualized cross-moment appearing in the empirical normal equations \eqref{eq:normal-eqn}.
We first record the two basic bounds that will be used repeatedly.  On the support of $K_{h_n}(Z_t)$, every entry of $X_t(h_n)$ is bounded by one. Consequently,
\begin{align}
\E_{\pi_c}\!\left[K^2_{h_n}(Z_t)\,\|X_t(h_n)\|^2\,(\Gamma_t^R)^2\right]
&\le \frac{C}{h_n}\E_{\pi_c}[B_{m_2,\,t}(S_t)\,B_{f,\,t}(S_t)], \label{eq:l2-cross}\\[2mm]
\E_{\pi_c}\!\left[K^2_{h_n}(Z_t)\,\|X_t(h_n)\|^4\right]
&\le \frac{C}{h_n}\E_{\pi_c}[B_{m_2,\,t}(S_t)\,B_{f,\,t}(S_t)]. \label{eq:l2-design}
\end{align}
Here (and in the rest of the proof) we use $C$ as a placeholder for universal constants that are not tracked.
To show \eqref{eq:l2-cross}, we first use $K_h(Z_t)\|X_t(h)\|\le C K_h(Z)$, then apply the same principle as in the proofs of \cref{lem:kernel-localization,lemma:localization}: Condition on $S_t$, split the integral over the two sides of the cutoff, and use the change of variables $z=c\pm h_nu$.  For example, the right-hand side contribution is bounded by
\[
\frac{C}{h_n}\E_{\pi_c}\!\left[\int_0^1 K^2(u)\,
m^R_{2,t}(S_t,c+h_nu,1)\,f_t(c+h_nu\mid S_t)du\right].
\]
For $R=Y$ this is bounded by the right-hand side of \eqref{eq:l2-cross} by \cref{assump5:regularity-for-consistency}.  For $R=A$, $\Gamma_t^A\le(1-\gamma)^{-1}$ in the infinite-horizon case and is bounded by $T+1$ in the finite-horizon case, so the same bound holds because $B_{m_2,\,t}\ge1$.  The left-hand side contribution is identical. For showing \eqref{eq:l2-design}, the same conditioning and change of variables give
\[
\E_{\pi_c}\!\left[K^2_{h_n}(Z_t)\|X_t(h_n)\|^4\right]
\le \frac{C}{h_n}
\E_{\pi_c}\!\left[\int_{-1}^{1}K^2(u)\, f_t(c+h_nu\mid S_t)\,du\right]
\le \frac{C}{h_n}\E_{\pi_c}[B_{m_2,\,t}(S_t)\,B_{f,\,t}(S_t)],
\]
where the last inequality again uses $B_{m_2,\,t}\ge1$.
This completes the proofs of \eqref{eq:l2-cross} and \eqref{eq:l2-design}.

\medskip

Returning to the main proof, we first control the variance term without fixed-effects. 
Since $\sum_{t=0}^T\gamma^t\,\E[B_{m_2,\,t}(S_t)\,B_{f,\,t}(S_t)]<\infty$ and $\sum_{t=0}^T\gamma^t\,<\infty$, the Cauchy--Schwarz inequality gives
\[
  \sum_{t=0}^{T}\gamma^t\, \left(\E_{\pi_c}[B_{m_2,\,t}(S_t)\,B_{f,\,t}(S_t)]\right)^{1/2}
  \le
  \left(\sum_{t=0}^{T}\gamma^t\,\right)^{1/2}
  \left(\sum_{t=0}^{T}\gamma^t\,
  \E_{\pi_c}[B_{m_2,\,t}(S_t)\,B_{f,\,t}(S_t)]\right)^{1/2}<\infty.
\]
Therefore, by Minkowski's inequality and \eqref{eq:l2-cross},
\[
\left\|
\sum_{t=0}^{T_n}\gamma^t\, K_{h_n}(Z_t)\,X_t(h_n)\,\Gamma_t^R
\right\|_{L_2}
\le C h_n^{-1/2}\sum_{t=0}^{T}\gamma^t\, \left(\E_{\pi_c}[B_{m_2,\,t}(S_t)\,B_{f,\,t}(S_t)]\right)^{1/2}
\le C h_n^{-1/2}.
\]
Since observations are independent across $i$, it follows that
\begin{align*}
&\var\!\left(
\frac1n\sum_{i=1}^n\sum_{t=0}^{T_n}\gamma^t\,
K_{h_n}(Z_{i,t})X_{i,t}(h_n)\Gamma^R_{i,t}
\right)\le
\frac1n
\E_{\pi_c}\!\left\|
\sum_{t=0}^{T_n}\gamma^t\,K_{h_n}(Z_t)X_t(h_n)\Gamma_t^R
\right\|^2
\le \frac{C}{nh_n}=o(1).
\end{align*}
Therefore, Chebyshev's inequality yields
\begin{align}
&\frac1n\sum_{i=1}^n\sum_{t=0}^{T_n}\gamma^t\,
K_{h_n}(Z_{i,t})X_{i,t}(h_n)\Gamma^R_{i,t}
-
\sum_{t=0}^{T_n}\gamma^t\,
\E_{\pi_c}\left[K_{h_n}(Z_t)X_t(h_n)\Gamma^R_t\right]
=\O((nh_n)^{-1/2}) =\o(1). \label{eq:uncentered-cross-conc}
\end{align}
The same argument, using \eqref{eq:l2-design} instead of \eqref{eq:l2-cross}, gives
\begin{align}
&\frac1n\sum_{i=1}^n\sum_{t=0}^{T_n}\gamma^t\,
K_{h_n}(Z_{i,t})\,X_{i,t}(h_n)^{\otimes2}
-
\sum_{t=0}^{T_n}\gamma^t\,
\E_{\pi_c}[K_{h_n}(Z_t)\,X_t(h_n)^{\otimes2}]
=\O((nh_n)^{-1/2}). \label{eq:uncentered-design-conc}
\end{align}
The corresponding scalar and vector moments involving only $K_{h_n}(Z_t)$ and $K_{h_n}(Z_t)X_t(h_n)$ concentrate by the same argument. In particular, 
for each fixed $t$,
\begin{align}
\frac1n\sum_{i=1}^nK_{h_n}(Z_{i,t})\,
&-\E_{\pi_c}\!\left[K_{h_n}(Z_t)\right]=\O((nh_n)^{-1/2}), \label{eq:fixed-mass}\\
\frac1n\sum_{i=1}^nK_{h_n}(Z_{i,t})X_{i,t}(h_n)
&-\E_{\pi_c}\!\left[K_{h_n}(Z_t)X_t(h_n)\right]=\O((nh_n)^{-1/2}), \label{eq:fixed-first}\\
\frac1n\sum_{i=1}^nK_{h_n}(Z_{i,t})\Gamma_{i,t}^R
&-\E_{\pi_c}\!\left[K_{h_n}(Z_t)\Gamma_t^R\right]=\O((nh_n)^{-1/2}), \label{eq:fixed-outcome}\\
\frac1n\sum_{i=1}^nK_{h_n}(Z_{i,t})X_{i,t}(h_n)\Gamma_{i,t}^R
&-\E_{\pi_c}\!\left[K_{h_n}(Z_t)X_t(h_n)\Gamma_t^R\right]=\O((nh_n)^{-1/2}), \label{eq:fixed-x-outcome}\\
\frac1n\sum_{i=1}^nK_{h_n}(Z_{i,t})\,X_{i,t}(h_n)^{\otimes 2}
&-\E_{\pi_c}\!\left[K_{h_n}(Z_t)\,X_t(h_n)^{\otimes 2}\right]=\O((nh_n)^{-1/2}). \label{eq:fixed-second}
\end{align}

We now pass from uncentered moments to time-demeaned moments. Throughout this step, take the arbitrary values used when a denominator vanishes to be bounded, say zero. %Since $K_{h_n}\ge0$ and every entry of $X_t(h_n)$ is bounded on the support of $K_{h_n}(Z_t)$, both $\mu_t(h_n)$ and $\overline X_t(h_n)$ are uniformly bounded whenever they are defined by their quotients, and they remain bounded under the preceding convention.
We first consider a fixed period $t$, and assume that $\bar f_t (c) = \E_{\pi_c}\big[f_t(c\mid S_t)\big]>0$. \cref{lem:kernel-localization} gives
\(
\E_{\pi_c}\!\left[K_{h_n}(Z_t)\right]\to 2\kappa_0\bar f_t (c)>0\). Combining this with \eqref{eq:fixed-mass}, we can say that the denominator in the empirical quotient defining $\overline X_t(h_n)$ (cf.~\cref{lem:fe-representation}) is bounded away from zero with probability tending to one. Therefore, it follows from \eqref{eq:fixed-mass} and \eqref{eq:fixed-first} that
\begin{align}
\overline X_t(h_n)-\mu_t(h_n)
&=
\frac{n^{-1}\sum_{i=1}^nK_{h_n}(Z_{i,t})X_{i,t}(h_n)}{n^{-1}\sum_{i=1}^nK_{h_n}(Z_{i,t})}
-
\frac{\E_{\pi_c}[K_{h_n}(Z_t)X_t(h_n)]}{\E_{\pi_c}[K_{h_n}(Z_t)]}
=\O((nh_n)^{-1/2}). \label{eq:mean-quotient}
\end{align}
Indeed, on the event that $n^{-1}\sum_{i=1}^n K_{h_n}(Z_{i,t})\ge \E_{\pi_c}[K_{h_n}(Z_t)]/2$, the norm of the difference in \eqref{eq:mean-quotient} is at most
\[
\frac{
\left\|n^{-1}\sum_{i=1}^nK_{h_n}(Z_{i,t})X_{i,t}(h_n)-\E_{\pi_c}[K_{h_n}(Z_t)X_t(h_n)]\right\|}{n^{-1}\sum_{i=1}^nK_{h_n}(Z_{i,t})}
+
C\left|\frac{1}{n^{-1}\sum_{i=1}^nK_{h_n}(Z_{i,t})}-\frac{1}{\E_{\pi_c}[K_{h_n}(Z_t)]}\right|,
\]
and both terms are $\O((nh_n)^{-1/2})$ by \eqref{eq:fixed-mass} and \eqref{eq:fixed-first}. 

To handle the `design' contribution in this fixed period, we use the following identities
\begin{align*}
&\frac1n\sum_{i=1}^nK_{h_n}(Z_{i,t})
\{X_{i,t}(h_n)-\overline X_t(h_n)\}^{\otimes2}=
\frac1n\sum_{i=1}^nK_{h_n}(Z_{i,t})X_{i,t}(h_n)^{\otimes2}
-
\left(\frac1n\sum_{i=1}^nK_{h_n}(Z_{i,t})\right)\overline X_t(h_n)^{\otimes2},
\end{align*}
 and
\begin{align*}
\E_{\pi_c}\!\left[K_{h_n}(Z_t)\{X_t(h_n)-\mu_t(h_n)\}^{\otimes2}\right]
&=
\E_{\pi_c}\!\left[K_{h_n}(Z_t)X_t(h_n)^{\otimes2}\right]-
\E_{\pi_c}\!\left[K_{h_n}(Z_t)\right]\mu_t(h_n)^{\otimes2}.
\end{align*}
Using \eqref{eq:fixed-mass}, \eqref{eq:fixed-second}, \eqref{eq:mean-quotient}, and boundedness of $\mu_t(h_n)$ and $\overline X_t(h_n)$, we conclude that 
\begin{equation}
    \label{eq:fixed-L-design}
    \frac1n\sum_{i=1}^nK_{h_n}(Z_{i,t})
    \{X_{i,t}(h_n)-\overline X_t(h_n)\}^{\otimes2}-\E_{\pi_c}\left[K_{h_n}(Z_t)
\{X_t(h_n)-\mu_t(h_n)\}^{\otimes2}\right]=\O((nh_n)^{-1/2}).
\end{equation}
Similarly, we can handle the `cross-moment' contribution using the identities
\begin{align*}
&\frac1n\sum_{i=1}^nK_{h_n}(Z_{i,t})\,
\{X_{i,t}(h_n)-\overline X_t(h_n)\}\,\Gamma_{i,t}^R \\
&\qquad=
\frac1n\sum_{i=1}^nK_{h_n}(Z_{i,t})\,X_{i,t}(h_n)\,\Gamma_{i,t}^R
-
\overline X_t(h_n)\left(\frac1n\sum_{i=1}^nK_{h_n}(Z_{i,t})\,\Gamma_{i,t}^R\right),
\end{align*}
 and
\begin{align*}
\E_{\pi_c}\!\left[K_{h_n}(Z_t)\,\{X_t(h_n)-\mu_t(h_n)\}\,\Gamma_t^R\right]
&=
\E_{\pi_c}\!\left[K_{h_n}(Z_t)\,X_t(h_n)\Gamma_t^R\right]
-
\mu_t(h_n)\,\E_{\pi_c}\!\left[K_{h_n}(Z_t)\,\Gamma_t^R\right].
\end{align*}
Thus, \eqref{eq:fixed-outcome}, \eqref{eq:fixed-x-outcome}, and \eqref{eq:mean-quotient} together imply that 
\begin{equation}
    \label{eq:fixed-L-cross}
\frac1n\sum_{i=1}^nK_{h_n}(Z_{i,t})
\{X_{i,t}(h_n)-\overline X_t(h_n)\}\Gamma_{i,t}^R-\E_{\pi_c}\!\left[K_{h_n}(Z_t)\{X_t(h_n)-\mu_t(h_n)\}\Gamma_t^R\right]=\O((nh_n)^{-1/2}).
\end{equation}

It remains to treat the case $\bar f_t (c)= \E_{\pi_c}\big[f_t(c\mid S_t)\big]=0$. Using $K_{h_n}(Z_{i,t})\norm{X_{i,t}(h_n)-\overline X_t(h_n)}\le C\,K_{h_n}(Z_{i,t})$ together with \eqref{eq:fixed-mass}, we arrive at
\begin{align*}
\left\|
\frac1n\sum_{i=1}^nK_{h_n}(Z_{i,t})
\{X_{i,t}(h_n)-\overline X_t(h_n)\}^{\otimes2}
\right\|
&\le C\frac1n\sum_{i=1}^nK_{h_n}(Z_{i,t})=\E_{\pi_c}[K_{h_n}(Z_t)]+\O((nh_n)^{-1/2}),
\end{align*}
Similarly, the population analogue is also bounded by $C\,\E_{\pi_c}[K_{h_n}(Z_t)]$. We will show now that $\E_{\pi_c}[K_{h_n}(Z_t)]=o(1)$. Since $f_t\ge 0$, $\E_{\pi_c}[f_t(c\mid S_t)]=0$ implies that $f_t(c\mid S_t)=0$ almost surely. Next, using the same arguments as in the proof of \cref{lem:kernel-localization} (condition on $S_t$, apply the change of variables $z=c\pm h_nu$), we deduce that
\begin{align*}
\left|\E_{\pi_c}[K_{h_n}(Z_t)]\right| &=\E_{\pi_c}\left|\int_{-1}^1 K(u) f_t(c+h_nu\mid S_t)\,du\right|\\
&\le \E_{\pi_c} \int_{-1}^1 K(u) \left|f_t(c+h_nu\mid S_t)-f_t(c\mid S_t)\right|\,du=o(1),%\numberthis\label{eq:Kgoesto0}
\end{align*}
using the continuity of $z\mapsto f_t(z\mid S_t)$ at $c$ and the dominated convergence theorem.  Thus, \eqref{eq:fixed-L-design} continues to hold in this case. 

Similarly for the cross-moment term, we apply Cauchy--Schwarz inequality to get
\begin{align*}
&\E_{\pi_c}\!\left[
\frac1n\sum_{i=1}^nK_{h_n}(Z_{i,t})
\norm{X_{i,t}(h_n)-\overline X_t(h_n)}\,|\Gamma_{i,t}^R|
\right]\\
&\le C\,\E_{\pi_c}\!\left[K_{h_n}(Z_t)|\Gamma_t^R|\right]  \le C\,\E_{\pi_c}^{1/2}[K_{h_n}(Z_t)]\,
\E_{\pi_c}^{1/2}[K_{h_n}(Z_t)(\Gamma_t^R)^2]
=o(1).%\numberthis\label{eq:KGgoesto0}
\end{align*}
The second factor is bounded by the same envelope calculation used above, while the first factor tends to zero. Markov's inequality then makes the empirical centered cross-moment $\o(1)$, and the population analogue is $o(1)$ by the same Cauchy--Schwarz bound. Hence, \eqref{eq:fixed-L-cross} continues to hold in this case.

Taken together, \eqref{eq:fixed-L-design} and \eqref{eq:fixed-L-cross} complete the proof for the finite-horizon case. For the infinite-horizon case, we now fix a finite window $L$ and bound the corresponding tails.
For the design part, boundedness of $X_{i,t}(h_n)-\overline X_t(h_n)$ on the support of the kernel weights gives
\begin{align*}
&\left\|
\frac1n\sum_{i=1}^n\sum_{t>L}^{T_n}\gamma^tK_{h_n}(Z_{i,t})
\{X_{i,t}(h_n)-\overline X_t(h_n)\}^{\otimes2}
\right\|  \le
C\frac1n\sum_{i=1}^n\sum_{t>L}^{T_n}\gamma^tK_{h_n}(Z_{i,t}).
\end{align*}
For all large $n$, the expectation of the right-hand side is bounded (again, by conditioning on $S_t$ and using the same change of variables as above) by the following
\[
C\sum_{t>L}\gamma^t\,\E_{\pi_c}\!\left[B_{f,\,t}(S_t)\right]
\le
C\sum_{t>L}\gamma^t\,\E_{\pi_c}\!\left[B_{m_2,\,t}(S_t)B_{f,\,t}(S_t)\right],
\]
which goes to $0$ as $L\to\infty$. Therefore, by Markov's inequality,
\[
\lim_{L\to\infty}\limsup_{n\to\infty}
\mathbb P\!\left(
\left\|
\frac1n\sum_{i=1}^n\sum_{t>L}^{T_n}\gamma^tK_{h_n}(Z_{i,t})
\{X_{i,t}(h_n)-\overline X_t(h_n)\}^{\otimes2}
\right\|>\varepsilon
\right)=0
\]
for every $\varepsilon>0$. The population design tail is controlled by the same deterministic bound with the empirical average replaced by expectation.
The cross-moment tail is analogous:
\begin{align*}
&\E\!
\left\|
\frac1n\sum_{i=1}^n\sum_{t>L}^{T_n}\gamma^t\,K_{h_n}(Z_{i,t})
\{X_{i,t}(h_n)-\overline X_t(h_n)\}\Gamma_{i,t}^R
\right\|
 \\
&\qquad\le
C\sum_{t>L}\gamma^t\,\E_{\pi_c}\!\left[K_{h_n}(Z_t)|\Gamma_t^R|\right] \\
&\qquad\le
C\sum_{t>L}\gamma^t
\{\E_{\pi_c}[K_{h_n}(Z_t)]\}^{1/2}
\{\E_{\pi_c}[K_{h_n}(Z_t)(\Gamma_t^R)^2]\}^{1/2} \\
&\qquad\le
C\sum_{t>L}^{T}\gamma^t\,\E_{\pi_c}\!\left[B_{m_2,\,t}(S_t)B_{f,\,t}(S_t)\right]\to0,
\end{align*}
where the final inequality uses the local density envelope for the first factor, and the local second-moment envelope for the second factor. Again, Markov's inequality handles the empirical tail, and the population tail is identical. Combining \eqref{eq:fixed-L-design}--\eqref{eq:fixed-L-cross} with these tail bounds yields
\begin{align}
\wh M_n-M_n=\o(1), \qquad
\wh v_n^R-v_n^R=\o(1), \label{eq:design-final-conc}
\end{align}
where 
\begin{equation*}
\resizebox{\textwidth}{!}{$\displaystyle
\begin{aligned}
    \wh M_n:=\frac{1}{n}\sum_{i=1}^n
\sum_{t=0}^{T_n}\gamma^t\,
K_{h_n}(Z_{i,t})
\{X_{i,t}(h_n)-\overline X_t(h_n)\}^{\otimes 2},\ \ M_n := 
\sum_{t=0}^{T_n}\gamma^t\,
\E_{\pi_c}\!\left[
K_{h_n}(Z_t)
\{X_t(h_n)-\mu_t(h_n)\}^{\otimes 2}
\right]
\end{aligned}
$}
\end{equation*}
\begin{equation*}
\resizebox{\textwidth}{!}{$\displaystyle
\begin{aligned}
    \wh v_n^R:=\frac{1}{n}\sum_{i=1}^n
\sum_{t=0}^{T_n}\gamma^t\,
K_{h_n}(Z_{i,t})
\{X_{i,t}(h_n)-\overline X_t(h_n)\}
\Gamma^R_{i,t},\ \ v_n^R := 
\sum_{t=0}^{T_n}\gamma^t\,
\E_{\pi_c}\!\left[
K_{h_n}(Z_t)
\{X_t(h_n)-\mu_t(h_n)\}
\Gamma^R_t
\right].
\end{aligned}
$}
\end{equation*}

It remains to transfer these moment convergences to the optimizer. \cref{lemma:localization} gives
\[
M_n\to \left(\sum_{t=0}^T\gamma^t\bar f_t (c)\right)M,
\]
where the limiting matrix is nonsingular. Hence there exists $\lambda_0>0$ such that
\(
\lambda_{\min}(M_n)\ge \lambda_0
\)
for all sufficiently large $n$. By \eqref{eq:design-final-conc} and Weyl's inequality, we have
\(
\lambda_{\min}(\wh M_n)\ge\lambda_0/2
\)
with probability tending to one. On this event, both inverses exist and
\[
\wh M_n^{-1}-M_n^{-1}
=\wh M_n^{-1}(M_n-\wh M_n)M_n^{-1},
\]
which implies that
\[
\norm{\wh M_n^{-1}-M_n^{-1}}
\le \frac{2}{\lambda_0^2}\norm{\wh M_n-M_n}=\o(1).
\]
The population cross-moments $v_n^R$ are bounded, because they converge (\cref{lemma:localization}), and \eqref{eq:design-final-conc} gives $\wh v_n^R=v_n^R+\o(1)$. Therefore,
\begin{align*}
\wh\tau^R(h_n)-\tau_n^R
&=e_1^\top\{\wh M_n^{-1}\wh v_n^R-M_n^{-1}v_n^R\} \\[2mm]
&=e_1^\top(\wh M_n^{-1}-M_n^{-1})v_n^R
+e_1^\top\wh M_n^{-1}(\wh v_n^R-v_n^R)
=\o(1).
\end{align*}
This completes the proof.
\end{proof}

\begin{lemma}[Second-order bias expansion]\label{lem:bias}
For $R\in\{Y,A\}$, define the population-level optimizer $\tau^R_n$ as in \eqref{eq:def-LLR-popln}, and its limit $\tau^R$ as in \eqref{eq:def-taustar}. 
    Under the assumptions of \cref{clt-for-twice-discounted-llr}, we have that
    $$ \tau^R_n- \tau^R=\frac{h_n^2}{2}\,\xi_1 (\mu_{R,1}''(c)-\mu_{R,0}''(c))+{o}(h_n^2),$$
where $\xi_1:=(\kappa_2^2-\kappa_1\kappa_3)/(\kappa_0\kappa_2-\kappa_1^2)$ with  $\kappa_j:=\int_0^1 u^j K(u)du$, and 
\begin{equation}
    \label{def-muRz-duplicate}
    \mu_{R,a}(z):=\frac{\sum_{t=0}^T \gamma^t\,\E_{\pi_c}\left[Q_{c,t}^R(S_t,\,z,\,a)\,f_t(z\mid S_t)\right]}{\sum_{t=0}^T \gamma^t\,\E_{\pi_c}\left[f_t(z\mid S_t)\right]}, \quad a=0,1.
\end{equation}
\end{lemma}

\begin{proof}
    Fix $R\in\{Y,A\}$.  We first handle the truncation error that only appears when $T=\infty$.  For any fixed
integer $j\ge0$, on the event $K_{h_n}(Z_t)>0$ we have $|U_t(h_n)|\le1$ and
$K_{h_n}(Z_t)\le \|K\|_\infty/h_n$.  Moreover,
\[
        \Gamma^R_{t,\infty}-\Gamma_t^R
        =\gamma^{T_n-t+1}\Gamma^R_{T_n+1,\infty},
        \qquad 0\le t\le T_n.
\]
Therefore
\begin{align*}
&\left|\sum_{t=0}^{T_n}\gamma^t\,
\E_{\pi_c}\left[K_{h_n}(Z_t)\ind{Z_t\ge c}(U_t(h_n))^j
\{\Gamma^R_{t,\infty}-\Gamma_t^R\}\right]\right|
\notag\\
&\quad\le
\frac{\|K\|_\infty}{h_n}\sum_{t=0}^{T_n}\gamma^t\,\gamma^{T_n-t+1}
\E_{\pi_c}|\Gamma^R_{T_n+1,\infty}|
=
\frac{\|K\|_\infty}{h_n}(T_n+1)\gamma^{T_n+1}
\E_{\pi_c}|\Gamma^R_{T_n+1,\infty}|.
\end{align*}
For $R=A$, $\E_{\pi_c}|\Gamma^A_{T_n+1,\infty}|\le(1-\gamma)^{-1}$.  For $R=Y$,
$\E_{\pi_c}|\Gamma^Y_{T_n+1,\infty}|\le(1-\gamma)^{-1}\sup_s(\E_{\pi_c} Y_s^2)^{1/2}$ by Cauchy--Schwarz
and \cref{assump5:regularity-for-consistency}. Finally, $h=O(n^{-1/5})$ and $T_n\ge (3/5+\eps)(\log n)/(\log \gamma^{-1})$ together imply that $(T_n+1)\gamma^{T_n+1}=o(h_n^3)$.
We therefore conclude that
\begin{equation}
\left|\sum_{t=0}^{T_n}\gamma^t\,
\E_{\pi_c}\left[K_{h_n}(Z_t)\ind{Z_t\ge c}(U_t(h_n))^j
\{\Gamma^R_{t,\infty}-\Gamma_t^R\}\right]\right| =o(h_n^2).    \label{eq:tail-bias-proof}
\end{equation}
The same argument applies on the left side of the cutoff.  

\medskip

We now turn to the second-order localization calculation.  Conditioning on $S_t$, using
$A_t=1$ on $\{Z_t\ge c\}$, and changing variables $z=c+h_nu$ gives
\begin{align}
&\E_{\pi_c}\left[K_{h_n}(Z_t)\,\ind{Z_t\ge c}(U_t(h_n))^j\,\Gamma^R_{t,\infty}\right]
\notag\\
&\quad=
\int_0^1 u^jK(u)\,
\E_{\pi_c}\left[Q^R_{c,t}(S_t,\,c+h_nu,\,1)\,f_t(c+h_nu\mid S_t)\right]\,du.
\label{eq:right-loc-correct}
\end{align}
    Define $g_t(s,\,z,\,a):=Q_{c,\,t}^R(s,\, z,\, a)f_t(z\mid s)$. Since $Q_{c,\,t}^R(s,\,\cdot\,,\,a)$ and $f(\,\cdot\mid s)$ are twice differentiable at $c$ with their second derivatives locally bounded by an integrable envelopes $B_{Q'',\,t}(s)$ and $B_{f'',\,t}(s)$ respectively, it follows that $g_t(s,\,\cdot\,,\,a)$ are also twice  differentiable at $c$ with second derivative locally bounded by an integrable envelope $B_{g'',\,t}(s)$.
We thus have the following Taylor expansion.
$$g_t(s,\,c+h_nu,\,1)=g_t(s,\,c,\,1)+h_nu \,\frac{\partial}{\partial z}g_t(s,\,c,\,1) + \frac{h_n^2u^2}{2}\left(\frac{\partial^2}{\partial z^2}g_t(s,\,c
,\,1)+r_{n,t}(s,\,h_n)\right),$$
    where the remainder $r_{n,t}(s,\,h_n)\to 0$ as $n\to \infty$, since $g_t(s,\,\cdot\,,\,1)$ is twice differentiable at $c$. Moreover, $\partial^2_z g_t(s,\,z,\,a)$ are bounded in $[c-\eta,c+\eta]$ by an integrable envelope, which implies that $\sum_{t=0}^{T_n}\gamma^t\,\E_{\pi_c} r_{n,t}(S_t, h_n)\to 0$ by the dominated convergence theorem. We can therefore continue from \eqref{eq:right-loc-correct} to say that 
    \begin{align*}
       & \sum_{t=0}^{T_n} \E_{\pi_c}\left[K_{h_n}(Z_t)\,\ind{Z_t\ge c}(U_t(h_n))^j\,\Gamma^R_{t,\infty}\right]\\
        &=\sum_{t=0}^{T_n} \gamma^t \,\E_{\pi_c}\left[\kappa_j g_t(S_t,\, c,\,1) + \kappa_{j+1}\, h_n\,\frac{\partial}{\partial z}g_t(S_t,\,c,\,1) + \kappa_{j+2}\,\frac{h_n^2}{2}\, \frac{\partial^2}{\partial z^2}g_t(S_t,\,c,\,1)\right]+o(h_n^2)\\
        &=\sum_{t=0}^{T} \gamma^t \,\E_{\pi_c}\left[\kappa_j g_t(S_t,\, c,\,1) + \kappa_{j+1}\, h_n\,\frac{\partial}{\partial z}g_t(S_t,\,c,\,1) + \kappa_{j+2}\,\frac{h_n^2}{2}\, \frac{\partial^2}{\partial z^2}g_t(S_t,\,c,\,1)\right]+o(h_n^2),
    \end{align*}
where in the last step we used $(T_n+1)\gamma^{T_n+1}=o(h_n^3)$ to bound the tail sums. Using above display in conjunction with \eqref{eq:tail-bias-proof}, we arrive at the following.
\begin{align*}
&\sum_{t=0}^{T_n} \E_{\pi_c}\left[K_{h_n}(Z_t)\,\ind{Z_t\ge c}(U_t(h_n))^j\,\Gamma^R_t\right]\\
   &=\sum_{t=0}^{T} \gamma^t \,\E_{\pi_c}\left[\kappa_j g_t(S_t,\, c,\,1) + \kappa_{j+1}\, h_n\,\frac{\partial}{\partial z}g_t(S_t,\,c,\,1) + \kappa_{j+2}\,\frac{h_n^2}{2}\, \frac{\partial^2}{\partial z^2}g_t(S_t,\,c,\,1)\right]+o(h_n^2). \numberthis\label{eq:right-second-correct}
\end{align*}
The left-side calculation is similar (with $z=c-h_nu$ and $U_t(h_n)=-u$), and yields
\begin{align*}
       & \sum_{t=0}^{T_n} \E_{\pi_c}\left[K_{h_n}(Z_t)\,\ind{Z_t< c}(U_t(h_n))^j\,\Gamma^R_t\right]\\
        &=(-1)^j\sum_{t=0}^{T} \gamma^t \, \E_{\pi_c}\left[\kappa_j g_t(S_t,\, c,\,0) - \kappa_{j+1}\, h_n\,\frac{\partial}{\partial z}g_t(S_t,\,c,\,0) + \kappa_{j+2}\,\frac{h_n^2}{2}\, \frac{\partial^2}{\partial z^2}g_t(S_t,\,c,\,0)\right]+o(h_n^2).\numberthis\label{eq:left-second-correct}
    \end{align*}
The last two displays also hold when $\Gamma^R_t$ is replaced with $1$ and $g_t$ is replace with $f_t$, essentially by the same argument as above.

\medskip

Now recall from \cref{lem:fe-representation} that the population-level optimizer in \eqref{eq:def-LLR-popln} satisfies 
\begin{align*}
\tau^R_n
&=
e_1^\top M_n^{-1}v_n^R, \\
    M_n &:= 
\sum_{t=0}^{T_n}\gamma^t\,
\E_{\pi_c}\!\left[
K_{h_n}(Z_t)
\{X_t(h_n)-\mu_t(h_n)\}^{\otimes 2}
\right],\numberthis\label{eq:bias-tau-rep}\\
 v_n^R &:= 
\sum_{t=0}^{T_n}\gamma^t\,
\E_{\pi_c}\!\left[
K_{h_n}(Z_t)
\{X_t(h_n)-\mu_t(h_n)\}
\Gamma^R_t
\right].
\end{align*}
We will use \eqref{eq:right-second-correct}, \eqref{eq:left-second-correct} (and their analogues with $\Gamma_t^R$ replaced with $1$ and $g_t$ replaced with $f_t$) repeatedly. First, note that the time-specific mean $\mu_t(h_n)$ (defined in \eqref{eq:def-scaling}) satisfies
\[
\sum_{t=0}^{T_n}\gamma^t\,
\E_{\pi_c}\!\left[
K_{h_n}(Z_t)
\left\|
\mu_t(h_n)-
\frac{1}{2\kappa_0}(\kappa_0, 0, \kappa_1)^\top
\right\|
\right]
=O(h_n).
\]
To see this, note that the second-order expansions (\eqref{eq:right-second-correct}, \eqref{eq:left-second-correct} with
$\Gamma_t^R$ replaced by $1$) imply that the denominator
$\E_{\pi_c}[K_{h_n}(Z_t)]$ is
$2\kappa_0\,\E_{\pi_c}[f_t(c\mid S_t)]+O(h_n^2)$, while the three numerator
moments corresponding to $A_t$, $U_t(h_n)$ and $A_tU_t(h_n)$ are respectively
\[
\kappa_0\,\E_{\pi_c}[f_t(c\mid S_t)]+O(h_n),\qquad
O(h_n),\qquad
\kappa_1\,\E_{\pi_c}[f_t(c\mid S_t)]+O(h_n).
\]
Summing over $t$ with weights $\gamma^t$ gives the above display. Similarly, 
\[M_n
=
F(c)
 M
+O(h_n),\] 
where
\[M=
\begin{pmatrix}
\kappa_0 & \kappa_1 & \kappa_1\\
\kappa_1 & 2\kappa_2 & \kappa_2\\
\kappa_1 & \kappa_2 & \kappa_2
\end{pmatrix}
-
\frac{1}{2\kappa_0}
\begin{pmatrix}
\kappa_0\\[0.15em]
0\\[0.15em]
\kappa_1
\end{pmatrix}
\begin{pmatrix}
\kappa_0 & 0 & \kappa_1
\end{pmatrix},\quad F(z):=\sum_{t=0}^T\gamma^t\,\E_{\pi_c}\!\left[f_t(z\mid S_t)\right].
\]
To see this, note that before subtracting the time-specific weighted mean,
the leading right-side contribution to
$\E_{\pi_c}[K_{h_n}(Z_t)X_t(h_n)^{\otimes2}]$ is
\[
\E_{\pi_c}[f_t(c\mid S_t)]
\begin{pmatrix}
\kappa_0 & \kappa_1 & \kappa_1\\
\kappa_1 & \kappa_2 & \kappa_2\\
\kappa_1 & \kappa_2 & \kappa_2
\end{pmatrix},
\]
whereas the leading left-side contribution is
\[
\E_{\pi_c}[f_t(c\mid S_t)]
\begin{pmatrix}
0 & 0 & 0\\
0 & \kappa_2 & 0\\
0 & 0 & 0
\end{pmatrix}.
\]
The leading contribution from subtracting
$\mu_t(h_n)^{\otimes2}\E_{\pi_c}[K_{h_n}(Z_t)]$ is
\[
\frac{1}{(2\kappa_0)^2}
\begin{pmatrix}
\kappa_0\\[0.15em]
0\\[0.15em]
\kappa_1
\end{pmatrix}
\begin{pmatrix}
\kappa_0 & 0 & \kappa_1
\end{pmatrix}\cdot 2\kappa_0\,\E_{\pi_c}[f_t(c\mid S_t)]=
\frac{1}{2\kappa_0)}
\begin{pmatrix}
\kappa_0\\[0.15em]
0\\[0.15em]
\kappa_1
\end{pmatrix}
\begin{pmatrix}
\kappa_0 & 0 & \kappa_1
\end{pmatrix}\E_{\pi_c}[f_t(c\mid S_t)].
\]
Combining the last three displays and summing over $t$ gives the desired expansion for
$M_n$.

\medskip

We now use the above expansions to derive the second-order term in
$\tau_n^R$ as expressed in \eqref{eq:bias-tau-rep}.  Since
\[
\sum_{t=0}^{T_n}\gamma^t
\E_{\pi_c}\!\left[
K_{h_n}(Z_t)\{X_t(h_n)-\mu_t(h_n)\}
\right]=0,
\]
we may subtract constants inside $v_n^R$.  In particular, with $\mu_{R,a}$ is defined in \eqref{def-muRz-duplicate},
\[
\tau_n^R-\tau^R
=
e_1^\top M_n^{-1}\wt v_n^R,\qquad 
\wt v_n^R:=
v_n^R
-
M_n
\begin{pmatrix}
\tau^R\\[0.15em]
h_n\mu_{R,0}'(c)\\[0.15em]
h_n\{\mu_{R,1}'(c)-\mu_{R,0}'(c)\}
\end{pmatrix},
\]
because the first coordinate of the vector $M_n^{-1}\wt v_n^R$ is still $\tau^R$. Here $\mu'_{R,a}$ is a shorthand for $\frac{\partial}{\partial z}\mu_{R,a}$ and likewise for $\mu''_{R,a}$.
Since
$\tau^R=\mu_{R,1}(c)-\mu_{R,0}(c)$, it follows from \eqref{eq:bias-tau-rep} by algebra that
\begin{align*}
\wt v_n^R&= 
\sum_{t=0}^{T_n}\gamma^t\,
\E_{\pi_c}\!\left[
K_{h_n}(Z_t)
\{X_t(h_n)-\mu_t(h_n)\}
\Gamma^R_t
\right]\\
&\qquad - 
\sum_{t=0}^{T_n}\gamma^t\,
\E_{\pi_c}\!\left[
K_{h_n}(Z_t)
\{X_t(h_n)-\mu_t(h_n)\}^{\otimes2}
\right]\begin{pmatrix}
\tau^R\\[0.15em]
h_n\mu_{R,0}'(c)\\[0.15em]
h_n\{\mu_{R,1}'(c)-\mu_{R,0}'(c)\}
\end{pmatrix}\\
&=\sum_{t=0}^{T_n}\gamma^t\,
\E_{\pi_c}\!\Bigg[
K_{h_n}(Z_t)\{X_t(h_n)-\mu_t(h_n)\}
\Big\{
\Gamma_t^R-\mu_{R,0}(c)-\tau^R A_t
\\[-0.2em]
&\hspace{10em}
-h_n\mu_{R,0}'(c)U_t(h_n)
-h_n\{\mu_{R,1}'(c)-\mu_{R,0}'(c)\}A_tU_t(h_n)
\Big\}
\Bigg].
\end{align*}
On the right side of the cutoff, $A_t=1$ and $z-c=h_nU_t(h_n)$, and hence
the conditional mean of the expression in braces is
\[
Q_{c,t}^R(S_t,z,1)-\mu_{R,1}(c)-\mu_{R,1}'(c)(z-c).
\]
Multiplying by $f_t(z\mid S_t)$ and summing over $t$ with weights $\gamma^t$ gives
\[
F(z)\mu_{R,1}(z)
-
F(z)\mu_{R,1}(c)
-
F(z)\mu_{R,1}'(c)(z-c)
=
F(z)\{\mu_{R,1}(z)-\mu_{R,1}(c)-\mu_{R,1}'(c)(z-c)\}.
\]
By the second-order smoothness assumption,
\[
\mu_{R,1}(z)-\mu_{R,1}(c)-\mu_{R,1}'(c)(z-c)
=
\frac{1}{2}\mu_{R,1}''(c)(z-c)^2+o((z-c)^2),
\]
and the preceding localization bounds imply that this remainder remains
$o(h_n^2)$ after integration against the kernel and summation over $t$.
Also $F(z)=F(c)+O(|z-c|)$ locally, so replacing $F(z)$ by $F(c)$ in the
second-order term only changes the expression by $O(|z-c|^3)$, and hence by
$o(h_n^2)$ after kernel localization.  Thus,
\begin{align*}
& \sum_{t=0}^{T_n}\gamma^t\,
\E_{\pi_c}\!\Bigg[
K_{h_n}(Z_t)\ind{Z_{t}\ge c}\{X_t(h_n)-\mu_t(h_n)\}
\Big\{
\Gamma_t^R-\mu_{R,0}(c)-\tau^R A_t
\\[-0.2em]
&\hspace{10em}
-h_n\mu_{R,0}'(c)U_t(h_n)
-h_n\{\mu_{R,1}'(c)-\mu_{R,0}'(c)\}A_tU_t(h_n)
\Big\}
\Bigg]\\
&=\sum_{t=0}^{T_n}\gamma^t\,
\E_{\pi_c}\!\Bigg[
K_{h_n}(Z_t)\ind{Z_{t}\ge c}\{X_t(h_n)-\mu_t(h_n)\}
\frac{h_n^2}{2}\mu_{R,1}''(c)(Z_{i,t}-c)^2
\Bigg]+o(h_n^2)\\
&=\frac{h_n^2}{2}F(c)\mu_{R,1}''(c)
\left\{
\begin{pmatrix}
\kappa_2\\[0.15em]
\kappa_3\\[0.15em]
\kappa_3
\end{pmatrix}
-
\frac{\kappa_2}{2\kappa_0}
\begin{pmatrix}
\kappa_0\\[0.15em]
0\\[0.15em]
\kappa_1
\end{pmatrix}
\right\}
+o(h_n^2).\tag{Using \eqref{eq:right-second-correct}, \eqref{eq:left-second-correct} and their analogues with $\Gamma_t^R$ replaced with $1$}
\end{align*}
Here the first vector comes from integrating $u^2(1,u,u)^\top K(u)$ over
$[0,1]$, while the second vector is the fixed-effect projection term: The
leading value of $\mu_t(h_n)$ is
$(2\kappa_0)^{-1}(\kappa_0,0,\kappa_1)^\top$, and
$\kappa_2=\int_0^1u^2K(u)\,du$.

The left side is identical except that $A_t=0$ and $z-c=-h_nu$.  Its
conditional mean after subtracting the first-order approximation is
\[
Q_{c,t}^R(S_t,z,0)-\mu_{R,0}(c)-\mu_{R,0}'(c)(z-c),
\]
and the same argument gives the left-side contribution
\[
\frac{h_n^2}{2}F(c)\mu_{R,0}''(c)
\left\{
\begin{pmatrix}
0\\[0.15em]
-\kappa_3\\[0.15em]
0
\end{pmatrix}
-
\frac{\kappa_2}{2\kappa_0}
\begin{pmatrix}
\kappa_0\\[0.15em]
0\\[0.15em]
\kappa_1
\end{pmatrix}
\right\}
+o(h_n^2).
\]
Combining the right and left contributions, we obtain
\begin{align}
&\wt v_n^R=
\frac{h_n^2}{2}F(c)\mu_{R,1}''(c)
\left\{
\begin{pmatrix}
\kappa_2\\[0.15em]
\kappa_3\\[0.15em]
\kappa_3
\end{pmatrix}
-
\frac{\kappa_2}{2\kappa_0}
\begin{pmatrix}
\kappa_0\\[0.15em]
0\\[0.15em]
\kappa_1
\end{pmatrix}
\right\}
+
\frac{h_n^2}{2}F(c)\mu_{R,0}''(c)
\left\{
\begin{pmatrix}
0\\[0.15em]
-\kappa_3\\[0.15em]
0
\end{pmatrix}
-
\frac{\kappa_2}{2\kappa_0}
\begin{pmatrix}
\kappa_0\\[0.15em]
0\\[0.15em]
\kappa_1
\end{pmatrix}
\right\}
+o(h_n^2).
\label{eq:bias-vector-expansion}
\end{align}
It remains to apply the inverse of the design matrix $M_n$.  Since
$M_n=F(c)M+O(h_n)$ and $F(c)>0$, we have
\[
M_n^{-1}=F(c)^{-1}M^{-1}+O(h_n).
\]
The $O(h_n)$ part contributes $o(h_n^2)$ when multiplied by the right side of
\eqref{eq:bias-vector-expansion}.  Hence
\begin{align*}
\tau_n^R-\tau^R
&=
\frac{h_n^2}{2}\mu_{R,1}''(c)\,
e_1^\top M^{-1}
\left\{
\begin{pmatrix}
\kappa_2\\[0.15em]
\kappa_3\\[0.15em]
\kappa_3
\end{pmatrix}
-
\frac{\kappa_2}{2\kappa_0}
\begin{pmatrix}
\kappa_0\\[0.15em]
0\\[0.15em]
\kappa_1
\end{pmatrix}
\right\}
\\
&\quad+
\frac{h_n^2}{2}\mu_{R,0}''(c)\,
e_1^\top M^{-1}
\left\{
\begin{pmatrix}
0\\[0.15em]
-\kappa_3\\[0.15em]
0
\end{pmatrix}
-
\frac{\kappa_2}{2\kappa_0}
\begin{pmatrix}
\kappa_0\\[0.15em]
0\\[0.15em]
\kappa_1
\end{pmatrix}
\right\}
+o(h_n^2).
\end{align*}
Now
\[
e_1^\top M^{-1}
=
\frac{1}{\kappa_0\kappa_2-\kappa_1^2}
\begin{pmatrix}
2\kappa_2 & -\kappa_1 & 0
\end{pmatrix}.
\]
Therefore,
\[
e_1^\top M^{-1}
\left\{
\begin{pmatrix}
\kappa_2\\[0.15em]
\kappa_3\\[0.15em]
\kappa_3
\end{pmatrix}
-
\frac{\kappa_2}{2\kappa_0}
\begin{pmatrix}
\kappa_0\\[0.15em]
0\\[0.15em]
\kappa_1
\end{pmatrix}
\right\}
=
\frac{\kappa_2^2-\kappa_1\kappa_3}
{\kappa_0\kappa_2-\kappa_1^2}
=\xi_1,
\]
and
\[
e_1^\top M^{-1}
\left\{
\begin{pmatrix}
0\\[0.15em]
-\kappa_3\\[0.15em]
0
\end{pmatrix}
-
\frac{\kappa_2}{2\kappa_0}
\begin{pmatrix}
\kappa_0\\[0.15em]
0\\[0.15em]
\kappa_1
\end{pmatrix}
\right\}
=
-\frac{\kappa_2^2-\kappa_1\kappa_3}
{\kappa_0\kappa_2-\kappa_1^2}
=-\xi_1.
\]
Substituting these two identities into the previous display gives
\[
\tau_n^R-\tau^R
=
\frac{h_n^2}{2}\xi_1
\{\mu_{R,1}''(c)-\mu_{R,0}''(c)\}
+o(h_n^2),
\]
which completes the proof.
\end{proof}

\begin{lemma}[Asymptotic linearity]
    \label{lem:asymp_linearity}
    For $R\in\{Y,A\}$, define the population-level optimizer $\tau^R_n$ as in \eqref{eq:def-LLR-popln}, and its limit $\tau^R$ as in \eqref{eq:def-taustar}. 
    Under the assumptions of \cref{clt-for-twice-discounted-llr}, the empirical optimizer $\wh \tau^R_n$ solving the local linear regression \eqref{eq:def-LLR} satisfies the following. 
    \begin{align*}
         \sqrt{nh_n}\left(\wh\tau^R_n-\tau^R_n\right)=\frac{1}{\sqrt{n}}\sum_{i=1}^n \psi_{n,i}^R + \o(1),
     \end{align*}
     where 
\begin{align*}
\psi_{n,i}^R
&:=\sqrt{h_n}\,e_1^\top M_n^{-1}
\sum_{t=0}^{T_n}\gamma^t\,K_{h_n}(Z_{i,t})\{X_{i,t}(h_n)-\mu_t(h_n)\}
\left\{\Gamma_{i,t}^R-\alpha_{n,t}^R-(\zeta_n^R)^\top X_{i,t}(h_n)\right\},\\
M_n&:=\sum_{t=0}^{T_n}\gamma^t\,
\E_{\pi_c}\left[K_{h_n}(Z_t)\{X_t(h_n)-\mu_t(h_n)\}^{\otimes2}\right], \qquad \mu_t(h)
  :=\frac{\E_{\pi_c} \left[K_h(Z_{t})X_{t}(h)\right]}
          {\E_{\pi_c}\left[ K_h(Z_{t})\right]},\\
\zeta_n^R &:= M_n^{-1} v_n^R, 
\qquad\qquad\qquad\qquad v_n^R := 
\sum_{t=0}^{T_n}\gamma^t\,
\E_{\pi_c}\!\left[
K_{h_n}(Z_t)
\{X_t(h_n)-\mu_t(h_n)\}\,
\Gamma^R_t
\right].\\
\alpha_{n,t}^R &:= 
\frac{\E_{\pi_c}\left[K_{h_n}(Z_t)\{\Gamma_t^R-(\zeta_n^R)^\top X_t(h_n)\}\right]}
     {\E_{\pi_c}[K_{h_n}(Z_t)]},\tag{and zero if the denominator is zero}
\end{align*}
and $K_{h}$ and $X_{i,t}(h_n)$ are defined in \eqref{eq:def-scaling}.
Moreover, $\E_{\pi_c}[\psi_{n,i}^R]=0$, and $M_n=F(c) M + O(h_n)$ where \[M=
\begin{pmatrix}
\kappa_0 & \kappa_1 & \kappa_1\\
\kappa_1 & 2\kappa_2 & \kappa_2\\
\kappa_1 & \kappa_2 & \kappa_2
\end{pmatrix}
-
\frac{1}{2\kappa_0}
\begin{pmatrix}
\kappa_0\\[0.15em]
0\\[0.15em]
\kappa_1
\end{pmatrix}
\begin{pmatrix}
\kappa_0 & 0 & \kappa_1
\end{pmatrix},\quad F(z):=\sum_{t=0}^T\gamma^t\,\E_{\pi_c}\!\left[f_t(z\mid S_t)\right].
\]
\end{lemma}

\begin{proof} Denote by $\wh M_n$ and $\wh v_n^R$ the empirical design matrix and empirical cross-moment, respectively. That is,
    \begin{align*}
    \wh M_n &:=\frac{1}{n}\sum_{i=1}^n
\sum_{t=0}^{T_n}\gamma^t\,
K_{h_n}(Z_{i,t})
\{X_{i,t}(h_n)-\overline X_t(h_n)\}^{\otimes 2},\\ 
\wh v_n^R &:=\frac{1}{n}\sum_{i=1}^n
\sum_{t=0}^{T_n}\gamma^t\,
K_{h_n}(Z_{i,t})
\{X_{i,t}(h_n)-\overline X_t(h_n)\}
\Gamma^R_{i,t}.
\end{align*}
Then, it follows from \cref{lem:fe-representation} that
\begin{align}
 \wh\tau_n^R =e_1^\top   \wh \zeta_n^R,\qquad \wh \zeta_n^R = \wh M_n^{-1} \wh v_n^R,\qquad \tau_n^R =e_1^\top \zeta_n^R,\qquad \zeta_n^R = (M_n)^{-1} v_n^R.\label{defs}
\end{align}
We showed in the proof of \cref{lemma:concentration} that $\wh M_n^{-1}- M_n^{-1}=\o(1)$ and $\wh v_n^R - v_n^R=\o(1)$. We will show now that
\begin{equation}
    \label{eq:asymp-lin-claim}
    \wh v_n - \wh M_n\zeta_n^R=\O((nh_n)^{-1/2}).
\end{equation}
First observe that by algebra,
\begin{align*}
    &\wh v_n - \wh M_n\zeta_n^R\\
    &=\frac{1}{n}\sum_{i=1}^n\sum_{t=0}^{T_n}\gamma^t\,K_{h_n}(Z_{i,t})\{X_{i,t}(h_n)-\overline X_t(h_n)\}\left\{\Gamma_{i,t}^R-(\zeta_n^R)^\top (X_{i,t}(h_n)-\overline X_t(h_n))\right\}\\
    &=\frac{1}{n}\sum_{i=1}^n \sum_{t=0}^{T_n}\gamma^t\,K_{h_n}(Z_{i,t})\{X_{i,t}(h_n)-\overline X_t(h_n)\}\left\{\Gamma_{i,t}^R-\alpha_{n,t}^R-(\zeta_n^R)^\top X_{i,t}(h_n)\right\}, 
\end{align*}
since $\sum_{i=1}^n K_{h_n}(Z_{i,t})(X_{i,t}(h_n)-\overline X_t(h_n))=0$ by definition. Next, we aim to replace $\overline X_t(h_n)$ with $\mu_t(h_n)$ in the last display. To see why, note that using the last display and algebra,
\begin{align*}
    &\wh v_n - \wh M_n\zeta_n^R-\frac{1}{n}\sum_{i=1}^n \sum_{t=0}^{T_n}\gamma^t\,K_{h_n}(Z_{i,t})\{X_{i,t}(h_n)-\mu_t(h_n)\}\left\{\Gamma_{i,t}^R-\alpha_{n,t}^R-(\zeta_n^R)^\top X_{i,t}(h_n)\right\}\\
    &=\sum_{t=0}^{T_n}\gamma^t\,\{\mu_t(h_n)-\overline X_t(h_n)\}\frac{1}{n}\sum_{i=1}^n K_{h_n}(Z_{i,t})\left\{\Gamma_{i,t}^R-\alpha_{n,t}^R-(\zeta_n^R)^\top X_{i,t}(h_n)\right\}.
\end{align*}
For each $t$, both factors in the above summands are
$\O((nh_n)^{-1/2})$ (using the concentration arguments used in the proof of \cref{lemma:concentration} and the localization arguments for the population part exactly as in the proof of \cref{lem:bias}).  Moreover, summability is justified by the integrable envelopes as in \cref{assump5:regularity-for-consistency,assump6:smoothness} and the dominated convergence theorem. Therefore, the right-hand side in the above display is
$\O((nh_n)^{-1})=\o((nh_n)^{-1/2})$. 
Thus,
\begin{multline}
    \label{eq:expansion-with-mu-t}
    \wh v_n - \wh M_n\zeta_n^R=\frac{1}{n}\sum_{i=1}^n \sum_{t=0}^{T_n}\gamma^t\,K_{h_n}(Z_{i,t})\{X_{i,t}(h_n)-\mu_t(h_n)\}\left\{\Gamma_{i,t}^R-\alpha_{n,t}^R-(\zeta_n^R)^\top X_{i,t}(h_n)\right\}\\+\o((nh_n)^{-1/2}).
\end{multline}
    Next, observe that it follows from the profiled population-level normal equations \eqref{eq:normal-eqn} that 
    \begin{equation}
        \label{psi-has-mean-zero}
        \E_{\pi_c}\left[\sum_{t=0}^{T_n}\gamma^t\,K_{h_n}(Z_{i,t})\{X_{i,t}(h_n)-\mu_t(h_n)\}\left\{\Gamma_{i,t}^R-\alpha_{n,t}^R-(\zeta_n^R)^\top X_{i,t}(h_n)\right\}\right]=0.
    \end{equation}
On the other hand, each summand has variance of order $h_n^{-1}$ because $K_{h_n}^2$ integrates to order
$h_n^{-1}$ on the $h_n$-neighborhood of the cutoff, and the series over $t$ is summable by
\cref{assump5:regularity-for-consistency,assump6:smoothness}; see \eqref{eq:uncentered-cross-conc} and \eqref{eq:uncentered-design-conc} in the proof of \cref{lemma:concentration}. This completes the proof of \eqref{eq:asymp-lin-claim}.

\medskip

We now return to the main proof. Using algebra, we deduce from the definitions \eqref{defs} that
\begin{align*}
    \wh\tau_n^R-\tau_n^R &= e_1^\top \left(\wh M_n^{-1} \wh v_n^R - \zeta_n^R\right)= e_1^\top \wh M_n^{-1}\left( \wh v_n^R - \wh M_n\zeta_n^R\right) \\
   &= e_1^\top  M_n^{-1}\left(\wh v_n^R - \wh M_n\zeta_n^R\right) + e_1^\top \left(\wh M_n^{-1} -M_n^{-1}\right)\left(\wh v_n^R - \wh M_n\zeta_n^R\right).
\end{align*}
Using \eqref{eq:asymp-lin-claim} and $\wh M_n^{-1}- M_n^{-1}=\o(1)$ from the proof of \cref{lemma:concentration}, we can say that the second term in the above display is $\o((nh_n)^{-1/2})$. Therefore,
\begin{align*}
    \sqrt{nh_n}\left(\wh\tau_n^R-\tau_n^R\right) &= e_1^\top M_n^{-1}\sqrt{nh_n}\left(\wh v_n^R - \wh M_n\zeta_n^R\right) + \o(1).
     \end{align*}
     This combined with \eqref{eq:expansion-with-mu-t} finishes the proof of the desired asymptotic linear expansion. Finally, $\E_{\pi_c}[\psi_{n,i}^R]=0$ follows from \eqref{psi-has-mean-zero}.
\end{proof}

\begin{lemma}\label{bivariate-clt-for-twice-discounted-llr}
     Suppose that \cref{assump:data-collected-under-thresholding-policy,assump:g-formula,assump:condtional-density,assump:continuous-Q,assump5:regularity-for-consistency,assump6:smoothness} hold, and that the local linear regression \eqref{eq:def-LLR} is run with bandwidth $h_n=O(n^{-1/5})$. If $T=\infty$, assume further that $\gamma<1$ and that $T_n\ge (3/5+\eps)(\log n)/(\log \gamma^{-1})$ for all large $n$, for some $\eps>0$.
     Then the numerator $\wh{\tau}^Y_n=\wh{\tau}^Y(h_n)$ and the denominator $\wh{\tau}^A_n=\wh{\tau}^A(h_n)$ of the twice-discounted local linear regression \eqref{eq:def-LLR} satisfy the following.
     $$
\sqrt{n h_n}
\left(\begin{matrix}\wh{\tau}^Y_n - \tau^Y
-\frac{1}{2}h_n^2\,\xi_1 (\mu_{Y,1}''(c)-\mu_{Y,0}''(c))\\[2mm]
\wh{\tau}^A_n - \tau^A
-\frac{1}{2}h_n^2\,\xi_1 (\mu_{A,1}''(c)-\mu_{A,0}''(c))
\end{matrix}\right) \stackrel{d}{\longrightarrow} \mathcal{N}\left(\begin{pmatrix}0\\ 0 \end{pmatrix}, \begin{pmatrix}
V_{YY} & V_{YA}\\V_{YA} & V_{AA}
\end{pmatrix}\right),
$$
where $\xi_1:=(\kappa_2^2-\kappa_1\kappa_3)/(\kappa_0\kappa_2-\kappa_1^2)$ with $\kappa_j:=\int_0^1 u^j K(u)du$, 
\begin{equation*}
    \mu_{R,a}(z):=\frac{\sum_{t=0}^T \gamma^t\,\E_{\pi_c}\left[Q_{c,t}^R(S_t,\,z,\,a)\,f_t(z\mid S_t)\right]}{\sum_{t=0}^T \gamma^t\,\E_{\pi_c}\left[f_t(z\mid S_t)\right]}, \quad a=0,1,\ R\in\{Y,A\}.
\end{equation*}
and the variance entries $V_{YY}$, $V_{AA}$ and $V_{YA}$ as defined in \cref{clt-for-twice-discounted-llr}.
\end{lemma}

\begin{proof} We use the second-order bias expansion from \cref{lem:bias}, the asymptotic linear representation from \cref{lem:asymp_linearity}, and apply the Lindeberg--Feller central limit theorem for triangular arrays.
First note that \cref{lem:bias} takes care of the bias term: For each \(R\in\{Y,A\}\),
\[
  \tau_n^R-\tau^R
  =\frac12h_n^2\xi_1\{\mu_{R,1}''(c)-\mu_{R,0}''(c)\}+o(h_n^2).
\]
Since \(h_n=O(n^{-1/5})\), \(\sqrt{nh_n}\,o(h_n^2)=o(1)\). 
On the other hand, \cref{lem:asymp_linearity} gives
\[
  \sqrt{nh_n}(\wh\tau_n^R-\tau_n^R)
  =\frac1{\sqrt n}\sum_{i=1}^n\psi_{n,i}^R+\o(1),
  \qquad R\in\{Y,A\},
\]
where
\[
\psi_{n,i}^R
:=
\sqrt{h_n}\,e_1^\top M_n^{-1}
\sum_{t=0}^{T_n}\gamma^tK_{h_n}(Z_{i,t})
\{X_{i,t}(h_n)-\mu_t(h_n)\}
\{\Gamma_{i,t}^R-\alpha_{n,t}^R-(\zeta_n^R)^\top X_{i,t}(h_n)\},
\]
\[
M_n:=\sum_{t=0}^{T_n}\gamma^t
\E_{\pi_c}\!\left[K_{h_n}(Z_t)\{X_t(h_n)-\mu_t(h_n)\}^{\otimes2}\right],
\qquad
\mu_t(h):=\frac{\E_{\pi_c}\{K_h(Z_t)X_t(h)\}}{\E_{\pi_c}\{K_h(Z_t)\}},
\]
\[
\zeta_n^R:=M_n^{-1}\sum_{t=0}^{T_n}\gamma^t
\E_{\pi_c}\!\left[K_{h_n}(Z_t)\{X_t(h_n)-\mu_t(h_n)\}\Gamma_t^R\right],
\]
and
\[
\alpha_{n,t}^R
:=
\frac{
\E_{\pi_c}\!\left[K_{h_n}(Z_t)\{\Gamma_t^R-(\zeta_n^R)^\top X_t(h_n)\}\right]
}{\E_{\pi_c}\{K_{h_n}(Z_t)\}},                        
\]
where the convention is to set the numerator as zero when the denominator is zero. \cref{lem:asymp_linearity} also gives $\E_{\pi_c}\{\psi_{n,i}^R\}=0$.
Thus it remains to show that
\[
  \frac1{\sqrt n}\sum_{i=1}^n
  \begin{pmatrix}\psi_{n,i}^Y\\ \psi_{n,i}^A\end{pmatrix}
  \dto
  N\!\left(\begin{pmatrix}
      0\\0
  \end{pmatrix},
  \begin{pmatrix}V_{YY}&V_{YA}\\ V_{YA}&V_{AA}\end{pmatrix}
  \right).                                            
\]
Because the rows are iid across units, the above follows from the Lindeberg--Feller CLT once we prove covariance convergence and the vector Lindeberg condition.

\medskip

We record a few limits from the proofs of \cref{lemma:localization,lemma:concentration,lem:bias,lem:asymp_linearity} that will be used in the covariance and Lindeberg calculations.  First, recall from \eqref{eq:limit-of-mu-t} that for each fixed \(t\) with \(\E_{\pi_c}\{f_t(c\mid S_t)\}>0\),
\begin{equation}
  \mu_t(h_n)\to
  \frac1{2\kappa_0}\begin{pmatrix}\kappa_0,
  \,0,\,\kappa_1\end{pmatrix}.    \label{eq:clt-6}       
\end{equation}
 If \(\E_{\pi_c}\{f_t(c\mid S_t)\}=0\), set \(\mu_t(h_n)=0\) (the corresponding terms will anyway be irrelevant because for these time indices, we have $f_t(c\mid S_t)=0$ almost surely).
Second, 
\begin{equation}
  M_n=F(c)M+O(h_n),
  \qquad M_n^{-1}=F(c)^{-1}M^{-1}+O(h_n).                         \label{eq:clt-8}       
\end{equation}
In particular, \(\norm{M_n^{-1}}\leq C\) for all large \(n\).  Also, from the expansion in the proof of \cref{lem:bias}, 
\begin{equation}
\zeta_n^R
=
\begin{pmatrix}
\tau^R,\, h_n\mu_{R,0}'(c),\, h_n\{\mu_{R,1}'(c)-\mu_{R,0}'(c)\}
\end{pmatrix}^\top
+O(h_n^2)\to\begin{pmatrix}
    \tau^R ,\, 0,\, 0
\end{pmatrix}^\top,
\ \ R\in\{Y,A\}.                                                 \label{eq:clt-9} 
\end{equation}
Lastly, we show that the profiled population fixed effects converge to \(\alpha_t^R\) (defined in \cref{clt-for-twice-discounted-llr}). For the numerator, split on the two sides of the cutoff.  On the right side, using \(z=c+h_nu\), \(A_t=1\), and \(X_t(h_n)=(1,u,u)^\top\),
\begin{align*}
&\E_{\pi_c}\!\left[K_{h_n}(Z_t)\1\{Z_t\geq c\}
\{\Gamma_t^R-(\zeta_n^R)^\top X_t(h_n)\}\right] \\
&\quad=
\int_0^1 K(u)\E_{\pi_c}\!\left[
\{Q^R_{c,t}(S_t,c+h_nu,1)-\zeta_{n,1}^R-u\zeta_{n,2}^R-u\zeta_{n,3}^R\}
 f_t(c+h_nu\mid S_t)
\right]du \\[2mm]
&\quad\to
\kappa_0\,\E_{\pi_c}\!\left[
\{Q^R_{c,t}(S_t,c,1)-\tau^R\}f_t(c\mid S_t)
\right],                                                   \numberthis    \label{eq:clt-10}
\end{align*}
where we used \eqref{eq:clt-9}, the continuity of \(Q^R_{c,t}(S_t,\cdot,1)\) and \(f_t(\cdot\mid S_t)\) at \(c\), and dominated convergence with the envelopes in \cref{assump5:regularity-for-consistency}.  On the left side, using \(z=c-h_nu\), \(A_t=0\), and \(X_t(h_n)=(0,-u,0)^\top\),
\begin{align*}
&\E_{\pi_c}\!\left[K_{h_n}(Z_t)\1\{Z_t<c\}
\{\Gamma_t^R-(\zeta_n^R)^\top X_t(h_n)\}\right] \\
&\quad=
\int_0^1 K(u)\E_{\pi_c}\!\left[
\{Q^R_{c,t}(S_t,c-h_nu,0)+u\zeta_{n,2}^R\}
 f_t(c-h_nu\mid S_t)
\right]du \\[2mm]
&\quad\to
\kappa_0\,\E_{\pi_c}\!\left[
Q^R_{c,t}(S_t,c,0)f_t(c\mid S_t)
\right].                                               \numberthis\label{eq:clt-11}
\end{align*}
Combining \eqref{eq:clt-10} and \eqref{eq:clt-11} yields \(\alpha_{n,t}^R\to\alpha_t^R\) whenever \(\E_{\pi_c}\{f_t(c\mid S_t)\}>0\).  If \(\E_{\pi_c}\{f_t(c\mid S_t)\}=0\), the denominator is \(o(1)\) and the corresponding weighted contribution in the variance formula is zero.

The convergence statements above are also valid after summing over \(t\) with weights \(\gamma^t\) or \(\gamma^{2t}\).  This follows from the same envelope bounds used in \cref{lemma:localization,lemma:concentration,lem:bias,lem:asymp_linearity} where the dominating series is \(\sum_t\gamma^t\,\E\{B_{m2,t}(S_t)B_{f,t}(S_t)\}\).

\paragraph{Covariance convergence.}
Fix \(G,H\in\{Y,A\}\). Using \(\E_{\pi_c}\psi_{n,i}^G=\E_{\pi_c}\psi_{n,i}^H=0\), 
\begin{align}
\E_{\pi_c}\{\psi_{n,i}^G\,\psi_{n,i}^H\}
&=h_ne_1^\top M_n^{-1}
\E_{\pi_c}\!\Bigg[
\sum_{t=0}^{T_n}\sum_{t'=0}^{T_n}\gamma^{t+t'}
K_{h_n}(Z_t)K_{h_n}(Z_{t'})                                      \notag\\
&\hspace{22mm}\times
\{X_t(h_n)-\mu_t(h_n)\}
\{X_{t'}(h_n)-\mu_{t'}(h_n)\}^\top                                  \notag\\
&\hspace{22mm}\times
\{\Gamma_t^G-\alpha_{n,t}^G-(\zeta_n^G)^\top X_t(h_n)\}
\{\Gamma_{t'}^H-\alpha_{n,t'}^H-(\zeta_n^H)^\top X_{t'}(h_n)\}
\Bigg]M_n^{-1}e_1 .                                         \label{eq:clt-12}
\end{align}
We first show that the terms with \(t\neq t'\) are \(o(1)\).  On the support of \(K_{h_n}(Z_t)\), \(\norm{X_t(h_n)}\leq C\).  From \eqref{eq:clt-6}, \(\norm{\mu_t(h_n)}\leq C\) for all fixed \(t\) with positive limiting density, and the remaining time points have vanishing weighted contribution.  Moreover, since \(\alpha_{n,t}^R\) is the weighted mean of \(\Gamma_t^R-(\zeta_n^R)^\top X_t(h_n)\), subtracting it can only reduce the weighted second moment:
\begin{align}
&\E_{\pi_c}\!\left[K_{h_n}(Z_t)
\{\Gamma_t^R-\alpha_{n,t}^R-(\zeta_n^R)^\top X_t(h_n)\}^2\right]  
\leq
\E_{\pi_c}\!\left[K_{h_n}(Z_t)
\{\Gamma_t^R-(\zeta_n^R)^\top X_t(h_n)\}^2\right]                         \notag\\
&\qquad\qquad\qquad\qquad\qquad\qquad\leq
C\,\E_{\pi_c}\!\left[K_{h_n}(Z_t)\big\{(\Gamma_t^R)^2+1\big\}\right]
\leq C\,\E_{\pi_c}\{B_{m2,t}(S_t)B_{f,t}(S_t)\}.              \label{eq:clt-13}
\end{align}
The last inequality follows by conditioning on \(S_t\), splitting into \(A_t=1\) and \(A_t=0\), changing variables \(z=c\pm h_nu\), and using the local bounds on \(m_{2,\,t}^R\) and \(f_t\).  For \(R=A\), the same bound is immediate since \(\Gamma_t^A\) is bounded by \((1-\gamma)^{-1}\) in the infinite-horizon case and by \(T+1\) in the finite-horizon case.

For \(t<t'\), the same conditioning argument applied twice gives
\begin{align}
&\left\|
\E_{\pi_c}\!\left[
K_{h_n}(Z_t)K_{h_n}(Z_{t'})
\{X_t(h_n)-\mu_t(h_n)\}
\{X_{t'}(h_n)-\mu_{t'}(h_n)\}^\top
\right.
\right.\notag\\[-1mm]
&\hspace{50mm}\left.
\left.\times
\{\Gamma_t^G-\alpha_{n,t}^G-(\zeta_n^G)^\top X_t(h_n)\}
\{\Gamma_{t'}^H-\alpha_{n,t'}^H-(\zeta_n^H)^\top X_{t'}(h_n)\}
\right]\right\|                                                          \notag\\
&\quad\leq
C\,\E_{\pi_c}\!\left[
B_{m2,t}^{1/2}(S_t)B_{m2,t'}^{1/2}(S_{t'})B_{f,t}(S_t)B_{f,t'}(S_{t'})
\right].      \label{eq:clt-14}
\end{align}
To prove \eqref{eq:clt-14}, we first split \(Z_t\) and \(Z_{t'}\) into their two sides of the cutoff.  On a fixed pair of sides, write \(Z_t=c\pm h_nu\) and \(Z_{t'}=c\pm h_nv\).  The two factors \(h_n^{-1}\) in the kernels cancel with the two Jacobians \(h_n\,du\) and \(h_n\,dv\).  The residualized regressors are bounded on \(u,v\in[0,1]\).  Conditional Cauchy--Schwarz bounds the conditional absolute cross moment of the two residual terms by the product of the square-root local second-moment envelopes, while the two conditional densities are bounded by \(B_{f,t}(S_t)\) and \(B_{f,t'}(S_{t'})\).  The integral of \(K(u)K(v)\) over \([0,1]^2\) goes into the constant in \eqref{eq:clt-14}.  The case \(t'>t\) is identical.

Multiplying \eqref{eq:clt-14} by \(h_n\gamma^{t+t'}\), summing over \(t\neq t'\), and using the second summability condition in \cref{assump6:smoothness} gives
\begin{align}
&h_n\sum_{\substack{0\leq t,t'\leq T_n\\ t\neq t'}}\gamma^{t+t'}
\left\|
\E_{\pi_c}\!\left[
K_{h_n}(Z_t)K_{h_n}(Z_{t'})
\{X_t(h_n)-\mu_t(h_n)\}
\{X_{t'}(h_n)-\mu_{t'}(h_n)\}^\top
\varepsilon_{t,n}^G\varepsilon_{t',n}^H
\right]
\right\|                                                                  \notag\\
&\quad\leq
Ch_n\sum_{t=0}^T\sum_{t'=t+1}^T\gamma^{t+t'}
\E_{\pi_c}\!\left[
B_{m2,t}^{1/2}(S_t)B_{m2,t'}^{1/2}(S_{t'})B_{f,t}(S_t)B_{f,t'}(S_{t'})
\right]
=o(1),                                                   \label{eq:clt-15}
\end{align}
where \(\varepsilon_{t,n}^R:=\Gamma_t^R-\alpha_{n,t}^R-(\zeta_n^R)^\top X_t(h_n)\).  Thus only the diagonal terms \(t=t'\) contribute to the covariance limit.

For a fixed diagonal time point, split again on the two sides of the cutoff.  On the right side, after the change of variables \(z=c+h_nu\),
\begin{align}
&h_n\E_{\pi_c}\!\left[
K_{h_n}^2(Z_t)\1\{Z_t\geq c\}
\{X_t(h_n)-\mu_t(h_n)\}^{\otimes2}
\varepsilon_{t,n}^G\varepsilon_{t,n}^H
\right]                                                               \notag\\
&\quad=
\int_0^1K^2(u)\,
\E_{\pi_c}\!\Bigg[
\{(1,u,u)^\top-\mu_t(h_n)\}^{\otimes2}                                  \notag\\
&\hspace{18mm}\times
\E_{\pi_c}\!\left[
\begin{aligned}
&\{\Gamma_t^G-\alpha_{n,t}^G-(\zeta_n^G)^\top(1,u,u)^\top\}\\[-1mm]
&\quad\times\{\Gamma_t^H-\alpha_{n,t}^H-(\zeta_n^H)^\top(1,u,u)^\top\}
\end{aligned}
\middle| S_t,Z_t=c+h_nu,A_t=1
\right]
 f_t(c+h_nu\mid S_t)
\Bigg]du .                                                 \label{eq:clt-16}
\end{align}
Using \cref{eq:clt-6,eq:clt-9,eq:clt-10,eq:clt-11}, the cross-moment continuity, and dominated convergence, the right side of \eqref{eq:clt-16} converges to
\begin{align}
&\E_{\pi_c}\!\left[
\E_{\pi_c}\!\left[
\{\Gamma_t^G-\alpha_t^G-\tau^G\}
\{\Gamma_t^H-\alpha_t^H-\tau^H\}
\mid S_t,Z_t=c,A_t=1
\right]f_t(c\mid S_t)
\right]                                                            \notag\\
&\hspace{35mm}\times
\int_0^1K^2(u)
\left\{(1,u,u)^\top-\frac1{2\kappa_0}\begin{pmatrix}\kappa_0,0,\kappa_1\end{pmatrix}^\top\right\}^{\otimes2}du .   \label{eq:clt-17}
\end{align}
The domination follows from \eqref{eq:clt-13}, with the cross term bounded by Cauchy--Schwarz and the local cross-moment envelope.  The left side is identical except that \(z=c-h_nu\), \(A_t=0\), and \(X_t(h_n)=(0,-u,0)^\top\).  Hence
\begin{align}
&h_n\E_{\pi_c}\!\left[
K_{h_n}^2(Z_t)\1\{Z_t<c\}
\{X_t(h_n)-\mu_t(h_n)\}^{\otimes2}
\varepsilon_{t,n}^G\varepsilon_{t,n}^H
\right]                                                               \notag\\[2mm]
&\quad\to
\E_{\pi_c}\!\left[
\E_{\pi_c}\!\left[
\{\Gamma_t^G-\alpha_t^G\}
\{\Gamma_t^H-\alpha_t^H\}
\mid S_t,Z_t=c,A_t=0
\right]f_t(c\mid S_t)
\right]                                                            \notag\\
&\hspace{35mm}\times
\int_0^1K^2(u)
\left\{(0,-u,0)^\top-\frac1{2\kappa_0}\begin{pmatrix}\kappa_0,0,\kappa_1\end{pmatrix}^\top\right\}^{\otimes2}du .  \label{eq:clt-18}
\end{align}
The diagonal summands in \cref{eq:clt-17,eq:clt-18}, after multiplication by \(\gamma^{2t}\), are dominated by \(C\gamma^{2t}\E_{\pi_c}\{B_{m2,t}(S_t)B_{f,t}(S_t)\}\), which is summable because \(\gamma^{2t}\leq\gamma^t\) and \cref{assump5:regularity-for-consistency}. Therefore,  dominated convergence for series applies.

In the infinite-horizon case, the covariance calculation above may be equivalently carried out with the full future sums \(\Gamma_{t,\infty}^R\).  The error from replacing \(\Gamma_t^R\) by \(\Gamma_{t,\infty}^R\) is negligible.  Indeed, for \(0\leq t\leq T_n\),
\[
  \Gamma_{t,\infty}^R-\Gamma_t^R
  =\gamma^{T_n-t+1}\Gamma_{T_n+1,\infty}^R.
\]
Using \(\norm{X_t(h_n)-\mu_t(h_n)}\leq C\), \(K_{h_n}\leq \norm{K}_\infty/h_n\), \(\norm{M_n^{-1}}\leq C\), and \(\sup_s\E_{\pi_c}\{(\Gamma_{s,\infty}^R)^2\}<\infty\),
\begin{align*}
&\left\|
\sqrt{h_n}\sum_{t=0}^{T_n}\gamma^tK_{h_n}(Z_t)
\{X_t(h_n)-\mu_t(h_n)\}
\{\Gamma_{t,\infty}^R-\Gamma_t^R\}
\right\|_{L_2}                                                   \\
&\quad\leq
C\sqrt{h_n}\sum_{t=0}^{T_n}\gamma^t\frac1{h_n}\gamma^{T_n-t+1}
\norm{\Gamma_{T_n+1,\infty}^R}_{L_2}
\leq C\frac{(T_n+1)\gamma^{T_n+1}}{\sqrt{h_n}}
=o(1).                                                   %        \tag{19}
\end{align*}
The last equality follows because \(T_n\geq(3/5+\varepsilon)\log n/\log\gamma^{-1}\) implies \(\gamma^{T_n+1}=O(n^{-3/5-\varepsilon})\), while \(h_n=O(n^{-1/5})\), so \((T_n+1)\gamma^{T_n+1}/\sqrt{h_n}=O((\log n)n^{-1/2-\varepsilon})=o(1)\).  The omitted tail \(t>T_n\) is also negligible by the summability in \cref{assump5:regularity-for-consistency}.  Indeed, using the same envelope bound as earlier,
\[
\left\|
\sqrt{h_n}\sum_{t>T_n}\gamma^tK_{h_n}(Z_t)
\{X_t(h_n)-\mu_t(h_n)\}\Gamma_{t,\infty}^R
\right\|_{L_2}
\leq
C\sum_{t>T_n}\gamma^t
\E_{\pi_c}^{1/2}\{B_{m2,t}(S_t)B_{f,t}(S_t)\}=o(1).    %             \tag{20}
\]
Thus the observed truncated sums and the infinite-horizon sums have the same covariance limit. We thus conclude that the covariance in \eqref{eq:clt-12} without the $e_1^\top M_n^{-1}$ factor has the desired limit. Finally, recall from \eqref{eq:clt-8} that \(M_n^{-1}=F(c)^{-1}M^{-1}+o(1)\).  Also note the algebraic identities
\[
  e_1^\top M^{-1}\left\{(1,u,u)^\top-\frac1{2\kappa_0}\begin{pmatrix}\kappa_0, \,0,\,\kappa_1\end{pmatrix}^\top\right\}
  =\frac{\kappa_2-\kappa_1u}{\kappa_0\kappa_2-\kappa_1^2},   
\]
and
\[
  e_1^\top M^{-1}\left\{(0,-u,0)^\top-\frac1{2\kappa_0}\begin{pmatrix}\kappa_0, \,0,\,\kappa_1\end{pmatrix}^\top\right\}
  =-\frac{\kappa_2-\kappa_1u}{\kappa_0\kappa_2-\kappa_1^2}.   
\]
Therefore both sides of the cutoff give the same scalar quadratic factor,
\[
\int_0^1K^2(u)
\left(\frac{\kappa_2-\kappa_1u}{\kappa_0\kappa_2-\kappa_1^2}\right)^2du
=
\frac{\kappa_2^2\rho_0-2\kappa_1\kappa_2\rho_1+\kappa_1^2\rho_2}
     {(\kappa_0\kappa_2-\kappa_1^2)^2}
=\xi_2.                                                     
\]
This, in conjunction with \cref{eq:clt-12,eq:clt-15,eq:clt-17,eq:clt-18}, gives
\[
  \E_{\pi_c}\{\psi_{n,i}^G\,\psi_{n,i}^H\}\to V_{GH},
  \qquad G,H\in\{Y,A\}.                             
\]
This proves convergence of the full \(2\times2\) covariance matrix.

\paragraph{Verifying the Lindeberg condition.}
We prove that for every \(\delta>0\), and $R\in\{Y,A\}$,
\[
  \E_{\pi_c}\!\left[
  \left\|\psi_{n,i}^R\right\|^2
  \1\!\left\{
  \left\|\psi_{n,i}^R\right\|>\delta\sqrt n
  \right\}
  \right]\to0.       
\]
The \(A\)-part is straightforward because \(\Gamma_t^A\) is bounded and, on the support of the kernel,
\begin{equation*}
  \abs{\psi_{n,i}^A}
  \leq C\sqrt{h_n}\sum_{t=0}^{T_n}\gamma^t\frac1{h_n}
  \leq \frac{C}{\sqrt{h_n}}
  =o(\sqrt n),        %\label{eq:clt-26}                                 
\end{equation*}
since \(nh_n\to\infty\).  We next focus on the \(Y\)-part.
Choose numbers \(b_n\to\infty\) such that \(b_n=o(\sqrt{nh_n})\), for example \(b_n=(nh_n)^{1/4}\).  Decompose the \(Y\)-residual in $\psi_{n,i}^Y$ as
\begin{multline*}
    \Gamma_t^Y-\alpha_{n,t}^Y-(\zeta_n^Y)^\top X_t(h_n)
=
\{\Gamma_t^Y-\alpha_{n,t}^Y-(\zeta_n^Y)^\top X_t(h_n)\}
\1\!\left\{\abs{\Gamma_t^Y-\alpha_{n,t}^Y-(\zeta_n^Y)^\top X_t(h_n)}\leq b_n\right\}
\\\quad+
\{\Gamma_t^Y-\alpha_{n,t}^Y-(\zeta_n^Y)^\top X_t(h_n)\}
\1\!\left\{\abs{\Gamma_t^Y-\alpha_{n,t}^Y-(\zeta_n^Y)^\top X_t(h_n)}> b_n\right\}.
\end{multline*}
Denote the corresponding two pieces of \(\psi_{n,i}^Y\), after splitting the residuals as above, by $\psi_{n,i}^{Y,\mathrm{trunc}}$ and $\psi_{n,i}^{Y,\mathrm{tail}}$, respectively.  For the truncated piece, using \(\norm{M_n^{-1}}\leq C\), \(\norm{X_t(h_n)-\mu_t(h_n)}\leq C\), \(K_{h_n}\leq\norm K_\infty/h_n\), and \(\sum_t\gamma^t\leq C\),
\[
  \abs{\psi_{n,i}^{Y,\mathrm{trunc}}}
  \leq C\sqrt{h_n}\sum_{t=0}^{T_n}\gamma^t\frac{b_n}{h_n}
  \leq C\frac{b_n}{\sqrt{h_n}}
  =o(\sqrt n).                                         
\]
This implies that, for all sufficiently large \(n\), 
\[
\E_{\pi_c}\!
\left[\left\|\psi_{n,i}^Y\right\|^2\ind{\left\|\psi_{n,i}^Y\right\|>\delta\sqrt{n}}\right]
\leq C\,\E_{\pi_c}\left\{\left(\psi_{n,i}^{Y,\mathrm{tail}}\right)^2\right\}.             
\]
It remains to show that the right side in the above display is \(o(1)\).  Expanding the square as in the covariance calculation gives a sum over \((t,t')\).  For each fixed \((t,t')\), after splitting into sides of the cutoff and changing variables \(z=c\pm h_nu\), \(z'=c\pm h_nv\), the corresponding term is bounded by the same integrable envelope as in \eqref{eq:clt-14}, multiplied by a conditional tail second moment.  Since \(b_n\to\infty\), this conditional tail second moment converges to zero for each fixed \((t,t')\).  The envelope in \eqref{eq:clt-14}, together with the diagonal envelope in \eqref{eq:clt-13}, is summable after multiplication by \(\gamma^{t+t'}\).  Dominated convergence for the diagonal and off-diagonal sums therefore yields
\[
  \E_{\pi_c}\left\{\left(\psi_{n,i}^{Y,\mathrm{tail}}\right)^2\right\}=o(1),
\]
which completes the proof.
\end{proof}

\section{Auxiliary Results}

\begin{lemma}\label{lemma:Celini}
     Under the parametric model \eqref{eq:cellini} of \citet{cellini2010value}, and provided that the horizon is infinite $(T=\infty)$, the population target $\tauRD$ defined in \eqref{eq:tauRD} reduces to the  following: $$\tauRD=\sum_{t=0}^\infty \gamma^t\,\theta_t,$$ where $\theta_t$ is as defined in \eqref{eq:cellini}.
\end{lemma}

\begin{proof} Under \eqref{eq:cellini}, we have $$Y_{i,t+j}=\sum_{t'\le t+j}\theta_{t+j-t'} A_{i,t'}+\eps_{i,t+j}=\sum_{m=0}^{t+j} \theta_m A_{i,t+j-m}+\eps_{i,t+j}.$$
Therefore,
\begin{align*}
&Q_{c,\,t}^Y(S_{i,t},\, c,\, 1)-Q_{c,\,t}^Y(S_{i,t},\, c,\, 0)\\
   &=\sum_{j=0}^\infty \gamma^j \left(\E_{\pi_c}\left[Y_{i,t+j}\mid S_{i,t}, Z_{i,t}=c, A_{i,t}=1\right] - \E_{\pi_c}\left[Y_{i,t+j}\mid S_{i,t}, Z_{i,t}=c, A_{i,t}=0\right]\right) \\
   &=\sum_{j=0}^\infty \gamma^j \sum_{m=0}^{t+j} \theta_m \left(\E_{\pi_c}\left[A_{i,t+j-m}\mid S_{i,t}, Z_{i,t}=c, A_{i,t}=1\right] - \E_{\pi_c}\left[A_{i,t+j-m}\mid S_{i,t}, Z_{i,t}=c, A_{i,t}=0\right]\right)\\
   &=\sum_{j=0}^\infty \gamma^j \sum_{m=0}^{j} \theta_m \left(\E_{\pi_c}\left[A_{i,t+j-m}\mid S_{i,t}, Z_{i,t}=c, A_{i,t}=1\right] - \E_{\pi_c}\left[A_{i,t+j-m}\mid S_{i,t}, Z_{i,t}=c, A_{i,t}=0\right]\right)\tag{terms for $m>j$ vanish}\\
   &=\sum_{m=0}^\infty  \sum_{j\ge m} \gamma^j\theta_m \left(\E_{\pi_c}\left[A_{i,t+j-m}\mid S_{i,t}, Z_{i,t}=c, A_{i,t}=1\right] - \E_{\pi_c}\left[A_{i,t+j-m}\mid S_{i,t}, Z_{i,t}=c, A_{i,t}=0\right]\right)\tag{swapping the order of summation}\\
    &=\sum_{m=0}^\infty \gamma^m\theta_m  \sum_{s\ge 0} \gamma^{s}\left(\E_{\pi_c}\left[A_{i,t+s}\mid S_{i,t}, Z_{i,t}=c, A_{i,t}=1\right] - \E_{\pi_c}\left[A_{i,t+s}\mid S_{i,t}, Z_{i,t}=c, A_{i,t}=0\right]\right)\\
    &=\left(\sum_{m=0}^\infty \gamma^m\theta_m\right)\left(Q_{c,\,t}^A(S_{i,t},\, c,\, 1)-Q_{c,\,t}^A(S_{i,t},\, c,\, 0)\right).
\end{align*}
The proof now follows from \eqref{eq:tauRD}.
\end{proof}

\begin{lemma}\label{lemma:VRD-on-gamma}
     In addition to the conditions of \cref{clt-for-twice-discounted-llr}, assume that: 
     \begin{enumerate}[label=(\roman*)]
     \item The conditional second moments of future outcomes are uniformly bounded at the threshold, i.e., there exists an $M\in (0,\infty)$ such that for all $t\ge 0$ and $j\ge 0$,
     \(
        \E_{\pi_c}[Y_{t+j}^2\mid S_t,\,Z_t=c,\,A_t=a]\le M
     \)
     almost surely, for both $a=0,1$.
    \item The conditional densities $f_t(c\mid S_t)$ are uniformly bounded and bounded away from zero, i.e.,
    \(0<\underline{f}\le f_t(c\mid S_t)\le \overline{f}<\infty
    \)
    almost surely, for all $t\ge 0$.
     \end{enumerate}
     Define 
     \begin{equation*}
        \cg{T}  := \sum_{t=0}^T \gamma^t = \begin{cases}
         (1-\gamma^{T+1})/(1-\gamma) & \text{if $T<\infty$ and $\gamma<1$,}\\
             (1-\gamma)^{-1} & \text{if $T=\infty$ and $\gamma<1$,}\\
             T+1 & \text{if $\gamma=1$}.
         \end{cases}
     \end{equation*}
    Then, the quantities $V_{YY}$, $V_{AA}$ and $V_{YA}$ defined in \cref{clt-for-twice-discounted-llr} satisfy
    \[
        V_{YY}=O(\cgg{T}),\qquad
        V_{AA}=O(\cgg{T}),\qquad
        |V_{YA}|=O(\cgg{T}),
    \]
    where the constants do not depend on $\gamma$, and 
    \begin{equation*}
        \cgg{T}:=\frac{\sum_{t=0}^T\gamma^{2t}\cg{T-t}^2}{\cg{T}^2}.
    \end{equation*}
    Consequently, the asymptotic variance $V_{\mathrm{RD}}$ in \cref{clt-for-twice-discounted-llr} satisfies 
    \begin{equation}\label{eq:VRDonGamma}
        V_{\mathrm{RD}} =
        O\left(
        \cgg{T}\cdot
        \frac{(1+|\tauRD|)^2}{(\Delta\mu_A(c))^2}
        \right),
    \end{equation}
    where
    \[
    \Delta\mu_A(c)=
    \frac{\sum_{t=0}^T \gamma^t\,\E_{\pi_c}[\{Q_{c,\,t}^A(S_t,\, c,\, 1)-Q_{c,\,t}^A(S_t,\, c,\, 0)\}f_t(c\mid S_t)]}
    {\sum_{t=0}^T \gamma^t\,\E_{\pi_c}[f_t(c\mid S_t)]}.
    \]
    If, in addition, there exist constants $\underline q>0$ and $L<\infty$ such that, for a.e.~$s_t$ and all $t\ge0$, \(Q_{c,\,t}^A(s_t,c,1)-Q_{c,\,t}^A(s_t,c,0)\ge \underline q\), and \(\left|Q_{c,\,t}^Y(s_t,c,1)-Q_{c,\,t}^Y(s_t,c,0)\right|
        \le
        L\left\{Q_{c,\,t}^A(s_t,c,1)-Q_{c,\,t}^A(s_t,c,0)\right\}\),
    then
    \[
        V_{\mathrm{RD}}=O(\cgg{T}).
    \]
   The behavior of $\cgg{T}$ across different regimes is characterized in \cref{lemma:C2rates}.
\end{lemma}

\begin{proof}
Fix $a\in\{0,1\}$ and define
\(
    \cal{F}_{t,\,a}:=\sigma(S_t,\,Z_t=c,\,A_t=a)
\).
By conditional Minkowski's inequality, applied first to finite partial sums and then by monotone convergence when $T=\infty$,
\begin{align*}
    \left(\E_{\pi_c}\left[(\Gamma_t^Y)^2\mid \cal{F}_{t,\,a}\right]\right)^{1/2}
    &\le
    \sum_{j=0}^{T-t}\gamma^j
    \left(\E_{\pi_c}\left[Y_{t+j}^2\mid \cal{F}_{t,\,a}\right]\right)^{1/2}
    \le M^{1/2}\cg{T-t}.
\end{align*}
Thus
\begin{equation}\label{boundedGtsquared}
    \E_{\pi_c}\left[(\Gamma_t^Y)^2\mid \cal{F}_{t,\,a}\right]
    \le M\cg{T-t}^2 .
\end{equation}
Similarly, since $A_{t+j}\in\{0,1\}$,
\begin{equation}\label{boundedHtsquared}
    \E_{\pi_c}\left[(\Gamma_t^A)^2\mid \cal{F}_{t,\,a}\right]
    \le \cg{T-t}^2 .
\end{equation}
Next, conditional Cauchy--Schwarz and the density bound imply
\[
    \left|m^Y_{t,a}(c)\right|
    \le
    \frac{
    \E_{\pi_c}\left[
    \left|Q^Y_{c,t}(S_t,c,a)\right|f_t(c\mid S_t)
    \right]}
    {\E_{\pi_c}[f_t(c\mid S_t)]}
    \le
    M^{1/2}\cg{T-t}\frac{\overline f}{\underline f}.
\]
Likewise,
\[
    0\le m^A_{t,a}(c)\le \cg{T-t}\frac{\overline f}{\underline f}.
\]
Moreover, since
\(
    \sum_{t=0}^T\gamma^t\,\E_{\pi_c}[f_t(c\mid S_t)]
    \ge \underline f\,\cg{T}
\),
we also have, for $R=Y,A$,
\begin{equation*}
    |\Delta\mu_R(c)|
    \lesssim
    \frac{\sum_{t=0}^T\gamma^t\cg{T-t}}{\cg{T}},
\end{equation*}
where the constant does not depend on $\gamma$.
We now bound the fixed-effect centered residuals that appear in the variance formula. Recall that
\(
    \alpha_t^R=
    \frac12\{m^R_{t,1}(c)+m^R_{t,0}(c)-\Delta\mu_R(c)\}
\) (cf.~\cref{clt-for-twice-discounted-llr}).
Combining the last two bounds, 
\begin{equation}\label{eq:alpha-bound}
    |\alpha_t^R+a\Delta\mu_R(c)|
    \lesssim
    \cg{T-t}
    +
    \frac{\sum_{s=0}^T\gamma^s\cg{T-s}}{\cg{T}},
    \qquad R\in\{Y,A\},\ a\in\{0,1\}.
\end{equation}
Therefore, using $(x-y)^2\le 2x^2+2y^2$, \eqref{boundedGtsquared}, \eqref{boundedHtsquared} and \eqref{eq:alpha-bound}, we derive that
\begin{equation}\label{eq:centered-second-bound}
    \E_{\pi_c}\left[
    \left(\Gamma_t^R-\alpha_t^R-a\Delta\mu_R(c)\right)^2
    \mid \cal{F}_{t,a}
    \right]
    \lesssim
    \cg{T-t}^2
    +
    \left(
    \frac{\sum_{s=0}^T\gamma^s\cg{T-s}}{\cg{T}}
    \right)^2 ,
\end{equation}
for $R=Y,A$.
We next record a simple deterministic inequality. Since $t\mapsto \cg{T-t}$ is non-increasing and $t\mapsto \gamma^t$ is non-increasing,
\[
    \frac{\sum_{t=0}^T \gamma^{2t}\cg{T-t}}
         {\sum_{t=0}^T \gamma^{2t}}
    \ge
    \frac{\sum_{t=0}^T \gamma^t\cg{T-t}}
         {\sum_{t=0}^T \gamma^t}.
\]
Indeed, after cross-multiplication, the difference between the two sides is proportional to
\[
    \sum_{s<t}\gamma^{s+t}(\gamma^s-\gamma^t)
    \{\cg{T-s}-\cg{T-t}\},
\]
which is nonnegative. Hence, by Jensen's inequality,
\[
    \left(
    \frac{\sum_{t=0}^T\gamma^t\cg{T-t}}{\cg{T}}
    \right)^2
    \le
    \frac{\sum_{t=0}^T\gamma^{2t}\cg{T-t}^2}
         {\sum_{t=0}^T\gamma^{2t}}.
\]
It follows that
\begin{equation}\label{eq:weighted-average-bound}
    \frac{\sum_{t=0}^T\gamma^{2t}}
         {\cg{T}^2}
    \left(
    \frac{\sum_{s=0}^T\gamma^s\cg{T-s}}{\cg{T}}
    \right)^2
    \le
    \frac{\sum_{t=0}^T\gamma^{2t}\cg{T-t}^2}{\cg{T}^2}
    =
    \cgg{T}.
\end{equation}

We can now bound $V_{YY}$. By definition of $V_{YY}$ in \cref{clt-for-twice-discounted-llr},
\begin{align*}
    V_{YY}
    &=
    \frac{\xi_2}{F(c)^2}
    \sum_{t=0}^T\sum_{a=0}^1
    \gamma^{2t}
    \E_{\pi_c}\left[
    \E_{\pi_c}\left[
    \left(\Gamma_t^Y-\alpha_t^Y-a\Delta\mu_Y(c)\right)^2
    \mid \cal{F}_{t,a}
    \right]
    f_t(c\mid S_t)
    \right],
\end{align*}
where $F(c)=\sum_{t=0}^T\gamma^t\E_{\pi_c}[f_t(c\mid S_t)]$. Since
$F(c)\ge \underline f\,\cg{T}$ and $f_t(c\mid S_t)\le \overline f$, \eqref{eq:centered-second-bound} and \eqref{eq:weighted-average-bound} yield
\[
    V_{YY}
    \lesssim
    \frac{\sum_{t=0}^T\gamma^{2t}\cg{T-t}^2}{\cg{T}^2}
    +
    \frac{\sum_{t=0}^T\gamma^{2t}}{\cg{T}^2}
    \left(
    \frac{\sum_{s=0}^T\gamma^s\cg{T-s}}{\cg{T}}
    \right)^2
    \lesssim
    \cgg{T}.
\]
The same argument, using \eqref{boundedHtsquared} instead of \eqref{boundedGtsquared}, gives
\[
    V_{AA}\lesssim \cgg{T}.
\]
Finally, by conditional Cauchy--Schwarz,
\begin{align*}
&\left|
\E_{\pi_c}\left[
\left(\Gamma_t^Y-\alpha_t^Y-a\Delta\mu_Y(c)\right)
\left(\Gamma_t^A-\alpha_t^A-a\Delta\mu_A(c)\right)
\mid \cal{F}_{t,a}
\right]\right|
\\
&\qquad\le
\left(
\E_{\pi_c}\left[
\left(\Gamma_t^Y-\alpha_t^Y-a\Delta\mu_Y(c)\right)^2
\mid \cal{F}_{t,a}
\right]
\right)^{1/2}
\left(
\E_{\pi_c}\left[
\left(\Gamma_t^A-\alpha_t^A-a\Delta\mu_A(c)\right)^2
\mid \cal{F}_{t,a}
\right]
\right)^{1/2}.
\end{align*}
The two factors are bounded by the same expression as in \eqref{eq:centered-second-bound}, up to constants. Therefore the preceding argument also gives
\[
    |V_{YA}|\lesssim \cgg{T}.
\]
Therefore, \eqref{eq:VRDonGamma} follows from the definition of $V_{\,\mathrm{RD}}$.

\medskip
It remains to prove the final claim. Using
\(
Q_{c,\,t}^A(s_t,c,1)-Q_{c,\,t}^A(s_t,c,0)\ge \underline q
\)
for a.e.~$s_t$ and all $t\ge0$, we get
\(
    \Delta\mu_A(c)
    \ge \underline q
\).
Furthermore, if
\[
        \left|Q_{c,\,t}^Y(s_t,c,1)-Q_{c,\,t}^Y(s_t,c,0)\right|
        \le
        L\left\{Q_{c,\,t}^A(s_t,c,1)-Q_{c,\,t}^A(s_t,c,0)\right\},
\]
then multiplying by $f_t(c\mid S_t)$, summing over $t$ with weights $\gamma^t$, and dividing by $F(c)$ gives
\[
    |\Delta\mu_Y(c)|\le L\Delta\mu_A(c).
\]
Thus $|\tauRD|\le L$. Combining this with \eqref{eq:VRDonGamma} gives
\(
    V_{\mathrm{RD}}=O(\cgg{T})
\),
as claimed.
\end{proof}

\begin{lemma}\label{lemma:C2rates}
For $T\in\{0,1,2,\dots\}$ and $\gamma\in[0,1]$, define
\[
\cg{T}:=\sum_{t=0}^T \gamma^t,
\qquad
\cgg{T}:=\frac{\sum_{t=0}^T \gamma^{2t}\cg{T-t}^2}{\cg{T}^2}.
\]
\begin{enumerate}
    \item\label{item1C} For each fixed $T<\infty$,
    \[
    \lim_{\gamma\,\uparrow\, 1} \cgg{T}
    = \frac{1}{(T+1)^2}\sum_{j=1}^{T+1} j^2
    = \frac{(T+2)(2T+3)}{6(T+1)}.
    \]

    \item\label{item2C} If $T\to\infty$ and $T(1-\gamma_{\,T})\to 0$, then
    $
    \cgg{T}\sim \frac{T}{3}$.

    \item\label{item3C} If $T\to\infty$ and $T(1-\gamma_{\,T})\to \lambda\in(0,\infty)$, then
    $\cgg{T}\sim TJ(\lambda)$,
    where
     $$J(\lambda):=\int_0^1 \frac{e^{-2\lambda u}(1-e^{-\lambda(1-u)})^2}{(1-e^{-\lambda})^2}\,du=\left(1-e^{-\lambda}\right)^{-2}\left( \frac{1-e^{-2\lambda}}{2\lambda}-\frac{2e^{-\lambda}(1-e^{-\lambda})}{\lambda}+e^{-2\lambda}\right).$$
    %$J(\lambda):= \left(1-e^{-\lambda}\right)^{-2} \displaystyle\int_0^1 e^{-2\lambda u}\bigl(1-e^{-\lambda(1-u)}\bigr)^2\,du$. 

    \item\label{item4C} If $T\to\infty$ and $T(1-\gamma_{\,T})\to\infty$, then
    \[
    \cgg{T}\sim \left(1-\gamma_{\,T}^2\right)^{-1}.
    \]
\end{enumerate}
The function $J$ is right-continuous at $\lambda=0$ with $J(0)=1/3$, so \cref{item2C} matches \cref{item3C} with $\lambda=0$.
\end{lemma}

\begin{proof}
When $T$ is fixed and $\gamma\uparrow1$, we have
\[
\frac{\cg{T-t}}{\cg{T}}\to \frac{T-t+1}{T+1},
\]
and as a result,
\[
\cgg{T}\to \sum_{t=0}^T \left(\frac{T-t+1}{T+1}\right)^2
=\frac{1}{(T+1)^2}\sum_{j=1}^{T+1} j^2
=\frac{(T+2)(2T+3)}{6(T+1)}.
\]
This proves \cref{item1C}.

Now let $T\to\infty$ and write $\gamma=\gamma_{\,T}$. First assume that $T(1-\gamma_{\,T})\to\lambda\in[0,\infty)$. In the case $\lambda=0$, if $\gamma_{\,T}= 1$ along a subsequence, the result subsequential limit of $C_2$ follows from \cref{item1C}. So in the following we assume $\gamma_{\,T}<1$ without loss of generality. For $\gamma<1$, we can write
\[
\cg{T-t}=\frac{1-\gamma^{T-t+1}}{1-\gamma},
\qquad
\cg{T}=\frac{1-\gamma^{T+1}}{1-\gamma},
\]
and therefore
\[
\cgg{T}
=\sum_{t=0}^T\left(\frac{\gamma^t(1-\gamma^{T-t+1})}{1-\gamma^{T+1}}\right)^2.
\]
Expanding the square and summing the geometric series gives the exact identity
\begin{equation}
\cgg{T}
=\frac{1}{(1-\gamma^{T+1})^2}\left(\displaystyle \frac{1-\gamma^{2T+2}}{1-\gamma^2}-2\gamma^{T+1}\frac{1-\gamma^{T+1}}{1-\gamma}+(T+1)\gamma^{2T+2}\right).
\tag{1}
\end{equation}
Since $T\to\infty$ and $T(1-\gamma_{\,T})\to \lambda\in[0,\infty)$, we necessarily have $\gamma_{\,T}\to1$, and hence we can write the following:
\[
\log \gamma_{\,T} = -(1-\gamma_{\,T})+O\bigl((1-\gamma_{\,T})^2\bigr),
\]
and
\[
T\log\gamma_{\,T} = -T(1-\gamma_{\,T})+o(1)\to -\lambda.
\]
Therefore, uniformly in $0\le t\le T$,
\[
\gamma_{\,T}^t = \exp\!\left(\frac{t}{T}\,T\log\gamma_{\,T}\right) \to e^{-\lambda t/T},
\qquad
\gamma_{\,T}^{T-t+1}\to e^{-\lambda(1-t/T)},
\qquad
\gamma_{\,T}^{T+1}\to e^{-\lambda}.
\]
Hence
\[
\frac{\cgg{T}}{T}
=\frac1T\sum_{t=0}^T\left(\frac{\gamma_{\,T}^t(1-\gamma_{\,T}^{T-t+1})}{1-\gamma_{\,T}^{T+1}}\right)^2
\to \int_0^1 \frac{e^{-2\lambda u}(1-e^{-\lambda(1-u)})^2}{(1-e^{-\lambda})^2}\,du
=J(\lambda).
\]
This proves \cref{item3C} when $\lambda\in(0,\infty)$, and the same formula with $\lambda=0$ gives \cref{item2C} once we check $J(0)=1/3$. For that, simply note that the integrand converges pointwise to $(1-u)^2$ and is uniformly bounded on $[0,1]\times[0,1]$, so
\(
J(0)=\int_0^1 (1-u)^2\,du=\frac13
\).
It is also straightforward to evaluate the integral in the definition of $J$ and conclude that, for any  $\lambda>0$,
\[
J(\lambda)
=\left(1-e^{-\lambda}\right)^{-2}\left( \frac{1-e^{-2\lambda}}{2\lambda}-\frac{2e^{-\lambda}(1-e^{-\lambda})}{\lambda}+e^{-2\lambda}\right).
\]
This completes the proof of \cref{item2C,item3C}.

It remains to treat the case $T(1-\gamma_{\,T})\to\infty$. Set $a_T:=\gamma_{\,T}^{T+1}$. Since $\log x\le -(1-x)$ for $x\in[0,1)$,
\[
a_T=\exp\bigl((T+1)\log\gamma_{\,T}\bigr)
\le \exp\bigl(-(T+1)(1-\gamma_{\,T})\bigr)\to0.
\]
Using \((1)\),
\[
\cgg{T}
=\frac{1+a_T}{1-a_T}\cdot\frac{1}{1-\gamma_{\,T}^2}
-\frac{2a_T}{1-a_T}\cdot\frac{1}{1-\gamma_{\,T}}
+\frac{(T+1)a_T^2}{(1-a_T)^2}.
\]
The first term is asymptotic to $(1-\gamma_{\,T}^2)^{-1}$ because $a_T\to0$. For the second term,
\[
(1-\gamma_{\,T}^2)\displaystyle \frac{2a_T}{1-a_T}\frac{1}{1-\gamma_{\,T}}
=\frac{2a_T}{1-a_T}(1+\gamma_{\,T})\to0.
\]
For the third term,
\[
(1-\gamma_{\,T}^2)\displaystyle \frac{(T+1)a_T^2}{(1-a_T)^2}
=\frac{(T+1)(1-\gamma_{\,T}^2)a_T^2}{(1-a_T)^2}
\le \frac{2(T+1)(1-\gamma_{\,T})e^{-2(T+1)(1-\gamma_{\,T})}}{(1-a_T)^2}\to0,
\]
because $x e^{-2x}\to0$ as $x\to\infty$. Hence
\(
\cgg{T}=(1-\gamma_{\,T}^2)^{-1}(1+o(1))
\),
as desired to show.
\end{proof}

\begin{lemma}\label{lemma:multiv-delta}
    Suppose that $h_n=O(n^{-1/5})$, and that $X_n$ and $Y_n$ are real-valued random variables such that
    $$X_n=x+b_x h_n^2 + \frac{U_n}{\sqrt{nh_n}}+r_{x,n}, \quad Y_n=y+b_y h_n^2 + \frac{V_n}{\sqrt{nh_n}}+r_{y,n},$$
    where $(U_n, V_n)$ converges weakly to a zero-mean bivariate Gaussian $(U, V)$, $r_{x,n}=\o(\sqrt{nh_n})$ and $r_{y,n}=\o(\sqrt{nh_n})$, and $y\neq 0$. Then,
    $$\sqrt{nh_n}\left(\frac{X_n}{Y_n}-\frac{x}{y}-{h_n^2}\left(\frac{b_x}{y}-\frac{xb_y}{y^2}\right)\right)\dto\normal\left(0\,,\,\frac{1}{y^2}\left(\var(U)+\frac{x^2}{y^2}\var(V)-\frac{2x}{y}\cov(U, V)\right)\right) .$$
\end{lemma}

\begin{proof}
    The proof is a straightforward application of the identity $\wh{x}\wh{y}-xy = y(\wh{x}-x)+x(\wh{y}-y)-(\wh{x}-x)(\wh{y}-y)$. Using this identity we can write
    \begin{align*}
        \frac{X_n}{Y_n} - \frac{x}{y} &= \frac{1}{y}(X_n - x) + x\left(\frac{1}{Y_n} - \frac{1}{y}\right) + \underbrace{(X_n-x)\left(\frac{1}{Y_n}-\frac{1}{y}\right)}_{=\,\o(\sqrt{nh_n})}\\
         &= \frac{1}{y}\left(b_x h_n^2 + \frac{U_n}{\sqrt{nh_n}}\right) - \frac{x}{y^2}\left(b_y h_n^2+\frac{V_n}{\sqrt{nh_n}}\right) + \o(\sqrt{nh_n})\\[2mm]
        & = h_n^2\left(\frac{b_x}{y}-\frac{xb_y}{y^2}\right) + \frac{1}{\sqrt{nh_n}}\left(\frac{1}{y}U_n-\frac{x}{y^2}V_n\right) + \o(\sqrt{nh_n}).
    \end{align*}
    The above combined with Slutsky's lemma and multivariate delta method completes the proof.
\end{proof}

\begin{lemma}[Law of large numbers for a triangular array of matrices]\label{triangular-array-lln}
    Suppose that for each $n\ge 1$, $X_{n,\,1},\dots,X_{n,\,n}$ are i.i.d.~$\R^{p\times q}$-valued random matrices, where $p,q\ge 1$ are fixed. Assume that $\sup_{n\ge1}\E\left\|X_{n,\,1}\right\|<\infty$ and  $\E\left[\|X_{n,\,1}\|\ind{\|X_{n,\,1}\|>\eps n}\right]\to 0$ as $n\to\infty$, for every $\eps>0$, where $\|\cdot\|$ is any norm that satisfies $|A_{i,\,j}|\le \|A\|$ for all $i\le p$ and $j\le q$. Then, $$\frac{1}{n}\sum_{i=1}^n X_{n,i}-\E\left[X_{n,\,1}\right]\ \Pto\ 0.$$ 
\end{lemma}

\begin{proof}
    This result essentially follows from \citet[Theorem 2.2.11]{Durrett2019} and the accompanying discussion. We present here the proof for $p=q=1$; the extension to arbitrary $p, q \ge 1$ follows by applying the same reasoning componentwise, noting that the required assumptions are satisfied for each entry $(i,j)$ with $i\le p$ and $j\le q$, because $\left|X_{n,\,1,\,i,\,j}\right|\le \|X_{n,\,1}\|$ and $ \left|X_{n,\,1,\,i,\,j}\right|\ind{\left|X_{n,\,1,\,i,\,j}\right|>\eps n}\le \left\|X_{n,\,1}\right\|\ind{\|X_{n,\,1}\|>\eps n}$. 
    
    In view of \citet[Theorem 2.2.11]{Durrett2019}, it suffices to show that (a) $n\,\P(\left|X_{n,\,1}\right|>n)\to 0$, and (b) $n^{-1}\,\E[X_{n,\,1}^2\ind{\left|X_{n,\,1}\right|\le n}]\to 0$. The first condition follows from the following argument: $$n\,\P(\left|X_{n,\,1}\right|>n)=n\,\E\left[\ind{\left|X_{n,\,1}\right|>n}\right]\le \E\left[\left|X_{n,\,1}\right|\ind{\left|X_{n,\,1}\right|>n}\right]\to 0.$$ To prove (b), note that
    $$\frac{1}{n}\E\left[X_{n,\,1}^2\ind{\left|X_{n,\,1}\right|\le n}\right]=\frac{1}{n}\int_0^n 2y\,\P(\left|X_{n,\,1}\right|>y)\,dy=\int_0^1 2\eps n \,\P(\left|X_{n,\,1}\right|>\eps n)\,d\eps.$$
    Define $g_n(\eps) := 2\eps n\, \P(\left|X_{n,\,1}\right|>\eps n)$. Note that $g_n(\eps)\to 0$ pointwise as $n\to\infty$, because 
    $$g_n(\eps)=2\eps n\,\E\left[\ind{\left|X_{n,\,1}\right|>\eps  n}\right]\le 2\,\E\left[\left|X_{n,\,1}\right|\ind{\left|X_{n,\,1}\right|>\eps n}\right]\to 0.$$
    Furthermore, $g_n(\eps) \le 2\sup_{n\ge 1}\E\left|X_{n,\,1}\right|<\infty$ for every  $\eps\in (0,1)$.  Therefore, we can apply the dominated convergence theorem to conclude that $\lim_{n\to\infty}\int_0^1 g_n(\eps)\,d\eps = 0$. This completes the proof, thanks to \citet[Theorem 2.2.11]{Durrett2019}.
\end{proof}

\end{document}